\documentclass[reqno,a4paper,11pt]{article}
\pdfoutput=1
\usepackage{xcolor}
\usepackage{graphicx}
\usepackage[textwidth = 430 pt, textheight = 630 pt]{geometry}

\definecolor{MyDarkBlue}{rgb}{0.15,0.25,0.45}
\usepackage{epsfig,rotating}
\usepackage{amsmath,amssymb}
\usepackage{amsfonts}
\usepackage{mathrsfs}
\usepackage{bbm}
\usepackage[normalem]{ulem}
\usepackage{enumitem}

\usepackage{latexsym}
\usepackage{amsthm}

\usepackage{tensor}

\usepackage[numbers,sort&compress]{natbib}

\usepackage[utf8x]{inputenc}

\usepackage{hyperref}
\hypersetup{
hypertexnames=false,
colorlinks=false,
citecolor=MyDarkBlue,
linkcolor=MyDarkBlue,
urlcolor=MyDarkBlue,
pdfauthor={Severin Bunk, Christian Sämann and Richard J. Szabo},
pdftitle={The 2-Hilbert Space of a Prequantum Bundle Gerbe},
pdfsubject={hep-th math-ph}
breaklinks=true
}

\usepackage{tikz}
\usepackage{mathtools,tensor}
\usepackage[all,knot,2cell]{xy}
\UseTwocells
\xyoption{arc}

\allowdisplaybreaks

\newcommand{\eor}{\hfill\ensuremath{\lhd}}%

\newtheoremstyle{thm} 								
{0.2cm}   					
{0.2cm}   	 				
{\itshape} 	
{}         				
{\bfseries}				
{}        				
{0.2cm} 					
{}         				

\newtheoremstyle{rmk} 								
{0.2cm}   					
{0.2cm}   	 				
{} 	
{}         				
{\bfseries}				
{}        				
{0.2cm} 					
{}         				

\theoremstyle{thm}
\newtheorem{theorem}{Theorem}
\newtheorem{lemma}[theorem]{Lemma}
\newtheorem{definition}[theorem]{Definition}
\newtheorem{proposition}[theorem]{Proposition}
\newtheorem{corollary}[theorem]{Corollary}

\theoremstyle{rmk}
\newtheorem{rmk_aux}[theorem]{Remark}
\newtheorem{ex_aux}[theorem]{Example}

\numberwithin{equation}{section}
\numberwithin{theorem}{section}

\newenvironment{remark}{\begin{rmk_aux}}{\eor \end{rmk_aux}}
\newenvironment{example}{\begin{ex_aux}}{\eor \end{ex_aux}}

\newcommand{\appendices}{
\section*{Appendix}\label{appendices}\setcounter{subsection}{0}
\addcontentsline{toc}{section}{Appendix}
\setcounter{equation}{0}
\setcounter{theorem}{0}
\makeatletter
\renewcommand{\theequation}{\Alph{subsection}.\arabic{equation}}
\renewcommand{\thesubsection}{\Alph{subsection}}
\renewcommand{\thetheorem}{\Alph{subsection}.\arabic{theorem}}
\makeatother
}

\def\slasha#1{\setbox0=\hbox{$#1$}#1\hskip-\wd0\hbox to\wd0{\hss\sl/\/\hss}}

\def\periodb#1{\setbox0=\hbox{$#1$}#1\hskip-\wd0\hbox to\wd0{-}}

\newcommand{\lsb}{\big(\hspace{-0.1cm}\big(}
\newcommand{\rsb}{\big)\hspace{-0.1cm}\big)}

\newcommand{\unit}{\mathbbm{1}}   			
\newcommand{\im}{\mathrm{im}}   			
\newcommand{\id}{\mathrm{id}}   			

\newcommand{\CA}{\mathcal{A}}    			

\newcommand{\xd}{\dot{x}}

\newcommand{\CB}{\mathcal{B}}
\newcommand{\CCB}{\mathscr{B}}
\newcommand{\CC}{\mathcal{C}}
\newcommand{\CCC}{\mathscr{C}}

\newcommand{\CD}{\mathcal{D}}
\newcommand{\CCD}{\mathscr{D}}

\newcommand{\CCF}{\mathscr{F}}
\newcommand{\CG}{\mathcal{G}}
\newcommand{\CCG}{\mathscr{G}}
\newcommand{\CH}{\mathcal{H}}
\newcommand{\CCH}{\mathscr{H}}
\newcommand{\CI}{\mathcal{I}}

\newcommand{\CK}{\mathcal{K}}

\newcommand{\CL}{\mathcal{L}}

\newcommand{\CN}{\mathcal{N}}
\newcommand{\CO}{\mathcal{O}}

\newcommand{\CP}{\mathcal{P}}

\newcommand{\CCP}{\mathscr{P}}
\newcommand{\CQ}{\mathcal{Q}}

\newcommand{\CCR}{\mathscr{R}}

\newcommand{\CS}{\mathcal{S}}
\newcommand{\CT}{\mathcal{T}}

\newcommand{\CU}{\mathcal{U}}
\newcommand{\CV}{\mathcal{V}}
\newcommand{\CCV}{\mathscr{V}}
\newcommand{\CW}{\mathcal{W}}

\newcommand{\CE}{\mathcal{E}}
\newcommand{\fra}{\mathfrak{a}}				
\newcommand{\frh}{\mathfrak{h}}				

\newcommand{\frU}{\mathfrak{U}}
\newcommand{\fru}{\mathfrak{u}}
\newcommand{\frX}{\mathfrak{X}}

\newcommand{\FR}{\mathbbm{R}}     			
\newcommand{\FC}{\mathbbm{C}}     			
\newcommand{\NN}{\mathbbm{N}}     			
\newcommand{\MM}{\mathbbm{M}}     			
\newcommand{\RZ}{\mathbbm{Z}}     			
\newcommand{\CPP}{{\mathbbm{C}P}}    			

\newcommand{\dd}{\mathrm{d}}     			
\newcommand{\dpar}{\partial}     			
\newcommand{\embd}{{\hookrightarrow}}     		
\newcommand{\de}{\mathrm{e}}     			
\newcommand{\di}{\mathrm{i}}     			
\newcommand{\eps}{{\varepsilon}}			

\newcommand{\sB}{\mathsf{B}}

\newcommand{\eand}{{\qquad\mbox{and}\qquad}}     		
\newcommand{\ewith}{{\qquad\mbox{with}\qquad}}
\newcommand{\efor}{{\qquad\mbox{for}\qquad}}
\newcommand{\eon}{{~~\mbox{on}~~}}

\newcommand{\der}[1]{\frac{\dpar}{\dpar #1}}   		
\newcommand{\tr}{\mathsf{tr}}     			
\newcommand{\pr}{\mathrm{pr}}     			
\newcommand{\curv}{\mathrm{curv}}     			

\newcommand{\au}{\mathfrak{u}}

\newcommand{\aso}{\mathfrak{so}}

\newcommand{\sU}{\mathsf{U}}     			

\newcommand{\sA}{\mathsf{A}}
\newcommand{\sG}{\mathsf{G}}

\newcommand{\sR}{\mathsf{R}}

\newcommand{\sHom}{\mathsf{Hom}}

\newcommand{\sLie}{\mathsf{Lie}}

\newcommand{\sH}{\mathsf{H}}
\newcommand{\sSU}{\mathsf{SU}}
\newcommand{\sPU}{\mathsf{PU}}

\newcommand{\sGL}{\mathsf{GL}}

\newcommand{\sSO}{\mathsf{SO}}
\newcommand{\sSpin}{\mathsf{Spin}}
\newcommand{\sString}{\mathsf{String}}
\newcommand{\sEnd}{\mathsf{End}}

\def\tyng(#1){\hbox{\tiny$\yng(#1)$}}			
\def\tyoung(#1){\hbox{\tiny$\young(#1)$}}			

\newcommand{\beq}{\begin{eqnarray}}
\newcommand{\eeq}{\end{eqnarray}}

\newcommand{\htimes}{\bullet}

\newcommand{\sft}{{\sf t}}

\newcommand{\sfd}{{\sf d}}

\newcommand{\sff}{{\sf f}}

\newcommand{\sfsw}{{\sf sw}}
\newcommand{\myxymatrix}[1]{\vcenter{\vbox{\xymatrix{#1}}}}

\newcommand{\sfa}{\mathsf{a}}

\newcommand{\sfr}{\mathsf{r}}
\newcommand{\sfl}{\mathsf{l}}

\newcommand{\sfs}{\mathsf{s}}

\newcommand{\sfR}{\mathsf{R}}
\newcommand{\sfS}{\mathsf{S}}

\renewcommand{\uline}[1]{\underline{\it #1}}

\newcommand{\CatCat}{\mathsf{Cat}}
\newcommand{\CatSet}{\mathsf{Set}}

\newcommand{\CatMfd}{\mathsf{Mfd}}
\newcommand{\CatMfdCat}{\mathsf{MfdCat}}
\newcommand{\CatLieGrpd}{\mathsf{LieGrpd}}
\newcommand{\CatLieGrp}{\mathsf{LieGrp}}

\newcommand{\CatBibun}{\mathsf{Bibun}}
\newcommand{\CatBimod}{\mathsf{Bimod}}
\newcommand{\CatMod}{\mathsf{Mod}}
\newcommand{\CatVect}{\mathsf{Vect}}
\newcommand{\CatHilb}{\mathsf{Hilb}}
\newcommand{\shom}{\mathsf{hom}}

\newcommand{\sisom}{\mathsf{isom}}

\newcommand{\CatLBGrb}{\mathsf{LBGrb}}

\newcommand{\CatDes}{\mathsf{Des}}
\newcommand{\CatVBdl}{\mathsf{VBdl}}
\newcommand{\CatHVBdl}{\mathsf{HVBdl}}

\newcommand{\CatHLBdl}{\mathsf{HLBdl}}

\newcommand{\doublearrow}{\mathrel{\substack{\longrightarrow\\[-0.6ex]
                      \longrightarrow}}}
\newcommand{\triplearrow}{\mathrel{\substack{\longrightarrow\\[-0.6ex]
                      \longrightarrow \\[-0.6ex]
                      \longrightarrow}}}
\newcommand{\quadarrow}{\mathrel{\substack{\longrightarrow\\[-0.6ex]
                      \longrightarrow \\[-0.6ex]
                      \longrightarrow\\[-0.6ex]
                      \longrightarrow}}}

\newcommand{\hol}{{\sf hol}}

\newcommand{\thra}{\twoheadrightarrow}
\newcommand{\xthra}[2][]{%
  \xrightarrow[#1]{#2}\mathrel{\mkern-14mu}\rightarrow
}

\newcommand{\fpmap}[4]{\overset{#1}{#2}{}^{#3}_{#4}}

\newcommand{\rk}{\,\mathrm{rk}}

\newcommand{\isom}{\overset{\cong}{\longrightarrow}}
\newcommand{\twoisom}{\overset{\cong}{\Longrightarrow}}

\newcommand{\ev}{\mathrm{ev}}
\newcommand{\bbG}{\mathbbm{G}}
\newcommand{\nt}{\notag\\}

\newcommand{\coker}{\mathrm{coker}}
\newcommand{\Gr}{\mathsf{Gr}}

\title{}

\begin{document}

\begin{titlepage}
\begin{flushright}
 EMPG--16--16
\end{flushright}
\vskip 2.0cm
\begin{center}
{\LARGE \bf The 2-Hilbert Space \\[0.3cm] of a Prequantum Bundle Gerbe}
\vskip 1.5cm
{\Large Severin Bunk, Christian S\"amann and Richard J.\ Szabo}
\setcounter{footnote}{0}
\renewcommand{\thefootnote}{\arabic{thefootnote}}
\vskip 1cm
{\em 
Department of Mathematics,
Heriot-Watt University\\
Colin Maclaurin Building, Riccarton, Edinburgh EH14 4AS, U.K.}\\
and\\  {\em Maxwell Institute for Mathematical Sciences, Edinburgh,
  U.K.} \\ and \\ {\em The Higgs Centre for Theoretical Physics,
  Edinburgh, U.K.}
\\[0.5cm]
{Email: {\ttfamily sb11@hw.ac.uk , C.Saemann@hw.ac.uk , R.J.Szabo@hw.ac.uk}}
\end{center}
\vskip 2.0cm
\begin{center}
{\bf Abstract}
\end{center}
\begin{quote}
We construct a prequantum 2-Hilbert space for any line bundle gerbe whose
Dix\-mier-Douady class is torsion. Analogously to usual
prequantisation, this 2-Hilbert space has the category of sections of
the line bundle gerbe as its underlying 2-vector
space. These sections are obtained as certain morphism categories in
Waldorf's version of the 2-category of line bundle gerbes. We show
that these morphism categories carry a monoidal structure under which
they are semisimple and abelian. We introduce a dual functor on the
sections, which yields a closed structure on the morphisms between
bundle gerbes and turns the category of sections into a 2-Hilbert
space. We discuss how these 2-Hilbert spaces fit various
expectations from higher prequantisation. We then extend the
transgression functor to the full 2-category of bundle gerbes and demonstrate
its compatibility with the additional structures 
introduced. We discuss various aspects of Kostant-Souriau
prequantisation in this setting, including its dimensional reduction to ordinary prequantisation.
\end{quote}
\end{titlepage}

\tableofcontents

\section{Introduction and summary}

In this paper we develop a construction of a 2-Hilbert space from a line bundle gerbe. This is a categorification of the first part of geometric quantisation, where a Hilbert space is constructed from sections of a prequantum line bundle over a symplectic manifold. Recall that the first Chern class of a prequantum line bundle is (an integer multiple of) the cohomology class of the symplectic form of the base manifold. 

\subsection{Motivation and context\label{sec:motivation}}

Our motivation for studying this construction is essentially twofold. Firstly, categorification of a mathematical construction or object usually leads to a deeper understanding, akin to deformations and other generalisations. Secondly, within string theory the need for higher geometric quantisation arises in various contexts, and the construction of a 2-Hilbert space is a first step towards such a procedure. Let us explain both points in more detail. 

Categorification \cite{Baez:9802029} is in general not a unique process. It is meant to provide a recipe for lifting notions based on sets (or 0-categories) to related notions using categories (or, rather generally, notions based on $n$-categories to related notions involving $n+1$-categories). There are essentially two main variants of what is more precisely called vertical categorification which have appeared in the string theory literature. The first variant amounts to constructing sections of decategorification maps; that is, one tries to find for a given set $S$ a category $\CCC$ together with a `nice' map $p$ from the isomorphism classes $h_0 \CCC$ to $S$. The standard examples are 1) $S=\NN_0$ and $\CCC$ the category of finite sets with $p$ the cardinality and 2) $S=\NN_0$ and $\CCC$ the category of finite-dimensional vector spaces with $p$ the dimension.

A second variant of vertical categorification is that of (weak) internalisation in another category or even in a higher category. Here one separates a mathematical notion into sets, structure maps and structure relations. The sets become objects in the category, the structure maps become morphisms and the structure relations become commutative diagrams. In this way a group internal to the category of smooth manifolds is simply a Lie group. On the other hand, a group internal to $\CatCat$, the category of small categories, is a categorified group or 2-group. (In fact, the latter also yields a section of a decategorification map.)

Definitions of mathematical objects that allow for a straightforward categorification are usually particularly useful, since they separate essential from non-essential aspects and allow for a direct identification of trivial relations. Two important examples of such definitions, which we shall encounter in this paper, are the definition of a principal bundle as a functor from the \v Cech groupoid of a cover to the delooping of the structure group and the definition of sections of a line bundle $L$ as morphisms from the trivial line bundle to $L$. Both definitions are readily categorified, see e.g.\ \cite{Waldorf:2007aa,Schweigert:2014nia}, which allows us to define and study sections of higher line bundles. The latter will form the 2-vector spaces underlying the 2-Hilbert spaces we are after. In the process of developing various aspects of the first steps in higher geometric quantisation, our awareness of the various mathematical structures involved in the corresponding steps of ordinary geometric quantisation increased significantly.

However, there is another more direct reason for constructing 2-Hilbert spaces, and that is their appearance in various situations in string theory that are currently of interest. Rather generally, M-theory requires a lift of the noncommutative deformations arising in string theory to higher noncommutative geometry, involving nonassociative structures. For example, D-branes ending on other D-branes usually form partially noncommutative manifolds. In particular, open D$p$-branes ending transversally on D$(p+2)$-branes will form what is known as a ``fuzzy funnel''. Here each point in the worldvolume of the D$p$-brane polarises into a fuzzy or quantised 2-sphere. The lift to M-theory of the D$2$--D$4$-brane system yields an open M2-brane ending on an M5-brane with a quantised 3-sphere expected to arise. Just as the 2-sphere comes with a symplectic volume form, defining a prequantum line bundle whose sections form the prequantum Hilbert space, the 3-sphere comes with a multisymplectic volume form, defining a prequantum gerbe whose sections are supposed to form a prequantum 2-Hilbert space.   A more detailed explanation of this point is found e.g.\ in \cite{Saemann:2012ex}.

Higher quantisation becomes equally prominent in flux compactifications of string theory and in double field theory \cite{Hull:2009mi}. In particular, the NS--NS $H$-flux defines a multisymplectic 3-form which is integral by (generalised) Dirac quantisation, hence it defines a prequantum gerbe on the string background and again a prequantum 2-Hilbert space. Applying successive T-duality transformations to such compactifications inevitably brings one into a `non-geometric' regime of string theory, where the target space can no longer be described as a classical Riemannian manifold. Non-geometric string theory requires in
   particular the quantisation of twisted Poisson structures involving
   nonassociative observable algebras. For example, on the simplest 3-dimensional multisymplectic manifolds $\FR^3$, $T^3$ and $S^3$, closed strings which wind and propagate in the ensuing non-geometric backgrounds probe a noncommutative and even nonassociative deformation of the original geometry~\cite{Blumenhagen:2010hj,Lust:2010iy,Blumenhagen:2011ph}, which is closely related to the appearance of categorified and other higher geometric structures~\cite{Mylonas:2012pg}; in these settings the noncommutative and nonassociative geometry is most efficiently described by structures internal to a certain representation category~\cite{Barnes:2014ksa,Barnes:2015uxa}.
The nonassociative observable algebras arising in non-geometric string theory can typically be embedded as subspaces of (associative) Poisson algebras of differential operators~\cite{Mylonas:2013jha}; an approach to geometric quantisation of twisted Poisson
   manifolds is pursued along these lines in~\cite{Petalidou:2007tx}.

We shall concisely summarise the steps involved in geometrically quantising a symplectic manifold in Section~\ref{ssec:gq_outline}. The purpose of this paper is then to identify categorified versions of the ingredients in each step to the extent possible. We shall start from a multisymplectic $3$-form $\omega$ on some manifold $M$. If $\omega$ is integral (this is the usual quantisation condition), then it represents the characteristic class of a $\sU(1)$-bundle gerbe. We can now define sections of its associated line bundle gerbes and we find that they form a structure that matches a categorified notion of a Hilbert space. 

The result is an analogue of the prequantum Hilbert space of ordinary geometric quantisation. There one would use a polarisation on the manifold to halve the state space to the actual Hilbert space. A categorified notion of polarisation remains an important open problem in higher quantisation, and unfortunately we do not gain any new insights into this issue from our perspective.

In some simple cases, however, we can understand e.g.\ the explicit action of categorified versions of the isometry group of the manifold on the 2-Hilbert space, which is a subgroup of the categorified multisymplectomorphisms. We also find that an expected higher Lie algebra of quantum observables acts as categorified endomorphisms on this 2-Hilbert space.

The configuration space for a closed string moving through a target space $M$ is the loop space\footnote{Or, more precisely, its unparameterised version, the knot space $\CK M$.} $\CL M$. Therefore it is only natural to ask about quantisations of this mapping space. Interestingly, the categorified line bundles we encountered above can be transported to loop space via the transgression map, where they turn into ordinary line bundles, cf.~\cite{0817647309}. For a number of reasons, it is interesting to study the translation of the various 2-categorical structures we found before to the loop space picture. As an example, the transgressed line bundle on the loop spaces of compact semi-simple Lie groups has close relations with the representation theory of the central extensions of the loop group, cf.~\cite{Pressley:1988aa,Sergeev:2008aa}. A detailed explanation of how  quantisation of loop spaces nicely matches the quantum geometry of open M2-branes ending on M5-branes in a constant $C$-field background as well as the noncommutative geometry of closed strings in non-geometric T-duality frames can be found in~\cite{Saemann:2012ab}. 

Further motivation for studying higher prequantum geometry and a summary of the structures involved can be found in~\cite{Schreiber:2016pxa} as well as in the much more detailed exposition~\cite{Schreiber:2013pra}. The papers~\cite{Rogers:2010sc,Fiorenza:2013kqa,Fiorenza:1304.6292} discuss higher prequantisation from a perspective omitting an explicit notion of higher Hilbert space.

\subsection{Summary of results}

Throughout this paper we have attempted to include reviews of relevant mathematical notions and their interrelationships to the extent that we would have found useful at the start of this project, and we hope that other readers will benefit from our exposition.

In particular, we review the theory of Lie groupoids and their various morphisms, which is useful for regarding gerbes as central groupoid extensions. We also introduce 2-vector spaces as module categories over rig categories, cf.\ \cite{Schreiber:2008aa,Schreiber:2009:357-401}. This allows for a unified perspective on various 2-vector spaces in the literature, such as those of Baez-Crans \cite{Baez:2003aa} and Kapranov-Voevodsky \cite{kapranov19942}.

Gerbes are presented from three different perspectives: first, as abelian principal 2-bundles given in terms of \v Cech cocycles; second, as Murray's bundle gerbes and third, as central groupoid extensions. The associated line bundles are introduced and they are the objects of Waldorf's convenient 2-category $\CatLBGrb^\nabla(M)$ \cite{Waldorf:2007aa}, which is also explained in detail.

We extend the morphism category in $\CatLBGrb^\nabla(M)$ by a direct sum functor, which eventually turns $\CatLBGrb^\nabla(M)$ into a closed symmetric monoidal 2-category whose morphism categories are semisimple abelian and cartesian symmetric monoidal categories by Theorem~\ref{thm:4.28}. On the morphism category of $\CatLBGrb^\nabla(M)$, we then define an inner product structure. In particular, we can restrict this inner product to sections of line bundle gerbes, which are full subcategories of $\sHom_{\CatLBGrb^\nabla(M)}$, and we obtain a notion of line bundle gerbe metric in Definition~\ref{def:2-bdl_metric}.

In Definition~\ref{def:2-Hilbert_space} we give a suitable notion of 2-Hilbert space, of which the  morphism subcategories corresponding to sections of line bundle gerbes form special instances, see Theorem~\ref{thm:5.9}. Our definition is similar in spirit but not equivalent to the 2-Hilbert spaces introduced in \cite{kapranov19942} and \cite{Baez:9609018}. We also find that the Lie 2-algebra of classical observables can be embedded in the category of endomorphisms of sections in Proposition~\ref{prop:5.11}.

We discuss the example of 2-Hilbert spaces arising from prequantum gerbes over $\FR^3$ in much detail. In particular, we find that the string 2-group $\sString(3)$ acts naturally on these 2-Hilbert spaces, cf.\ Theorem~\ref{thm:6.5}.

In our last section, we turn to the transgression functor, taking the prequantum gerbe to a line bundle over loop space. We give a definition, which by Theorem~\ref{st:Transgression_is_functorial} is a functor between $h_1\CatLBGrb^\nabla(M)$, the 2-isomorphism classes in the 2-category of line bundle gerbes, and the category of hermitean fusion line bundles over loop space, both regarded as closed and symmetric monoidal additive categories. The proof of this theorem requires a number of technical results, including an explicit construction of the transgressed connection, which are interesting in their own right and which can be found in Sections~\ref{sect:Transgression_of_algebraic_data} and~\ref{sect:connection_on_CTCG}.

The Grothendieck completion of $h_1\CatLBGrb^\nabla(M)$ given in Definition~\ref{def:Grothendieck_completion} also allows us to extend the transgression functor to a functor of closed and symmetric monoidal categories enriched over abelian groups, see Theorem~\ref{thm:7.14}. 

Finally, we give the appropriate Kostant-Souriau prequantisation map over loop space in Proposition~\ref{prop:KS_prequantisation_map}, and we briefly comment on its string theory reduction.

\subsection{Overview of geometric prequantisation}\label{ssec:gq_outline}

Let us concisely recall the basics of geometric prequantisation in a language that will allow for a straightforward translation to higher geometric quantisation. The relevant original literature is \cite{Souriau-1970aa,Kostant-1970aa,Berezin:1974du,Kostant:1975qe,Kirillov2001}; see also \cite{IuliuLazaroiu:2008pk} for a review. We shall be referring to the following points throughout the text:
\begin{enumerate}[fullwidth]
 \item\label{it:observables} On a symplectic manifold $(M,\omega)$,
   there is a Hamiltonian vector field $X_f\in\frX(M)$ for each observable $f\in
   C^\infty(M)$ such that $\iota_{X_f} \omega=-\dd f$. This induces a
   Poisson algebra structure $\{f,g\}_\omega:=-\iota_{X_f}\iota_{X_g}\omega$
   on $C^\infty(M)$ which, as a Lie algebra, is a central extension of
   the Lie algebra of Hamiltonian vector fields on $M$: $[X_f,X_g]=X_{\{f,g\}_\omega}$.
 \item\label{it:prequantum_line_bundle} A prequantum line bundle of a
   symplectic manifold $(M,\omega)$ is the associated line bundle $L$
   to a principal $\sU(1)$-bundle endowed with a connection $\nabla$
   of curvature $F_\nabla=-2 \pi \, \di \, \omega$. If such a line
   bundle $L$ exists, that is,\ $[\omega] \in H^2(M,\RZ_M)$, then we call $(M,\omega)$ \emph{quantisable}. Each prequantum line bundle admits a hermitean metric $h$ compatible with the connection. If $M$ is K\"ahler, compatibility with the complex structure determines $h$ up to multiplication by a positive constant.
 \item\label{it:tautological_line_bundle} This prequantum line bundle
   was called the tautological 0-gerbe\footnote{This produces, in
     fact, the dual of the tautological line bundle $\CO(-1)$ over
     $\CPP^n$ for the canonical symplectic form.} for $\omega$ in
   \cite{Johnson:2003aa}. It is readily constructed if we generalise
   from open covers to the surjective submersion known as the path
   fibration. Let $M$ be 1-connected and choose a base point $x_0$ in
   $M$. Denote the space of paths based at this point by $\CP M$; it
   is fibred over $M$ by the endpoint evaluation map
   $\dpar$.\footnote{We will not go into the details of
     infinite-dimensional geometry here, but one could view
     $\CP M$ as a Fr\'echet manifold or even as a diffeological
     space~\cite{Waldorf:0911.3212,Waldorf:2010aa}, and $\dpar$ is
     then a surjective submersion.} The fibred product $\CP M^{[2]}$
   consisting of pairs of paths $(\gamma_1,\gamma_2)$ with the same
   endpoint can now be identified with the based loop space $\Omega
   M$: From each pair of paths $(\gamma_1,\gamma_2)\in \CP M^{[2]}$,
   we can form a based loop $\ell:=\bar \gamma_2\circ
   \gamma_1\in\Omega M$ modulo some technicalities;\footnote{To ensure
     a smooth joining, we should introduce ``sitting instances'' at
     the endpoints; that is, we reparametrise the paths such that they are constant in a neighbourhood of their endpoints.} inversely, given a loop $\ell\in \Omega M$, we can split it into two paths with endpoints $z=\ell(1)$.
 
 We can now define a function $f:\Omega M\rightarrow \sU(1)$ by
\begin{equation}
 f(\ell)=\exp\Big(2\pi\,\di\, \int_\Sigma \, \omega\Big) \ewith \dpar \Sigma=\ell~.
\end{equation}
Since $[\omega]\in H^2(M,\RZ_M)$, this function depends only on $\ell$
and not on the choice of $\Sigma$. The function $f$ is a cocycle, as 
\begin{equation}
\begin{aligned}
 f(\gamma_1,\gamma_2)\,
 f(\gamma_2,\gamma_3)&=\exp\Big( 2\pi\, \di\, \int_{\Sigma_{12}} \,
 \omega+2\pi\, \di\, \int_{\Sigma_{23}}\,
 \omega\Big) \\[4pt] 
 &= \exp\Big(2\pi\, \di\, \int_{\Sigma_{12}\cup\Sigma_{23}}\, \omega\Big) \\[4pt]
 &= f(\gamma_1,\gamma_3)~,
\end{aligned}
\end{equation}
where we decomposed each loop $\ell=\bar\gamma_2\circ\gamma_1$ into its two based paths $(\gamma_1,\gamma_2)$. Thus $f$ defines a principal $\sU(1)$-bundle $P_\omega$ over $M$, and we can define a connection $\nabla_\omega$ on $P_\omega$ with gauge potential
\begin{equation}
 A_\omega=2\pi\,\di\, \int_I\, \ev^* \omega~,
\end{equation}
where $\ev:\CP M\times I\rightarrow M$ denotes the evaluation. This connection satisfies $\dd A_\omega=\dpar^* F_{\nabla_\omega}$, and the line bundle $L_\omega$ associated to $(P_\omega,\nabla_\omega)$ is the prequantum line bundle with first Chern class $[\omega]$.

\item\label{it:vector_space} The prequantum Hilbert space 
  of a quantisable manifold $(M,\omega)$ is now constructed from the vector space $\hat
  \CH=\Gamma(M,L)$ of sections of $L$; we can identify
  $\hat\CH=\shom_{\CatVBdl(M)}(I_0,L)$ where $I_0$ is the trivial line
  bundle over $M$. Note that $\hat \CH$ is a left
  module over the ring of smooth functions $C^\infty(M) =\shom_{\CatVBdl(M)}(I_0,I_0)$, and since
  there is a trivial embedding of $\FC$ as constant functions into $C^\infty(M)$, it is a complex vector space.

\item\label{it:inner_product} If $M$ is compact of dimension $2n$ and $L$ is hermitean, the vector space $\hat\CH$ is additionally a pre-Hilbert space if endowed with the scalar product
\begin{equation}
\langle \eps_1,\eps_2 \rangle = \int_M\, \dd\mu_M~h(\eps_1,\eps_2)~,
\end{equation}
where $\dd\mu_M$ is some volume form on $M$. This can be the canonical
Liouville form $\dd\mu_M=\frac{1}{n!}\, \omega^n$, but we can also
introduce an additional measure $\mu_M'=\mu_M \, \rho$ for some smooth
positive function $\rho$ on $M$. Completion with respect to this scalar product then turns $\hat\CH$ into a Hilbert space. If $M$ is non-compact, we restrict
$\hat \CH$ to the subspace of square-integrable smooth sections $\hat \CH_{\rm
  res}$ of the prequantum line bundle $L$; in this case we have to ensure that
$\rho$ decays fast enough to allow for an interesting space $\hat
\CH_{\rm res}$.

\item\label{it:rep_observables} The observables in $C^\infty(M)
  =\shom_{\CatVBdl(M)}(I_0,I_0)$ can be recovered as the real sections of
  the trivial line bundle $I_0$  over $M$ or, equivalently, as real bundle
  endomorphisms of $I_0$.

\item\label{it:symmetry_group_action} If the symplectic form underlying a K\"ahler manifold $(M,\omega)$ has a symmetry group of isotropies, then we find an induced action of this isotropy group on the Hilbert space of sections. This usually extends to the double cover of the isotropy group or its spin group and the Hilbert space decomposes into a direct sum of irreducible representations, see e.g.~\cite{Murray:2006pi} for a large class of examples. All this is readily seen if we replace the prequantum line bundles with appropriate equivariant prequantum line bundles.

\item\label{it:quantisation_map} The space $\Gamma(M,L)$ is endowed with an action of the Lie algebra $(C^\infty(M),\{-,-\}_\omega)$ given by the \emph{Kostant-Souriau prequantisation map}
\begin{equation}
 Q: C^\infty(M) \longrightarrow \sEnd_{\FC}(\Gamma(M,L)) \ , \quad f
 \longmapsto \nabla^L_{X_f} + 2 \pi \,\di\, f \ ,
\end{equation}
and this representation extends to the full prequantum Hilbert space.
\end{enumerate}
As stated before, a full quantisation requires the introduction of the
notion of polarisation (typical examples being a realisation of $M$ as a cotangent bundle $M=T^*N$ or a
choice of K\"ahler structure on $M$). Restricting to polarised sections (defined via half-densities
on $N$ or holomorphic sections of $L\otimes K^{1/2}$ with $K^{1/2}$ the square root of the canonical line bundle of $M$) then reduces
the prequantum Hilbert space to the actual quantum Hilbert space. The
appropriate notion of polarisation for higher prequantisation remains
unclear, see the discussion in \cite{Rogers:2011zc,Saemann:2012ab}, and we have
nothing to say on this problem in the present paper.

The categorification of these statements will be discussed in the ensuing sections: Statements~\eqref{it:observables}, \eqref{it:prequantum_line_bundle} and \eqref{it:tautological_line_bundle} in Section~\ref{ssec:multisymplectic}, statements \eqref{it:vector_space} and \eqref{it:inner_product} in Section~\ref{ssec:2Hilbert_space_line_bgr}. Statement~\eqref{it:rep_observables} will be addressed in Section~\ref{ssec:observables} and statement~\eqref{it:symmetry_group_action} in Section~\ref{ssec:symmetries}. Finally, a higher equivalent version of statement~\eqref{it:quantisation_map} is obtained in Section~\ref{sect:KS_prequantisation} after transgression to loop space.

\subsection{Outline}

Since the discussion will be rather technical at times, let us give a short outline of the contents of this paper.

Section~\ref{sect:Preliminaries} gives a review of the relevant background material on Lie groupoids, \v Cech and Deligne cohomology, Lie 2-groupoids, 2-vector spaces and representations of 2-groups.

Section~\ref{sect:U(1)-bundle-gerbes} starts with Hitchin-Chatterjee gerbes and their functorial definition, reviews Murray's $\sU(1)$-bundle gerbes as well as gerbes in the form of central groupoid extensions, and compares their morphisms. Connective structures are introduced and various examples are reviewed.

Section~\ref{sect:2-cat_of_BGrbs} contains the definition of line bundle gerbes associated to $\sU(1)$-bundle gerbes, as well as a discussion of the morphisms and 2-morphisms underlying Waldorf's 2-cate\-gory of line bundle gerbes. There we also identify monoidal and closed structures on this 2-category, and define observables in multisymplectic geometry as well as the notion of prequantum line bundle gerbe.

Section~\ref{sect:2-Hilbert_spaces_from_Bgerbes} is concerned with constructing the 2-Hilbert spaces. After giving a suitable definition, we identify the corresponding 2-Hilbert space structures on the category of sections of a line bundle gerbe and embed expected observables into 2-endomorphisms on these 2-Hilbert spaces.

Section~\ref{sect:action_of_2-grps_on_their_BGrbs} gives an explicit example of higher prequantisation, looking at the case of $\FR^3$. We present the 2-Hilbert space as well as the explicit action of the string 2-group of isotropies on $\FR^3$ onto this 2-Hilbert space.

Section~\ref{sect:Transgression} begins with a review of Waldorf's version of transgression, which we then extend to a transgression functor living on an appropriate reduced closed and symmetric monoidal additive category of line bundle gerbes.

In the final Section~\ref{sec:outlook}, we present an outlook on the non-torsion case. Three appendices contain some technical results concerning morphisms of line bundle gerbes, the proof of one of our main theorems (Theorem~\ref{st:Direct_sum_as_fctr}) and some detailed background on the transgression of differential forms to mapping spaces.

\subsection{Notation}

In the following we briefly summarise our notation for the reader's convenience. 

Generic categories and groupoids will be labelled by calligraphic
letters $\CCC, \CCG, \CCV,\dots$. We also write $\CCC=(\CCC_1\rightrightarrows \CCC_0)$ to depict the maps
$\sfs,\sft$ from the morphisms $\CCC_1$ to their source and target
objects $\CCC_0$. We write $\CCC^{\rm op}$ for the opposite category with the same objects as $\CCC$ but with the directions of morphisms reversed (equivalently with the source and target maps $\sfs,\sft$ interchanged). We denote the category of morphisms in a 2-category $\CCC$ by $\CCC_1=(\CCC_{1,1}\rightrightarrows \CCC_{1,0})$. By $\shom_{\CCC}$ we mean
(1-)morphisms $(a\to b)\in\CCC_1$ in a 1- or 2-category $\CCC$ and by
$2\shom_{\CCC}$ we denote 2-morphisms $(f\Rightarrow g)\in
\CCC_{1,1}$; isomorphisms in $\CCC$ are denoted $\sisom_{\CCC}$.

We write $\CatCat$ for the category of small categories, $\CatSet$ for the category of sets and $\CatMfd$ for the category of smooth manifolds. By $\CatVect$ we denote the category of finite-dimensional complex vector spaces. The category of finite-rank vector bundles on $M\in\CatMfd$ is $\CatVBdl(M)$, and its extensions to hermitean vector bundles and hermitean vector bundles with connection are denoted by $\CatHVBdl(M)$ and $\CatHVBdl^\nabla(M)$, respectively. 

Further 2-categories we shall introduce comprise the 2-category of
2-vector spaces $2\CatVect_\CCR$ over a rig category $\CCR$, cf.\
Definition~\ref{def:2-vector_space}, and the 2-category of hermitean line bundle gerbes with connective structure over
$M$, $\CatLBGrb^\nabla(M)$, cf.\ Definition~\ref{def:lbgrb_2-category}. 

Given a manifold $M$, $\frX(M)$ is the Lie algebra of vector fields on
$M$. Given a base point on $M$, $\Omega M$ and $\CP M$ denote
its based loop and path spaces, respectively, cf.\
statement~(\ref{it:tautological_line_bundle}) of
Section~\ref{ssec:gq_outline}. The free loop space is denoted by $\CL M:=C^\infty(S^1,M)$.

We use $\CG=(P,\sigma^Y)$ for most $\sU(1)$-bundle gerbes and
$\CL=(L,\sigma^Y)$ for line bundle gerbes, where $\sigma^Y:Y\thra M$
refers to a surjective submersion, and $P$ and $L$ are principal and
line bundles over the fibre product $Y^{[2]}:=Y\times_M Y$, cf.\
Definitions~\ref{def:u(1)-bundle_gerbe}
and~\ref{def:line_bundle_gerbes}. We depict a $\sU(1)$-bundle gerbe by a diagram
\begin{equation}
     \myxymatrix{
	  P \ar@{->}[d] & \\
	  Y^{[2]} \ar@<1.5pt>[r] \ar@<-1.5pt>[r] & Y \ar@{->}[d]^{\sigma^Y} \\
	    & M
  }
  \end{equation}
where the middle row denotes the \v Cech groupoid of $\sigma^Y$, cf.\ Definition~\ref{def:Cechgroupoid}. Morphisms between line bundle gerbes are denoted by $(E,\zeta^Z)$, cf.\ Definition~\ref{def:lbgr_morph}, and $(\phi,\omega^W)$ refers to a 2-morphism, cf.\ Definition~\ref{def:2hom_bgrb}, where $\zeta^Z$ and $\omega^W$ denote surjective submersions with domains $Z$ and $W$. 

The descent bundle of a descent datum $(E,\varphi)$ in the stack $\CatVBdl:\CatMfd^{\rm op}\to\CatCat$ with respect to a surjective submersion $\sigma^Y:Y\thra M$ is written $D_{\sigma^Y}(E,\varphi)\to M$. This extends to a canonical functor $D_{\sigma^Y}$ from the category of descent data $\CatDes(\CatVBdl,\sigma^Y)$ to $\CatVBdl(M)$, which is an equivalence of categories with canonical inverse $D_{\sigma^Y}^{-1}:\CatVBdl(M)\to \CatDes(\CatVBdl,\sigma^Y)$ that sends $E\in\CatVBdl(M)$ to the pair $(\sigma^{Y*}E,\id)$.

Finally, for any morphism $\phi:E\to F$ of vector bundles, by $\phi^{\rm t}:F^*\to E^*$ we mean the
transpose or Banach space adjoint, never the Hilbert space adjoint
which is metric dependent and which we denote by $\phi^*:F\to E$. When $\phi$ is an
isomorphism we also write
$\phi^{-{\rm t}}:=(\phi^{-1})^{\rm t}$.

\section{Preliminaries on higher structures}\label{sect:Preliminaries}

In this section we recall definitions of various higher
structures and their
properties that we will need in this paper, beginning with a quick
review of 2-categories, and emphasising the role played by the \v
Cech complex of a manifold.

\subsection{2-categories}

\begin{definition}
 A \uline{strict 2-category} is a category enriched over $\CatCat$.
\end{definition}
That is, a strict 2-category $\CCC$ consists of a class of objects
$\CCC_0$ together with a category $\CCC_1(a,b)$ of morphisms
$\CCC_{1,1}(a,b)\rightrightarrows\CCC_{1,0}(a,b)$ over pairs $a,b\in
\CCC_0$. The elements of $\CCC_{1,0}$ are also called
\emph{1-morphisms}, while the elements of $\CCC_{1,1}$, the morphisms
between morphisms, are called \emph{2-morphisms}. There are now two
compositions: the \emph{vertical composition} of 2-morphisms, denoted
by $\circ$, and the functorial \emph{horizontal composition} of 1- and
2-morphisms, denoted by $\htimes_\CCC$ or simply $\htimes$. This definition is readily iterated to strict $n$-categories. The canonical example is the strict 2-category $\CatCat$ of categories, functors and natural transformations. 

Besides the strict 2-functors, which are simply functors enriched over $\CatCat$, we will also require the notion of a normalised pseudofunctor.
\begin{definition}
 A \uline{normalised pseudofunctor}, also called a \uline{normalised weak 2-functor}, $\Phi$ between two strict
 2-categories $\CCC$ and $\CCD$ consists of a function
 $\Phi_0:\CCC_0\rightarrow \CCD_0$, a functor
 $\Phi_1^{ab}:\CCC_1(a,b)\rightarrow \CCD_1(\Phi_0(a),\Phi_0(b))$, and
 an invertible 2-morphism
 $\Phi_2^{abc}:\Phi_1^{ab}(x)\htimes_\CCD\Phi_1^{bc}(y)\Rightarrow
 \Phi_1^{ac}(x\htimes_\CCC y)$ for all $a,b,c\in \CCC_{0}$ and $x,y\in
 \CCC_{1,0}$ such that the diagram 
  \begin{equation}
  \xymatrixcolsep{3pc}
  \myxymatrix{
  & \Phi_1^{ac}(x\htimes y)\, \htimes\, \Phi_1^{cd}(z) \ar@{=>}[dr]^{\Phi_2^{acd}} & \\
  (\Phi_1^{ab}(x)\, \htimes\,\Phi_1^{bc}(y))\, \htimes\, \Phi_1^{cd}(z) \ar@{=>}[ur]^{\Phi_2^{abc}\htimes \id}  \ar@{=>}[d]_{=} & & \Phi_1^{ad}((x\htimes y)\htimes z) \ar@{=>}[d]^{=}\\
  \Phi_1^{ab}(x)\, \htimes\,(\Phi_1^{bc}(y)\, \htimes\, \Phi_1^{cd}(z))\ar@{=>}[dr]_{\id\htimes \Phi_2^{bcd}} &  & \Phi_1^{ad}(x\htimes(y\htimes z))\\
  & \Phi_1^{ab}(x)\, \htimes\,\Phi_1^{bd}(y\htimes z)\ar@{=>}[ur]_{\Phi_2^{abd}} & }
  \end{equation}
 commutes for all $a,b,c,d\in \CCC_0$ and $x,y,z\in \CCC_{1,0}$.
\end{definition}
Analogously, there is the notion of pseudonatural transformations.
\begin{definition}
 Given two normalised pseudofunctors $\Phi,\Psi:\CCC\rightarrow \CCD$,
 a \uline{pseudonatural} \uline{transformation}, or \uline{weak natural 2-transformation},
 $\alpha:\Phi\Rightarrow \Psi$ consists of 1-morphisms
 $\alpha_1^a:\Phi_0(a)\rightarrow \Psi_0(a)$ for each $a\in \CCC_0$
 together with invertible 2-morphisms $\alpha_2^{ab}(x)$ for all $a,b\in \CCC_0$ and $x\in \CCC_{1,0}(a,b)$ defined by
\begin{equation}
 \xymatrixcolsep{5pc}
\myxymatrix{
 \Phi_0(b) \ar@{->}[r]^{\Phi_1^{ab}(x)}  \ar@{->}[d]_{\alpha_1^b} & \Phi_0(a) \ar@{->}[d]^{\alpha_1^a}\\
 \Psi_0(b) \ar@{->}[r]_{\Psi_1^{ab}(x)} \ar@{=>}[ru]^{\alpha_2^{ab}(x)}  & \Psi_0(a)
}
\end{equation}
such that the diagram
\begin{equation}
\xymatrixcolsep{2.6pc}\myxymatrix{
  \Psi_1^{ab}(x)\, \htimes\,  (\alpha_1^b\, \htimes\,  \Phi_1^{bc}(y))  \ar@{=>}[r]^{ =}  &  (\Psi_1^{ab}(x)\, \htimes\,  \alpha_1^b)\, \htimes\,  \Phi_1^{bc}(y)  \ar@{=>}[r]^{\alpha_2^{ab}(x)\htimes \id}  & (\alpha_1^a\, \htimes\, \Phi_1^{ab}(x))\, \htimes\, \Phi_1^{bc}(y)  \ar@{=>}[d]^{=}\\
 \Psi_1^{ab}(x)\, \htimes\, (\Psi_1^{bc}(y)\, \htimes\,  \alpha_1^c)  \ar@{=>}[u]^{\id\, \htimes\,  \alpha_2^{bc}(y)} & &   \alpha_1^a\, \htimes\, (\Phi_1^{ab}(x)\, \htimes\, \Phi_1^{bc}(y))  \ar@{=>}[d]^{\id\, \htimes\,  \Phi_2^{abc}}\\
(\Psi_1^{ab}(x)\, \htimes\, \Psi_1^{bc}(y))\, \htimes\,  \alpha_1^c   \ar@{=>}[r]_{~~~\Psi_2^{abc}\, \htimes\, \id}  
 \ar@{=>}[u]^{=}  & \Psi_1^{ac}(x\htimes y)\, \htimes\,  \alpha_1^c
 \ar@{=>}[r]_{\alpha_2^{ac}(x\htimes y)} & \alpha_1^a\, \htimes\,  \Phi_1^{ac}(x\htimes y)
 }
\end{equation}
commutes for all $a,b,c\in\CCC_0$ and $x,y\in \CCC_{1,0}$.
\end{definition}

\begin{remark}
 We shall also encounter {\em weak 2-categories} or {\em bicategories} in our discussion. Since we will not require many details, let us just sketch their definition. A weak 2-category is a 2-category in which the horizontal composition is unital only up to 2-isomorphisms, which are known as {\em left} and {\em right unitors},
 \begin{equation}
  \sfl_f:f\htimes \id_{\sfs(f)}\Longrightarrow f\eand \sfr_f:\id_{\sft(f)}\htimes f \Longrightarrow f~,
 \end{equation}
for any 1-morphism $f$.
 Moreover, the horizontal composition is associative only up to a 2-isomorphism, known as the {\em associator},
 \begin{equation}
  \sfa_{f,g,h}: (f\htimes g)\htimes h \Longrightarrow f\htimes
  (g\htimes h) \ ,
 \end{equation}
 for 1-morphisms $f,g,h$. The unitors and the associator have to satisfy certain coherence axioms, see e.g.\ \cite{Benabou:1967:1} for details. 
 
 The most general notion of 2-functor between bicategories is that of a {\em lax 2-functor}, which differs from our normalised weak 2-functor in two ways. Firstly, it is not normalised and there is an additional 2-morphism $\Phi^a_2:\id_{\Phi_0(a)}\Rightarrow \Phi_1^{aa}(\id_a)$. Secondly, the 2-morphisms $\Phi_2^{abc}$ are not necessarily invertible. A detailed discussion of the various 2-functors and natural 2-transformations between bicategories is found e.g.\ in \cite{Jurco:2014mva}.
\end{remark}

\subsection{Lie groupoids and their morphisms}\label{ssec:groupoids}

\begin{definition}
 A \uline{groupoid} is a small category in which all morphisms are
 isomorphisms. A \uline{Lie groupoid} is a groupoid object in
 $\CatMfd$.\footnote{Since $\CatMfd$ does not possess all pullbacks,
   we have to demand that the internal source and target maps
   $\sfs,\sft$ are surjective submersions.} \uline{Morphisms between Lie groupoids} are internal functors between the corresponding groupoid objects in $\CatMfd$. Lie groupoids and their morphisms form the category $\CatLieGrpd$.
\end{definition}
Simple examples of Lie groupoids include the discrete Lie groupoid
$M\rightrightarrows M$ for a manifold $M$, and the delooping of a Lie group $\sG$, $\sB\sG:=(\sG\rightrightarrows *)$, where $*$ denotes the singleton set.

Consider now a surjective submersion $\sigma^Y:Y\thra M$ between two manifolds $M,Y\in\CatMfd$. The $k$-fold fibre product $Y^{[k]}:=Y\times_M Y\times_M\cdots\times_M Y$ consists of $k$-tuples $(y_1,\ldots,y_k)$ with $y_i\in Y$ and $\sigma^Y(y_1)=\cdots=\sigma^Y(y_k)$. An example of a Lie groupoid, which is important to our discussion, is the \v Cech groupoid of $\sigma^Y$.
\begin{definition}\label{def:Cechgroupoid}
 The \uline{\v Cech groupoid} $\check \CCC(Y\thra M)$ of a surjective
 submersion $\sigma^Y: Y\thra M$ between two manifolds $M,Y\in\CatMfd$ has
 objects $Y$ and morphisms $Y^{[2]}$ with structure maps
 $\sfs(y_1,y_2)=y_2$, $\sft(y_1,y_2)=y_1$ and $\id_y=(y,y)$,
 composition $(y_1,y_2)\htimes (y_2,y_3)=(y_1,y_3)$, and the inverse of a morphism $(y_1,y_2)^{-1}=(y_2,y_1)$ for $y,y_1,y_2,y_3\in Y$.
\end{definition}

To establish a relation between the \v Cech groupoid of a surjective
submersion $\sigma^Y: Y\thra M$ and the base manifold $M$, the morphisms in $\CatLieGrpd$ are not enough; instead, we should introduce generalised morphisms, which are also known as Hilsum-Skandalis morphisms~\cite{MR925720,Moerdijk:2003bb}. Below we summarise some relevant definitions and results. For a more detailed review, see e.g.\ \cite{Blohmann:2007ez,Demessie:2016ieh}
\begin{definition}
 A \uline{(left-) principal bibundle} from a Lie groupoid
 $\CCH=(\CCH_1\rightrightarrows \CCH_0)$ to a Lie groupoid
 $\CCG=(\CCG_1\rightrightarrows \CCG_0)$ is a smooth manifold $B$
 together with a smooth map $\tau:B\rightarrow \CCG_0$ and a
 surjective submersion $\sigma:B\thra\CCH_0$. This defines a double fibration
\begin{equation}
  \xymatrixcolsep{2pc}
  \xymatrixrowsep{2pc}
  \myxymatrix{ \CCG_1\ar@<-.5ex>[dr] \ar@<.5ex>[dr] & & \ar@{->}[dl]_{\tau} B\ar@{->>}[dr]^{\sigma} & & \ar@<-.5ex>[dl] \ar@<.5ex>[dl]\CCH_1\\
  & \CCG_0 &  & \CCH_0
    }
\end{equation}
which is endowed with left- and right-action maps 
 \begin{equation}
  \CCG_1{}\times^{\sfs,\tau}_{\CCG_0} B\longrightarrow B\eand B{}\times^{\sigma,\sft}_{\CCH_0} \CCH_1\longrightarrow B~,
 \end{equation}
 such that the following relations are satisfied:
 \begin{itemize}
  \item[(i)] $g_1(g_2b)=(g_1\,g_2) b$ for all $(g_1,g_2,b)\in \CCG_1\times^{\sfs,\sft}_{\CCG_0}\CCG_1\times^{\sfs,\tau}_{\CCG_0}B$;
  \item[(ii)] $(bh_1) h_2=b (h_1\, h_2)$ for all $(b,h_1,h_2)\in B\times^{\sigma,\sft}_{\CCH_0}\CCH_1\times^{\sfs,\sft}_{\CCH_0}\CCH_1$;
  \item[(iii)] $b\;\id_\CCH(\sigma(b))=b$ and $\id_\CCG(\tau(b))\;b=b$ for all $b\in B$;
  \item[(iv)] $g (bh)=(gb)h$ for all $(g,b,h)\in \CCG_1\times^{\sfs,\tau}_{\CCG_0} B\times^{\sigma,\sft}_{\CCH_0} \CCH_1$;
  \item[(v)] The map $\CCG_1\times^{\sfs,\tau}_{\CCG_0}B\rightarrow
    B\times_{\CCH_0}B$, $(g,b)\mapsto (gb,b)$ is an isomorphism.
 \end{itemize}
 
 The \uline{(horizontal) composition of bibundles} $B:\CCG\rightarrow \CCF$ and $B':\CCH\rightarrow \CCG$ is defined as 
 \begin{equation}
  B\htimes B'=(B\times_{\CCG_0} B'\, )\, \big/\, \CCG_1~,
 \end{equation}
where the quotient is by the diagonal $\CCG_1$-action.
 
 A \uline{morphism of bibundles} $B,B':\CCH\rightarrow \CCG$ is an
 element of $\shom_\CatMfd(B,B'\, )$ which is equivariant with respect to the left- and right action maps.
 
 A \uline{bibundle equivalence} is a bibundle $B:\CCH\rightarrow \CCG$ which is simultaneously a bibundle $B:\CCG\rightarrow \CCH$.
\end{definition}

Bibundles indeed generalise morphisms of Lie groupoids. Given such a morphism $\phi:\CCH\rightarrow \CCG$, we define its \emph{bundlisation} as the bibundle
\begin{equation}
  \xymatrixcolsep{2pc}
  \xymatrixrowsep{2pc}
  \myxymatrix{ \CCG_1\ar@<-.5ex>[dr] \ar@<.5ex>[dr] & &
    \ar@{->}[dl]_{\sft} \hat\phi = \CCH_0\times^{\phi_0,\sfs}_{\CCG_0} \CCG_1\ar@{->>}[dr]^{\pi} & & \ar@<-.5ex>[dl] \ar@<.5ex>[dl]\CCH_1\\
  & \CCG_0 &  & \CCH_0
    }
\end{equation}
where $\sft$ is the target map in $\CCG$ and $\pi$ is the trivial projection. The left- and right-action maps are given by
\begin{equation}
 g'(x,g):=(x,g'\circ g)\eand (x,g)h:=(\sfs(h),g\circ \phi_1(h))
\end{equation}
for $g,g'\in \CCG_1$, $h\in \CCH_1$ and $(x,g)\in \hat\phi$. This motivates the following definition.
\begin{definition}
 The weak 2-category {$\CatBibun$} has Lie groupoids as objects, bibundles as 1-morphisms and bibundle morphisms as 2-morphisms. Horizontal composition in this bicategory is composition of bibundles. An equivalence in $\CatBibun$ is a bibundle equivalence.
\end{definition}
The category $\CatBibun$ can be viewed as the 2-category of ``stacky manifolds'', since Lie groupoids are presentations of differentiable stacks and bibundle equivalence amounts to Morita equivalence.

We can now readily establish a relation between the \v Cech groupoid
$\check \CCC(Y\thra M)$ and the manifold $M$: The former is equivalent to the discrete Lie groupoid $M\rightrightarrows M$ via a bibundle with total space $Y$, and obvious structure and action maps:
\begin{equation}\label{eq:moritaY}
  \xymatrixcolsep{2pc}
  \xymatrixrowsep{2pc}
  \myxymatrix{ Y^{[2]}\ar@<-.5ex>[dr] \ar@<.5ex>[dr] & & \ar@{->>}[dl]_{\id_Y} Y\ar@{->>}[dr]^{\sigma^Y} & & \ar@<-.5ex>[dl] \ar@<.5ex>[dl] M\\
  & Y &  & M
    }
\end{equation}

We shall make use of bibundles in two other contexts. Firstly, certain
bibundle equivalences encode stable isomorphisms between
$\sU(1)$-bundle gerbes as explained in
Section~\ref{ssec:gerbes}. Secondly, we can restrict the weak 2-category
$\CatBibun$ to Lie groupoids with a single object which has linear
spaces as morphisms, and correspondingly we restrict to bibundles with
linear left- and right-actions as well as to linear bibundle
morphisms; in particular, there is the category $\CatBimod$ of unital associative algebras, bimodules and bimodule morphisms which appears in Sections~\ref{ssec:2-vector-spaces} and \ref{ssec:lie2grp_reps}.

\subsection{Deligne cohomology}\label{ssec:Cech_Deligne}

The nerve $\CN\check \CCC(Y\thra M)$ of the \v Cech groupoid $\check
\CCC(Y\thra M)$ of a surjective submersion $\sigma^Y:Y\thra M$ is a
simplicial manifold with $(k-1)$-simplices given by the $k$-fold fibre products $Y^{[k]}=Y\times_M Y\times_M\cdots\times_M Y$:
\begin{equation}
 \CN\check \CCC(Y\thra M)~=~\cdots Y^{[4]} \quadarrow Y^{[3]}\triplearrow Y^{[2]}\doublearrow Y
\end{equation}
The face maps $\sff^k_i:Y^{[k]}\rightarrow Y^{[k-1]}$ can be combined
into a coboundary operator on differential forms $\check \delta : \Omega^p(Y^{[k]})\rightarrow \Omega^p(Y^{[k+1]})$ by
\begin{equation}
 (\check \delta \alpha)(s):= \sum_{j=0}^k\, (-1)^j\,(\sff^k_j)^* \alpha(s)
\end{equation}
for $\alpha\in \Omega^p(Y^{[k]})$ and $s\in Y^{[k+1]}$; then $\check \delta\circ \check \delta=0$. This gives rise to the exact complex~\cite{Murray:9407015,Murray:2007ps}
\begin{equation}\label{eq:fundamental_sequence}
0\longrightarrow \Omega^p(M)\xrightarrow{(\sigma^Y)^*}\Omega^p(Y)\xrightarrow{~\check \delta~} \Omega^p(Y^{[2]})\xrightarrow{~\check \delta~} \Omega^p(Y^{[3]})\xrightarrow{~\check \delta~} \cdots
\end{equation}
Similarly, we define for each function $g$ from $Y^{[k]}$ to an
abelian group $\sA$ a map $\check \delta g:Y^{[k+1]}\rightarrow \sA$
given by
\begin{equation}
 (\check \delta g)(s):= \sum_{j=0}^k\, (-1)^j\, g(\sff^k_js)~.
\end{equation}
Finally, we also define for an $\sA$-bundle $P\rightarrow Y^{[k-1]}$ an $\sA$-bundle $\check \delta P$ over $Y^{[k]}$ by
\begin{equation}
 \check \delta P =(\sff^k_0)^*(P)~\otimes~\big((\sff^k_1)^*(P) \big)^*~\otimes~(\sff^k_2)^*(P)~\otimes~\cdots
\end{equation}

A simple choice for a surjective submersion is $Y = \frU =
\bigsqcup_a\, U_a$, the total space of an open cover $(U_a)_{a \in
  \Lambda}$ of $M$. In this case, the space $\frU^{[k]}$ is just the
disjoint union of all $k$-fold intersections $U_{a_1\cdots
  a_k}:=U_{a_1}\cap\cdots \cap U_{a_k}$, including self-intersections
$U_a \cap U_a$, of sets in the cover. 

\begin{definition}
 Given a sheaf $\CS$ of abelian groups over $M$ and an open cover
 $\frU$ of $M$, let $\check
 C^\bullet=\bigcup_k \, \check C^k$ be the cosimplicial object with $\check C^k(\frU,\CS):=\CS(\frU^{[k]})$. Elements of $\check C^k(\frU,\CS)$ are  \uline{$k$-cochains}. The coboundary operator $\check \delta$ on $\check C^\bullet$ is the \uline{\v Cech differential}. The complex $(\check C^\bullet,\check \delta)$ is the \uline{\v Cech complex}, containing \uline{\v Cech $k$-cocycles} and \uline{\v Cech $k$-coboundaries}. We denote the \uline{$k$-th \v Cech cohomology group} by $\check H^k(\frU,\CS)$.
\end{definition}
We can remove the dependence on the cover $\frU$ by taking the direct limit over all covers of a manifold $M$. This yields the \v Cech cohomology groups $\check H^k(M,\CS)$.

If $\sA$ is an abelian Lie group, let $\underline{\sA}_M$ denote the
sheaf of smooth $\sA$-valued functions on a manifold $M$, and $\sA_M$
the sheaf of locally constant $\sA$-valued functions on
$M$. Evidently, $\underline{\RZ}_M=\RZ_M$.

\begin{example}
Consider an open cover $\frU=\bigsqcup_a\, U_a$ of a manifold $M$ and
$\underline{\sU(1)}_M$ the sheaf of smooth $\sU(1)$-valued functions
on $M$. Then a \v Cech 1-cocycle $f=(f_{ab})$ is a collection of maps
$f_{ab}: U_{ab}\rightarrow \sU(1)$ satisfying $f_{ab}\,
f_{bc}=f_{ac}$ on $U_{abc}$. Thus $f$ encodes a principal $\sU(1)$-bundle $P_f$ over $M$ subordinate to the cover $\frU$. Equivalent such principal $\sU(1)$-bundles are in the same \v Cech cohomology class in $\check H^1(\frU,\underline{\sU(1)}_M)$. The exponential map gives rise to the short exact sequence
\begin{equation}\label{eq:exp-map-sequence}
	0 \longrightarrow \RZ_M \longrightarrow \FR_M
        \xrightarrow{\exp(2\pi\, \di\, -)} \sU(1)_M \longrightarrow 0
\end{equation}
which induces a long exact cohomology sequence containing
\begin{equation}
 \cdots \longrightarrow \check H^1(\frU,\underline{\FR}_M) \longrightarrow \check H^1(\frU,\underline{\sU(1)}_M) \xrightarrow{~\sigma~} \check H^2(\frU,\underline{\RZ}_M) \longrightarrow \check H^2(\frU,\underline{\FR}_M) \longrightarrow \cdots
\end{equation}
Since the sheaf $\underline{\FR}_M$ of smooth real functions on $M$ is a fine sheaf and therefore the cohomology groups $\check H^k(\frU,\underline{\FR}_M)$ are trivial, the connecting homomorphism $\sigma$ is an isomorphism and we have $\check H^1(\frU,\underline{\sU(1)}_M)\cong \check H^2(\frU,\underline{\RZ}_M)$. 
By mapping $\check H^2(\frU,\RZ_M)$ into $\check H^2(\frU,\FR_M)$ and using the \v Cech-de Rham isomorphism, we obtain a 2-form in the de Rham cohomology group $H^2_{\rm dR}(M)$ representing the first Chern class\footnote{As explained in detail below, it represents the first Chern class modulo torsion.} of $P_f$.
\label{ex:U1bundle}\end{example}

To endow the principal $\sU(1)$-bundle of the example above with a connection, we need to generalise from \v Cech cohomology to Deligne cohomology.
\begin{definition}
 The \uline{Deligne complex in degree $n$} is the complex 
  \begin{equation}
  \CD^\bullet(n) =%
  \underline{\sU(1)}_M \xrightarrow {\dd \log}
  \di\, \underline{\Omega}^1_M \xrightarrow{ \ \dd \ } \cdots
  \xrightarrow{ \ \dd \ } \di\, \underline{\Omega}_M^n
  \end{equation}
  where $\di\, \underline{\Omega}_M^n$ is the sheaf of $\di \, \FR$-valued differential $n$-forms on $M$, and
  \begin{equation}
  \dd \log g = g^* \mu_{\sU(1)} = g^{-1} \, \dd g
  \end{equation}
  with $\mu_\sG$ the Maurer-Cartan form on a Lie group $\sG$. The set
  $C^k(M,\CD^\bullet(n))$ of \v Cech $k$-cochains for the sheaf
  $\CD^\bullet(n)$ gives the \uline{Deligne $k$-cochains in degree
    $n$}. The collection of sets $C^k(M,\CD^\bullet(n))$ forms a complex with differential $\delta_{\CD}$ given by the evident graded sum of the \v Cech differential and the differential in $\CD^\bullet(n)$. The resulting hypercohomology group $H^k(M,\CD^\bullet(n))$ is  the \uline{$k$-th Deligne cohomology group in degree $n$}.
\end{definition}
Our conventions differ from a perhaps more common but equivalent choice, cf.~\cite{Waldorf:2007aa}; another way of representing Deligne cohomology is found in \cite{0817647309}, which contains a degree shift compared to our conventions here.

\begin{proposition}[\cite{0817647309}]
	\label{st:Deligne_ex_seq_1}
	For a good open cover $\frU$, denote the composition of the projection of a Deligne $k$-cocycle $(g,\alpha_1,\ldots,\alpha_n) \in H^k(M,\CD^\bullet(n))$ onto its component $g \in \check{H}^k(\frU,\underline{\sU(1)}_M) \cong H^k(M,\underline{\sU(1)}_M)$ with the isomorphism $H^k(M,\underline{\sU(1)}_M) \rightarrow H^{k+1}(M,\RZ_M)$ by
	\begin{equation}
	 \pr_\CD: H^k(M,\CD^\bullet(n)) \longrightarrow H^{k+1}(M,\RZ_M)~.
	\end{equation}
	This map has the following properties:
	\begin{enumerate}
		\item It is an isomorphism for $k > n$.
		\item For $k = n$ there is an exact sequence of abelian groups
		\begin{equation}
		\label{eq:DD_exact_sequence}
		0 \longrightarrow 2\pi\, \di\, \Omega^n_{{\rm cl},\RZ}(M)
                \longrightarrow \di\, \Omega^n(M) \xrightarrow{ \ t^n \
                } H^n(M,\CD^\bullet(n)) \xrightarrow{\pr_\CD} H^{n+1}(M,\RZ_M) \longrightarrow 0
		\end{equation}
		where $\Omega^n_{{\rm cl},\RZ}(M)$ denotes the closed global $n$-forms on $M$ with integer periods.
	\end{enumerate}
\end{proposition}

That is, for $k > n$ all information encoded in the class $[g,\alpha_1,\ldots,\alpha_n]$ can be recovered from the class $\pr_\CD(g)$. For $k \leq n$, the highest degree form $\alpha_n$ is not closed.
Hence we can project to that part of a Deligne cocycle and apply the exterior derivative.
As the latter squares to zero, this map is well-defined on cohomology classes.
For the same reason and by the Deligne cocycle condition we have $\check{\delta}\, \dd \alpha_n = 0$.
As the \v{C}ech complex is a resolution, this implies that there
exists a unique element $\beta$ of $\di\, \Omega^{k+1}(M)$ such that $\dd \alpha_n = \check{\delta} \beta$.
This is called the \emph{curvature of the Deligne class}, and we have constructed a group homomorphism
\begin{equation}
\curv: H^k(M,\CD^\bullet(n)) \longrightarrow 2\pi\, \di\,
\Omega^{k+1}_{{\rm cl},\RZ}(M) \ .
\end{equation}

\begin{proposition}[\cite{Waldorf:2007aa}]
	\label{st:Deligne_ex_seq_2}
	The curvature $\curv: H^k(M,\CD^\bullet(n)) \rightarrow 2\pi\,
        \di\,
        \Omega^{k+1}_{{\rm cl},\RZ}(M)$ has the following properties:
	\begin{enumerate}
		\item It is an isomorphism for $0 \leq k < n$.
		\item For $k = n$ there is a short exact sequence of abelian groups
		\begin{equation}
		\label{eq:curv_exact_sequence}
		0 \longrightarrow H^n(M,\sU(1)) \longrightarrow
                H^n(M,\CD^\bullet(n)) \xrightarrow{\curv} 2\pi\, \di\,
                \Omega^{n+1}_{{\rm cl},\RZ}(M) \longrightarrow 0 
		\end{equation}
	\end{enumerate}
\end{proposition}

The images of the maps $\pr_\CD$ and $\curv$ are closely related.
From the short exact sequence \eqref{eq:exp-map-sequence}, we obtain
homomorphisms $H^{n+1}(M,\RZ_M) \rightarrow H^{n+1}(M,\FR) \cong
H^{n+1}_{\mathrm{dR}}(M)$. Their images now agree with the image of
the curvature of the Deligne class under $-\di: \di\,
\Omega^{n+1}_{{\rm cl},\RZ}(M) \rightarrow H^{n+1}_{\mathrm{dR}}(M)$. The representation in de Rham cohomology, however, loses information since 
\begin{equation}
	\frac{\Omega^{k+1}_{{\rm cl},\RZ}(M)}{\dd \Omega^{k}(M)} \cong \frac{H^{k+1}(M,\RZ_M)}{\mathrm{Tor}(H^{k+1}(M,\RZ_M))}
\end{equation}
and the kernel of the homomorphism from integer to real cohomology is precisely the torsion subgroup of $H^{k+1}(M,\RZ_M)$.

\begin{corollary}
	\begin{enumerate}
		\item The curvature $\curv(\alpha)$ of $\alpha \in C^n(M,\CD^\bullet(n))$ is exact if and only if $\pr_\CD(\alpha)$ is torsion.
		\item If $[\alpha] \in
                  \mathrm{Tor}(H^n(M,\CD^\bullet(n)))$, then its
                  curvature $n+1$-form is exact.
	\end{enumerate}
\end{corollary}

\begin{example}
 Equivalent principal $\sU(1)$-bundles with connection are now described by the first Deligne cohomology group in degree $1$. A representative for this class is a Deligne 1-cocycle in degree 1,
 \begin{equation}
  (f_{ab},A_{a}) \ \in \ C^1(\frU,\CD^\bullet(1))\cong
  C^1(\frU ,\underline{\sU(1)}_M)\oplus C^0(\frU ,
  \di\, \underline{\Omega}_M^1) \ ,
 \end{equation}
 such that the cocycle relations
 \begin{equation}
  f_{ab}\, f_{bc}\, f_{ca}= 1\eon U_{abc} \eand A_a-A_b =
  \dd \log f_{ab} \eon U_{ab}
 \end{equation}
  are fulfilled. The equivalence relations, also known as gauge transformations, are
 described by Deligne 1-coboundaries: A coboundary $(g_a)\in \CC^0(\frU,\underline{\sU(1)}_M)$ between two Deligne 1-cocycles $(f_{ab},A_a)$ and $(\tilde f_{ab},\tilde A_a)$ satisfies
\begin{equation}
\tilde f_{ab}=g^{-1}_a\, f_{ab}\, g_b \eand 
 \tilde A_{a}=A_{a}+ \dd\log g_a~.
\end{equation}
We thus have a \v Cech 1-cocycle $(f_{ab})$ defining smooth transition
functions of a principal $\sU(1)$-bundle on $M$ together with local connection 1-forms $A_a$ on the patches $U_a$: $\nabla|_{U_a}=\dd+A_a$.

The projection $\pr_\CD(f_{ab},A_a)$ yields the first Chern class of
the principal $\sU(1)$-bundle, while $\curv$ gives its curvature
$F_\nabla$ which is a global 2-form on $M$. For non-vanishing torsion, the
curvature 2-form $F_\nabla$ is not sufficient to describe the full first
Chern class, as then there can be non-trivial flat bundles; an example of such a case is the canonical bundle on an Enriques surface.
\end{example}

\begin{remark} 
 It is now natural to seek similar interpretations for the $n$-th
 Deligne cohomology groups in degree $n$. They define geometric
 objects that one might call local Hitchin-Chatterjee $(n-1)$-gerbes
 or certain kinds of abelian principal $n$-bundles. We will come back
 to them in Section~\ref{sect:U(1)-bundle-gerbes}.
\label{rem:higherDeligne}\end{remark}

\subsection{Lie 2-groupoids and Lie 2-groups\label{sec:Lie2groups}}

The notion of a Lie groupoid is now readily categorified. A
particularly elegant and general way of doing this is to use Kan
simplicial manifolds. Here, however, we are merely interested in Lie 2-groupoids, which are straightforwardly defined. We can categorify using either strict 2-categories or bicategories, leading to strict and weak 2-groupoids.

\begin{definition}
 A \uline{(strict) 2-groupoid} is a (strict) 2-category in which both 1- and 2-mor\-phisms are (strictly) invertible.
\end{definition}
One readily internalises this definition in the strict 2-category
$\CatMfdCat$ whose objects are categories internal to $\CatMfd$, 1-morphisms are functors in $\CatMfd$ and 2-morphisms are natural transformations in $\CatMfd$.
\begin{definition}
 A \uline{(strict) Lie 2-groupoid} is a (strict) 2-groupoid internal to $\CatMfdCat$.
\end{definition}
Just as any set $X$ can be lifted to a discrete category $X\rightrightarrows X$ with all morphisms being identities, any Lie groupoid $\CCG_1\rightrightarrows \CCG_0$ can be lifted to a strict Lie 2-groupoid $\CCG_1\rightrightarrows \CCG_1\rightrightarrows \CCG_0$ with all 2-morphisms being identities. In particular, we thus obtain the \v Cech 2-groupoid $\check \CCC(Y\thra M)$ of a surjective submersion.

A Lie group $\sG$ is encoded in the morphisms of its delooping, which is the one object Lie groupoid $\sB\sG=(\sG\rightrightarrows *)$. Analogously, we can define Lie 2-groups via their delooping.
\begin{definition}
 A \uline{(strict) Lie 2-group} is the monoidal category of morphisms of a (strict) Lie 2-groupoid with one object.
\end{definition}
Equivalently, a strict Lie 2-group is a category internal to the category of Lie groups $\CatLieGrp$ or a group object internal to $\CatMfdCat$. An important example of a strict Lie 2-group is $\sB\sU(1)=(\sU(1)\rightrightarrows *)$ with delooping 
\begin{equation}
 \sB\sB\sU(1)=(\sU(1)\rightrightarrows *\rightrightarrows *)~.
\end{equation}

Applying the tangent functor to a Lie 2-group, we obtain a linear category which is part of the structure of a {\em Lie 2-algebra}. The notion of a Lie 2-algebra is a vertical categorification of that of a Lie algebra; that is, a Lie 2-algebra is a linear category endowed with a functorial bilinear operation, called a \emph{Lie bracket functor}, which is skew-symmetric and satisfies the Jacobi identity up to natural transformations, called respectively {\em alternator} and {\em Jacobiator}, that in turn fulfill certain coherence laws \cite{Roytenberg:0712.3461}. We are exclusively interested in {\em semistrict Lie 2-algebras}, for which the alternator is trivial. A full Lie functor, taking a Lie 2-group to a semistrict Lie 2-algebra (even Lie $\infty$-groupoids to Lie $\infty$-algebroids) was given by \v Severa~\cite{Severa:2006aa}. 

It will also prove useful for us to note that strict 2-groups are categorically equivalent to crossed modules of Lie groups \cite{Baez:0307200}.
\begin{definition}\label{def:crossed_module_Lie_groups}
 A \uline{crossed module of Lie groups} is a pair of Lie groups $\sG$
 and $\sH$, together with an action $\triangleright:\sG\times
 \sH\rightarrow \sH$ as automorphisms and a group homomorphism
 $\theta:\sH\to \sG$ such that
 \begin{equation}
  \begin{aligned}
    \theta(g\triangleright h)=g \, \theta(h) \,
    g^{-1}\eand\theta(h_1)\triangleright h_2 = h_1\, h_2\, h_1^{-1}~.
  \end{aligned}
 \end{equation}
\end{definition}
In particular, given a crossed module of Lie groups
$\sH\xrightarrow{ \ \theta \ }\sG$, the corresponding strict Lie 2-group
reads as $\sG\ltimes \sH\rightrightarrows \sG$ with $\sfs(g,h)=g$,
$\sft(g,h)=\theta(h)\, g$ and $\id_g=(g,\unit)$. The horizontal composition is given by the group multiplication on the semidirect product $\sG\ltimes \sH$ and the vertical composition is the group multiplication on $\sH$.

Applying the tangent functor to a crossed module of Lie groups $\sH\xrightarrow{\ \theta \ } \sG$, we
obtain a {\em crossed module of Lie algebras}
$\theta:\sLie(\sH)\rightarrow \sLie(\sG)$ with an action
$\triangleright$ satisfying obvious relations.\footnote{With a slight
abuse of notation, here we use the same notation for $\theta$, $\triangleright$ and their differentials.} This crossed module of Lie algebras is categorically equivalent to a semistrict Lie 2-algebra \cite{Baez:2003aa}, and can also be obtained by applying \v Severa's Lie functor \cite{Severa:2006aa} to the strict Lie 2-group corresponding to $\sH\xrightarrow{\ \theta \ } \sG$.

\subsection{2-vector spaces}\label{ssec:2-vector-spaces}

One of the crucial ingredients in our later discussion will be 2-vector spaces. Let us therefore also discuss their definition in appropriate detail.

A complex vector space is a $\FC$-module, and therefore a 2-vector
space should be a module category for some other category $\CCC$,
where $\CCC$ is a categorified ground field taking over the role of $\FC$. From our previous discussion, one might be tempted to use the category $\FC\rightrightarrows *$ as a categorified ground field, but our later constructions will motivate the use of $\CatVect$, which is also in line with more general principles of categorification.

In regarding a complex vector space as a $\FC$-module, we demand compatibility with addition and multiplication in $\FC$. On a categorified ground field such as $\CatVect$, we cannot expect the inverses of addition and multiplication to exist. We therefore focus on the following type of categories.
\begin{definition}
A category $\CCR$ with two monoidal structures $\otimes$ and $\oplus$, where $\otimes$ distributes over $\oplus$, is  a \uline{rig category}.
\end{definition}

The suitable notion of module category for rig categories generalises notions of left-module categories as e.g.\ in \cite{Douglas:2013aea}.
\begin{definition}
Given a rig category $\CCR$, a \uline{left $\CCR$-module category} is a monoidal category $(\CCV,\oplus_\CCV)$ together with a functor $\rho: \CCR \times \CCV \rightarrow \CCV$ and natural transformations capturing distributivity, associativity and unitality of $\rho$,
\begin{equation}
\begin{aligned}
\sfd^{\rm l}_{U,\CV,\CW}&: \rho (U, \CV \oplus_{\CCV} \CW) \isom \rho(U, \CV) \oplus_{\CCV} \rho(U, \CW)~,\\[4pt]
\sfd^{\rm r}_{U,V,\CW}&: \rho(U \oplus_{\CCR} V, \CW) \isom \rho(U, \CW) \oplus_{\CCV} \rho(V,\CW)~,\\[4pt]
\sfa_{U,V,\CW}&: \rho(U \otimes_\CCR V, \CW) \isom \rho \big( U, \rho(V,\CW) \big)~,\\[4pt]
\sfl_\CW&: \rho(1_\CCR,\CW) \isom \CW~.
\end{aligned}
\end{equation}
We demand that $\rho(0_\CCR,\CW) = \rho(V,0_\CCV) = 0_\CCV$ and that the obvious coherence axioms are satisfied.
\end{definition}
\noindent In particular, $\CCR$ is trivially a module category over itself.

This directly yields the definition of the 2-category of 2-vector spaces.
\begin{definition}\label{def:2-vector_space}
 Given a rig category $\CCR$, a \uline{2-vector space} is a left $\CCR$-module category. The weak 2-category {$2\CatVect_\CCR$} consists of left $\CCR$-module categories, functors weakly commuting with the module actions between them and natural transformations between those.
\end{definition}

A straightforward yet important example is the following.
\begin{example}\label{ex:Baez-Crans-2-vector-spaces}
The simplest possible choice for a rig category is the discrete
category $\FC\rightrightarrows \FC$ with obvious monoidal
structures. The resulting 2-vector spaces
$2\CatVect_{\FC\rightrightarrows\FC}$ consist simply of categories
having complex vector spaces as sets of objects and morphisms. These
are categories internal to $\CatVect$, and therefore they are precisely
the 2-vector spaces of Baez and Crans \cite{Baez:2003aa}. {Baez-Crans
  2-vector spaces} are sufficient to describe a theory of semistrict
Lie $2$-algebras as defined in Section~\ref{sec:Lie2groups}, but they are rather restrictive in other regards. We
shall encounter them again in our discussion of the Lie 2-algebra of
observables on a 2-plectic manifold that we undertake in
Section~\ref{ssec:multisymplectic}. 
\end{example}

The next obvious candidate for a categorified ground field is $\CatVect$. By themselves, the resulting left $\CatVect$-module categories are awkward, but one can consider left $\CatVect$-module categories with basis \cite{Schreiber:2009:357-401}. To understand what a basis of a 2-vector space might be, recall that a vector space $V$ can be identified with the set of homomorphisms from a basis $B$ to $\FC$: $V\cong \shom_{\CatSet}(B,\FC)$. Analogously, we define the following.
\begin{definition}
 A \uline{2-basis of a 2-vector space} $\CCV$ is a category $\CCB$ such that $\CCV$ is equivalent to $\shom_{\CatCat}(\CCB,\CatVect)$.
\end{definition}
Just as in the case of vector spaces, we cannot possibly expect a 2-vector space to come with a unique 2-basis. We shall focus on 2-bases given by $\CatVect$-enriched categories $\CCB$ such that the 2-vector spaces are in fact algebroid modules. Since $\CatVect$ itself is enriched in $\CatVect$, we replace $\shom_{\CatCat}(\CCB,\CatVect)$ with the corresponding functor enriched in $\CatVect$. In these cases, we have the following statement.
\begin{proposition}
Two 2-vector spaces with bases $\CCB_1$ and $\CCB_2$ are equivalent if their bases are Morita equivalent, cf.\ Section~\ref{ssec:groupoids}.
\end{proposition}
This is a direct consequence of the fact that two algebroids have equivalent module categories if they are Morita equivalent.

Consider now the discrete category $\CCB_X=(X\rightrightarrows X)$ for some set $X$, which can be enriched in $\CatVect$ by adding the ground field $\FC$ over each morphism.
\begin{definition}
 The \uline{free 2-vector space $\CCV_X$ over $X$} is the 2-vector space with 2-basis $\CCB_X$,
 \begin{equation}
  \CCV_X:=\shom_{\CatCat}(\CCB_X,\CatVect)~.
 \end{equation}
\end{definition}
Explicitly, $\CCV_X$ consists of $X$-indexed families of vector spaces and $X$-indexed families of morphisms between vector spaces. A further specialisation is now the following.
\begin{definition}[cf.\ \cite{kapranov19942}]
 A \uline{Kapranov-Voevodsky 2-vector space} is a free 2-vector space over an $n$-element set $X$. 
 We have
 \begin{equation}
  \shom_{\CatCat}(\CCB_X,\CatVect)\cong \CatVect^n~.
 \end{equation}
 We call $n$ the \uline{2-rank} of $\CatVect^n$.
\end{definition}
An equivalence $\CCV \cong \CatVect^n$ carries over to $\CCV$ additional properties stemming from those of $\CatVect$, like for instance that $\CCV$ is a semisimple abelian category (cf.\ also~\cite{Baez:9609018}) where binary limits and colimits coincide with the additive monoidal structure.

We shall be mostly interested in free 2-vector spaces over sets. Before continuing, however, let us briefly comment on the more general case \cite{Schreiber:2009:357-401,Schreiber:2008aa}. A $\CatVect$-enriched 2-basis $\CCB=(\CCB_1\rightrightarrows \CCB_0)$ is in fact an algebroid and we regard them as objects in a weak 2-category of algebroids, bibundles and bibundle morphisms. This weak 2-category can be regarded as a linearised version of $\CatBibun$. Restricting to 2-bases with a single object, we obtain the weak 2-category $\CatBimod$ of algebras, bimodules and bimodule morphisms. This gives rise to the following result.
\begin{proposition}\label{prop:representation_functor}
 The $\CatVect$-enriched functor $\shom_{\CatCat}(-,\CatVect)$ embeds the weak 2-cate\-gory $\CatBimod$ into $2\CatVect_\CatVect$. In particular, an algebra $\fra$ is mapped to the category of (left-) $\fra$-modules, which is a left $\CatVect$-module category. A bimodule acts on a right $\fra$-module by right multiplication, which also induces the action of bimodule morphisms.
\end{proposition}

We will not go further into the theory of 2-vector spaces and the
2-category they form, and close the present discussion with the notion of the dual of a 2-vector space.
\begin{definition}
	The \uline{dual 2-vector space of a 2-vector space} $\CCV$ is given by
	\begin{equation}
	\CCV^* = \shom_{2\CatVect_\CatVect}(\CCV,\CatVect) \ .
	\end{equation}
\end{definition}

\subsection{Lie 2-group representations}\label{ssec:lie2grp_reps}

As an immediate application of 2-vector spaces, let us consider some examples of representations of Lie 2-groups on 2-vector spaces. The general discussion is very involved and goes far beyond the scope of this paper, see e.g.\  \cite{Barrett:2004zb,Elgueta:2007:53-92,Bartlett:2008ji,Freidel:2012:0-0}.

A representation of a Lie group $\sG$ is simply a
functor from $\sB\sG$ to $\CatVect$. Accordingly, the higher version
of this notion is as follows.
\begin{definition}
 A \uline{2-representation of a strict Lie 2-group} $\CCG$ is a normalised pseudofunctor from the Lie 2-groupoid $\sB\CCG$ to $2\CatVect_{\CatVect}$.
\end{definition}
Again, since the weak 2-category $2\CatVect_{\CatVect}$ is rather
inconvenient, we have to restrict to meaningful 2-vector spaces. We
are particularly interested in the categorified version of the
fundamental representation of $\sB\sU(1)$. For the Lie group $\sU(1)$,
the relevant vector space is $\shom_{\CatSet}(*,\FC)\cong \FC$. There
is a canonical representation of strict Lie 2-groups, which is given
by the following.
\begin{proposition}[\cite{Schreiber:2009:357-401}]
\label{prop:canonical_rep}
 Given a crossed module of Lie groups $\sH\xrightarrow{ \ \theta \ }\sG$ and
 a representation $\varrho$ of $\sH$ such that the action of $\sG$ on
 $\sH$ extends to algebra automorphisms of the representation algebra
 $\langle \varrho(\sH)\rangle$, there is the following
 2-representation of the corresponding strict Lie 2-group: The single object is mapped to $\langle \varrho(\sH)\rangle$, an element of $\sG$ is mapped to the algebra automorphism corresponding to the action of $\sG$ and an element $h$ of $\sH$ is mapped to $\varrho(h)$. The resulting pseudofunctor $\Psi^\varrho$ is then composed with the pseudofunctor from Proposition~\ref{prop:representation_functor} embedding $\CatBimod$ into $2\CatVect_\CatVect$.
\end{proposition}

\begin{example}\label{rem:def_rho_f}
Let us specialise to the fundamental representation
$\varrho_{\rm f}$ of $\sU(1)$. The pseudo\-functor
$\Psi^{\varrho_{\rm f}}$ of Proposition~\ref{prop:canonical_rep} maps the single object $*$ to $\FC\in \CatBimod_0$, the single morphism $*$ to $\id_\FC\in \CatBimod_{1,0}$ and the group $\sU(1)$ to an element of $\CatBimod_{1,1}\subset\sGL(1,\FC)=\FC^\times$. The 2-vector space underlying this representation has basis $\FC\rightrightarrows *$ and it is the category of left $\FC$-modules: $\shom_{\CatCat}(\FC\rightrightarrows *,\CatVect)\cong \CatMod_\FC \cong \CatVect$, as expected.
\end{example}

Let us give one more example that will appear in our later discussion.
\begin{example}\label{ex:u(1)u(n)-rep}
 Given a strict Lie 2-group $\CCG:=(\sU(1)\ltimes
 \sU(n)\rightrightarrows \sU(n))$, the fundamental representation
 $\varrho_{\rm f}$ of $\sU(n)$ induces an interesting 2-representation
 on Baez-Crans 2-vector spaces. Recall that the category
 $2\CatVect_{\FC\rightrightarrows\FC}$ consists of Baez-Crans 2-vector
 spaces, cf.\ Example~\ref{ex:Baez-Crans-2-vector-spaces}, linear
 functors between these and natural transformations between those. We
 now have a normalised pseudofunctor $\Psi^{\varrho_{\rm f}}:\sB\CCG\rightarrow
 2\CatVect_{\FC\rightrightarrows\FC}$ with $\Psi^{\varrho_{\rm f}}_0(*)=\FC^n$,
 \begin{equation}
 \begin{aligned}
  \Psi^{\varrho_{\rm f}}_{1,1}\Big(g\xrightarrow{(h,g)} e(h)\,
  g\Big)&=\Big(\varrho_{\rm
    f}(g)\xrightarrow{(\varrho_f(e(h)),\varrho_f(g))} \varrho_{\rm
    f}(e(h)\, g)\Big)~,\\[4pt]
  \Psi^{\varrho_{\rm f}}_{1,0}(g)&=(\varrho_{\rm f}(g))~,
 \end{aligned}
 \end{equation}
 where $e:\sU(1)\rightarrow \sU(n)$ is the diagonal embedding, and
 $\Psi^{\varrho_{\rm f}}_2$ is trivial. 
\end{example}

\section{\texorpdfstring{$\sU(1)$}{U(1)}-bundle
  gerbes}\label{sect:U(1)-bundle-gerbes}

In this section we review the main definitions, results and examples concerning
bundle gerbes that will play a prominent role throughout this paper.

\subsection{Definitions}\label{ssec:gerbes}

The original definition of a gerbe dates back to Giraud \cite{Giraud:1971}. Here we are interested in the geometrically more accessible bundle gerbes, which were introduced by Murray \cite{Murray:9407015}, see
also \cite{Murray:2007ps}, and, in a more restricted form, by Hitchin and Chatterjee~\cite{Hitchin:1999fh,Chatterjee:1998}. We first give the very general
definition of principal $2$-bundles as studied e.g.\ in
\cite{Baez:0801.3843,Nikolaus:1207ab,nikolaus1207}, and then specialise to bundle gerbes.
\begin{definition}\label{def:principal_2_bundles}
 Given a strict Lie 2-group $\CCG$ and a surjective submersion $\sigma^Y:Y\thra M$, a \uline{principal 2-bundle $(\Phi,\sigma^Y)$ with structure 2-group $\CCG$ subordinate to $\sigma^Y$} is a normalised pseudofunctor $\Phi$ from the \v Cech 2-groupoid $\check \CCC (Y\thra M)$ to $\sB\CCG$. We also refer to these bundles simply as \uline{principal $\CCG$-bundles}. An \uline{isomorphism of principal 2-bundles} is a pseudonatural isomorphism between the underlying pseudofunctors. 
\end{definition}
Ordinary principal bundles with structure Lie group $\sG$ are principal 2-bundles with structure 2-group $\sG\rightrightarrows \sG$. We are particularly interested in principal 2-bundles with structure group $\sB\sU(1)$.
\begin{definition}
 A \uline{Hitchin-Chatterjee gerbe subordinate to a surjective
   submersion} $\sigma^Y:Y\thra M$ is a principal $\sB\sU(1)$-bundle
 $(\Phi,\sigma^Y)$. It is  \uline{local} if $Y\thra M$ is an
 ordinary open cover $\frU= \bigsqcup_a\, U_a \thra M$.
\end{definition}
Let us work through the details. A normalised pseudofunctor $\Phi$ between $\check \CCC(Y\thra M)$ and $\sB\sB\sU(1)$ contains a trivial function $\Phi_0:Y\rightarrow *$ together with a trivial functor $\Phi_1:(Y^{[2]}\rightrightarrows Y)\rightarrow \sB\sU(1)=(\sU(1)\rightrightarrows *)$. The only non-trivial information is contained in $\Phi_2:Y^{[3]}\rightarrow \sU(1)$ with
\begin{equation}
 \Phi_2^{(y_1,y_3,y_4)}\,
 \Phi_2^{(y_1,y_2,y_3)}=\Phi_2^{(y_1,y_2,y_4)}\, \Phi_2^{(y_2,y_3,y_4)}
\end{equation}
for $(y_1,y_2,y_3,y_4)\in Y^{[4]}$. In other words, $\Phi_2$ forms a
$\underline{\sU(1)}_M$-valued \v Cech 2-cocycle over $Y$. If we
imagine a trivial principal $\sU(1)$-bundle over $Y^{[2]}$ with fibres
$P_{(y_1,y_2)}$ over $(y_1,y_2)\in Y^{[2]}$, then $\Phi_2$ provides an
isomorphism $\Phi_2^{(y_1,y_2,y_3)}:P_{(y_1,y_2)}\otimes
P_{(y_2,y_3)}\to P_{(y_1,y_3)}$.

Isomorphisms of Hitchin-Chatterjee gerbes given by pseudonatural
transformations are \v Cech $2$-coboundaries, and thus local
Hitchin-Chatterjee gerbes up to isomorphism form \v Cech cohomology classes in $H^2(M,\underline{\sU(1)}_M)\cong H^3(M,\RZ_M)$. This gives rise to the characteristic class of the gerbe.
\begin{definition}\label{def:DD_class_local_HC_gerbe}
 The \uline{Dixmier-Douady class of a local Hitchin-Chatterjee gerbe} $(\Phi,\sigma^{\frU})$ is the cohomology class $[\Phi_2]\in H^2(M,\underline{\sU(1)}_M)\cong H^3(M,\RZ_M)$.
\end{definition}
Inversely, any Dixmier-Douady class defines a non-empty isomorphism class of local Hitchin-Chatterjee gerbes. We will generalise the notion of Dixmier-Douady class below.

We now follow \cite{Murray:9407015} and lift the \v Cech cocycle $\Phi_2$ to an isomorphism of potentially non-trivial $\sU(1)$-bundles over $Y^{[2]}$.
\begin{definition}\label{def:u(1)-bundle_gerbe}
 A \uline{$\sU(1)$-bundle gerbe $(P,\sigma^Y)$ subordinate to a
   surjective submersion} $\sigma^Y: Y\thra M$ is a principal $\sU(1)$-bundle $P$ over $Y^{[2]}$ together with a smooth isomorphism $\mu$ of $\sU(1)$-bundles, called the \uline{bundle gerbe multiplication},
 \begin{equation}
  \mu_{(y_1,y_2,y_3)}:P_{(y_1,y_2)}\otimes
  P_{(y_2,y_3)}\xrightarrow{ \ \cong \ } P_{(y_1,y_3)} \efor (y_1,y_2,y_3)\in Y^{[3]}~,
 \end{equation}
 which is associative in the sense that the diagram 
  \begin{equation}
  \label{eq:assoc_bgrb}
  \begin{tikzpicture}[baseline=(current  bounding  box.center)]
  \node (P12P23P34) at(0,0) {$P_{(y_1,y_2)} \otimes P_{(y_2,y_3)} \otimes P_{(y_3,y_4)}$};
  \node(P13P34) at (7,0) {$P_{(y_1,y_3)} \otimes P_{(y_3,y_4)}$};
  \node(P12P24) at (0,-2.5) {$P_{(y_1,y_2)} \otimes P_{(y_2,y_4)}$};
  \node(P14) at (7,-2.5) {$P_{(y_1,y_4)}$};  
  
  \draw[->](P12P23P34)--(P13P34) node[pos=.5,above]{\scriptsize{$\mu_{(y_1,y_2,y_3)} \otimes \mathbbm{1}$}};
  \draw[->](P12P23P34)--(P12P24) node[pos=.5,left]{\scriptsize{$\mathbbm{1} \otimes \mu_{(y_2,y_3,y_4)}$}};
  \draw[->](P13P34)--(P14) node[pos=.5,right]{\scriptsize{$\mu_{(y_1,y_3,y_4)}$}};
  \draw[->](P12P24)--(P14) node[pos=.5,below]{\scriptsize{$\mu_{(y_1,y_2,y_4)}$}};
  \end{tikzpicture}
  \end{equation}
  commutes. A \uline{local $\sU(1)$-bundle gerbe} is a $\sU(1)$-bundle gerbe subordinate to a surjective submersion given by an open cover.
\end{definition}

Let us briefly recall a few standard constructions. Given a map $f\in
\shom_{\CatMfd}(N,M)$, there is the pullback of a surjective
submersion $\sigma^Y: Y\thra M$,
\begin{equation}
 \xymatrixcolsep{5pc}
\myxymatrix{
 f^*(Y) \ar@{->>}[d] \ar@{->}[r]^{\hat f}  & Y \ar@{->>}[d] \\
 N \ar@{->}[r]_{f}  & M
}
\end{equation}
This induces a map $\hat f^{[2]}:f^*(Y)^{[2]}\rightarrow Y^{[2]}$ and
we can define the \emph{pullback} of a bundle gerbe $(P,\sigma^Y)$
along $f$ as $(\hat f^{[2]*} P,\sigma^{f^*(Y)})$. Every principal
bundle $P$ has a dual bundle $P^*$ with inverse transition functions; therefore every bundle gerbe $(P,\sigma^Y)$ has a \emph{dual} $(P,\sigma^Y)^*=(P^*,\sigma^Y)$. There is a \emph{product} between bundle gerbes $(P,\sigma^Y)$ and $(Q,\sigma^X)$ over a manifold $M$, $(P,\sigma^Y)\otimes (Q,\sigma^X)=(P\otimes Q,Y\times_M X)$, where $P\otimes Q$ is the product of the pullbacks of $P$ and $Q$ to $(Y\times_M X)^{[2]}=Y^{[2]}\times_M X^{[2]}$.

Given a sufficiently fine open cover $\frU= \bigsqcup_a\, U_a
\twoheadrightarrow M$, there are sections $s_a:U_a \rightarrow Y$ such
that one has maps
\begin{equation}
 p_{ab}:U_{ab}\longrightarrow P
\end{equation}
with $p_{ab}(u)\in P_{(s_a(u),s_b(u))}$ for all $u\in U_{ab}$. Using the bundle gerbe multiplication, they define an element $[g]\in H^2(M,\underline{\sU(1)}_M)$ by
\begin{equation}\label{eq:general_DD}
 p_{ab}(u)\, p_{bc}(u)=p_{ac}(u)\, g_{abc}(u)~,
\end{equation}
for all $u\in U_{abc}$, which is independent of the choice of cover and sections, cf.\ \cite{Carey:1997xm}. The following definition then extends Definition~\ref{def:DD_class_local_HC_gerbe}.
\begin{definition}
 The \uline{Dixmier-Douady class of a $\sU(1)$-bundle gerbe} $(P,\sigma^Y)$ is the \v Cech cohomology class of the 2-cocycle $g$ defined in \eqref{eq:general_DD}.
\end{definition}

Let us now come to triviality and stable isomorphism.
\begin{definition}\label{def:gerbes_trivial}
 A $\sU(1)$-bundle gerbe $(P,\sigma^Y)$ is  \uline{trivial} if
 there is a principal $\sU(1)$-bundle $R\to Y$ such that
 $P\cong\check \delta R$.\footnote{Here
   $\check \delta$ is defined in Section~\ref{ssec:Cech_Deligne}.} We say that $R$ is a \uline{trivialisation} of $(P,\sigma^Y)$. 
 
 Two $\sU(1)$-bundle gerbes $(P,\sigma^Y)$ and $(Q,\sigma^X)$ are  \uline{stably isomorphic} if $(P,\sigma^Y)^*\otimes (Q,\sigma^X)$ is trivial.
\end{definition}
We then have the following result \cite{Murray:9407015}.
\begin{proposition}
 A $\sU(1)$-bundle gerbe is trivial if and only if its Dixmier-Douady
 class is trivial. Two $\sU(1)$-bundle gerbes are stably isomorphic if
 and only if their Dixmier-Douady classes agree.
\end{proposition}

We can now compare $\sU(1)$-bundle gerbes with local Hitchin-Chatterjee gerbes, cf.\ \cite{Murray:9908135}.
\begin{proposition}\label{prop:gerb_isomorphism_HC}
 Every $\sU(1)$-bundle gerbe is stably isomorphic to a local Hitchin-Chat\-terjee gerbe.
\end{proposition}

Finally, let us give an alternative definition in terms of central groupoid extensions that will prove to be useful. This definition readily extends to gerbes over differentiable stacks, cf.\ \cite{Behrend:0605694}.
\begin{definition}[\cite{MR1103911}]
 Let $\CCG=(\CCG_1\rightrightarrows \CCG_0)$ be a Lie groupoid. A
 \uline{$\sU(1)$-central Lie groupoid} \uline{extension of} $\CCG$
 consists of a Lie groupoid $\CCH=(\CCH_1\rightrightarrows \CCG_0)$
 together with a morphism $\phi\in\shom_{\CatBibun}(\CCH,\CCG)$ and a left $\sU(1)$-action on $\CCH_1$, making $\phi_1:\CCH_1\rightarrow \CCG_1$ a (left-) principal $\sU(1)$-bundle, such that $(s_1 g_1)\circ(s_2 g_2)=(s_1s_2)(g_1\circ g_2)$ for all $s_{1},s_{2}\in \sU(1)$ and $(g_1,g_2)\in \CCH_1\times_{\CCG_0}\CCH_1$.
\end{definition}
It is then not too difficult to obtain the following result, cf.\ e.g.\ \cite{Ginot:2008fia}.
\begin{proposition}
 A {$\sU(1)$-bundle gerbe over a manifold} $M$ is a $\sU(1)$-central
 Lie group\-oid extension $\CCH$ of the \v Cech groupoid $\check
 \CCC(Y\thra M)=(Y^{[2]}\rightrightarrows Y)$ of  a surjective
 submersion $\sigma^Y: Y\thra M$.
\end{proposition}
Explicitly, there is a sequence of Lie groupoids for a $\sU(1)$-bundle gerbe $(P,\sigma^Y)$ given by
\begin{equation}
 \myxymatrix{
    \sU(1)\ar@<1.5pt>[d] \ar@<-1.5pt>[d]_{}="g1" & & P \ar@<1.5pt>[d] \ar@<-1.5pt>[d]_{}="g2" & & Y^{[2]} \ar@<1.5pt>[d] \ar@<-1.5pt>[d]_{}="g3"& & M \ar@<1.5pt>[d] \ar@<-1.5pt>[d]_{}="g4"\\
  {*} & & Y & & Y & & M   
  \ar@{->} "g1"+<3ex,0.0ex>;"g2"+<-2.5ex,0.0ex>
  \ar@{->} "g2"+<3ex,0.0ex>;"g3"+<-2.5ex,0.0ex>
  \ar@{->}^{\cong} "g3"+<3ex,0.0ex>;"g4"+<-2.5ex,0.0ex>
 }
\end{equation}
where the last map is the bibundle or Morita equivalence \eqref{eq:moritaY}. In this language, stable isomorphisms between
$\sU(1)$-bundle gerbes correspond to bibundle equivalences between
$\sU(1)$-central Lie groupoid extensions.

\subsection{Connective structures}\label{ssec:connective_structure}

From our discussion in Section~\ref{ssec:Cech_Deligne} (cf.\
Remark~\ref{rem:higherDeligne}), we glean that the appropriate notion
of a local Hitchin-Chatterjee gerbe with connection should be
described by a Deligne 2-cocycle in degree 2. Let us now generalise this observation to $\sU(1)$-bundle gerbes following~\cite{Murray:9407015}. This is done by endowing the principal bundle of the $\sU(1)$-bundle gerbe with a connection.
\begin{definition}
 Let $(P,\sigma^Y)$ be a $\sU(1)$-bundle gerbe. A \uline{bundle gerbe
   connection} is a connection $\nabla$
 on $P$ that respects the
 smooth isomorphism of principal bundles over $Y^{[3]}$.
\end{definition}
An immediate consequence of this definition is $\check\delta F_\nabla=0$. We can now use repeatedly the fact that the coboundary operator $\check
\delta:\Omega^p(Y^{[k]})\rightarrow \Omega^p(Y^{[k+1]})$ induces the
exact sequence \eqref{eq:fundamental_sequence}. First of all, there
is a 2-form $B\in \di\, \Omega^2(Y)$ such that $\check\delta B=F_\nabla$,
which is called a \emph{curving} of $\nabla$. The pair $(\nabla,B)$ is
also called a \emph{connective structure}. Given a curving, we note
that $\check\delta\, \dd B=\dd \, \check\delta B =\dd F_\nabla =0$ and
therefore $\dd B=(\sigma^Y)^* H$ for some closed 3-form $H\in
\di\, \Omega^3(M)$, which we call the \emph{curvature} of
$(\nabla,B)$. The de Rham class $[\frac{1}{2\pi\, \di} \, H]\in
H^3_{\rm
  dR}(M)$ is integral as it is the image of the Dixmier-Douady class of the $\sU(1)$-bundle gerbe $(P,\sigma^Y)$ in real cohomology. 

\begin{remark}
All previously introduced notions such as pullback, dual and product readily generalise to $\sU(1)$-bundle gerbes with connective structure.
\end{remark}

We can now define isomorphism classes of $\sU(1)$-bundle gerbes with connective structure, extending Definition~\ref{def:gerbes_trivial}.
\begin{definition}
 A $\sU(1)$-bundle gerbe with connective structure $(P,\sigma^Y,A,B)$
 is  \uline{trivial} if there is a principal $\sU(1)$-bundle
 with connection $(R,a)$ over $Y$ such that $P\cong \check\delta R$
 and $A=\check\delta a$. Then $(R,a)$ is a \uline{trivialisation} of $(P,\sigma^Y,A,B)$. 
 
 Two $\sU(1)$-bundle gerbes with connective structures $(P_1,\sigma^{Y_1},A_1,B_1)$ and $(P_2,\sigma^{Y_2},A_2,$ $B_2)$ are  \uline{stably isomorphic} if $(P_1,\sigma^{Y_1},A_1,B_1)^*\otimes (P_2,\sigma^{Y_2},A_2,B_2)$ is trivial.
\end{definition}

Let us look at a connective structure for a local $\sU(1)$-bundle gerbe $(P,\sigma^\frU)$ over $M$, where $\frU=\bigsqcup_a\, U_a$ is an open cover of $M$, in some more detail. The $\sU(1)$-bundle gerbe is given in terms of principal bundles $P_{ab}$ over $U_{ab}$ together with the bundle gerbe multiplication 
\begin{equation}
 \mu_{abc}:P_{ab}\otimes P_{bc}\longrightarrow P_{ac}
\end{equation}
over $U_{abc}$. Let us now assume that $\frU$ is a
good cover with all double intersections $U_{ab}$ contractible. Then
the principal $\sU(1)$-bundles $P_{ab}$ are trivial, and the bundle gerbe multiplication $\mu$ can be encoded in functions $g_{abc}:U_{abc}\rightarrow \sU(1)$. Associativity of the bundle gerbe multiplication amounts to
\begin{equation}
 g_{acd}\, g_{abc}=g_{abd}\, g_{bcd}~.
\end{equation}
A bundle gerbe connection is given by 1-forms $A_{ab}\in \di\, \Omega^1(U_{ab})$, where the compatibility with $\mu$ reads as
\begin{equation}
  (\check\delta A)_{abc}=A_{bc} -A_{ac}+ A_{ab} = -
  \dd \log g_{abc}
\end{equation}
over $U_{abc}$. Since $\check\delta\, \dd A=\dd\,\check\delta A =0$,
there are 2-forms $B_a \in \di\, \Omega^2(U_a)$ such that
\begin{equation}
  B_{a} - B_{b} = \dd A_{ab}
\end{equation}
over $U_{ab}$. Again $(\check\delta \, \dd B)_{ab} = (\dd \, \check\delta
B)_{ab}=0$, and therefore there is a global 3-form $H\in \di\, \Omega^3(M)$ such that 
\begin{equation}
 H=\dd B_a
\end{equation}
over $U_a$.

Note that $(g_{abc},A_{ab},B_a)$ yields a Deligne cocycle in
$H^2(M,\CD^\bullet(2))$ as defined in
Section~\ref{ssec:Cech_Deligne}. The image of $\pr_\CD$ recovers the
Dixmier-Douady class and the image of ${\rm curv}$ agrees here with
$H$, the curvature 3-form of the connective structure on the $\sU(1)$-bundle gerbe.

Since any $\sU(1)$-bundle gerbe is stably isomorphic to one subordinate to a good open cover, we have the following result.
\begin{proposition}
 Stable isomorphism classes of $\sU(1)$-bundle gerbes with connective structures are classified by the Deligne cohomology group $H^2(M,\CD^\bullet(2))$.
\end{proposition}

\subsection{Examples}

Let us briefly list some interesting examples of $\sU(1)$-bundle
gerbes. 

\begin{example}\label{ex:very_trivial_gerbe}
The most trivial example of a trivial $\sU(1)$-bundle gerbe is obtained from the trivial principal $\sU(1)$-bundle $P$ over a manifold $M$. We cover $M$ by itself, $Y=M$, with $\sigma^Y=\id_M$ and $Y^{[2]}\cong M$. The bundle gerbe multiplication over $Y^{[3]}\cong M$ is the identity: $P\otimes P\cong P$. 
\end{example}

\begin{example}
A more general description of trivial $\sU(1)$-bundle gerbes $(P,\sigma^Y)$ over a manifold $M$ arises from a surjective
submersion $\sigma^Y: Y\twoheadrightarrow M$ and a principal $\sU(1)$-bundle
$Q\rightarrow Y$ by setting $P_{(y_1,y_2)}=Q^*_{y_1}\otimes Q_{y_2}$.
\end{example}

\begin{example}\label{ex:BGrb_from_fus_LBdl}
  While surjective submersions based on (good) open covers might be
  most familiar, it is very useful to consider also the \emph{path
    fibration} $\dpar:\CP M\thra M$, see
  Section~\ref{ssec:gq_outline}, statement
  \eqref{it:tautological_line_bundle}. $\sU(1)$-bundle gerbes
  subordinate to the surjective submersion $\dpar$ are defined by
  principal $\sU(1)$-bundles $P$ according to the diagram
  \begin{equation}
  \label{eq:path_space_BGrb}
    \myxymatrix{
	  P \ar@{->}[d] & \\
	  \Omega M  \ar@<1.5pt>[r] \ar@<-1.5pt>[r] & \CP M \ar@{->}[d]^{\dpar} \\
	    & M
  }
  \end{equation}
Such concrete examples of $\sU(1)$-bundle gerbes will frequently appear in the following. 
  
  A principal $\sU(1)$-bundle over the free loop space $\CL M$ of $M$
  with a product structure as above, given for pairs of concatenable
  loops in this case, is called a \emph{fusion bundle} over $\CL
  M$. It is shown in~\cite{Waldorf:0911.3212,Waldorf:2010aa} that
  the category of fusion bundles on $\CL M$ is equivalent to the
  category of $\sU(1)$-bundle gerbes on $M$. This example will be
  central to the discussion in Section~\ref{sect:Transgression}.
\end{example}

\begin{example}
Consider a central extension of groups $\sU(1)\rightarrow
\widehat\sG\rightarrow \sG$. The obstruction to lifting a principal
$\sG$-bundle $\pi: Q\thra M$ to a principal $\widehat \sG$-bundle is a
class in $H^3(M,\RZ_M)$, and we can construct the corresponding
$\sU(1)$-bundle gerbe explicitly \cite{Murray:9407015}. There is a map
$g:Q^{[2]}\rightarrow \sG$ with $y_1\, g(y_1,y_2)=y_2$. Pulling back the $\sU(1)$-bundle $\widehat \sG\rightarrow \sG$ along $g$ yields the \emph{lifting bundle gerbe} $(g^*\widehat \sG,\pi)$.

Of particular interest is the lifting bundle gerbe for the central
extension $\sU(1)\rightarrow \sU(\CH)\rightarrow \sPU(\CH)$ associated
to the unitary group of an infinite-dimensional separable Hilbert
space $\CH$. 
\label{ex:centralext}\end{example}

\begin{example}\label{ex:central_loop_extension}
We can combine Examples~\ref{ex:BGrb_from_fus_LBdl} and~\ref{ex:centralext} in the following way. For a compact 1-connected simple Lie group $\sG$, the central extensions of the based loop group~\cite{Pressley:1988aa},
  \begin{equation}
	  1 \longrightarrow \sU(1) \longrightarrow \widehat{\Omega_k \sG} \longrightarrow \Omega \sG \longrightarrow \unit~,
  \end{equation}
with $k\in H^3(\sG,\RZ_\sG)\cong\RZ$,  provide principal $\sU(1)$-bundles over
$\Omega \sG$ which are compatible with the multiplication on $\Omega
\sG$. Hence this yields the $\sU(1)$-bundle gerbes
$(P,\sigma^Y)=(\widehat{\Omega_k\sG},\partial)$ with $Y=\CP\sG$.
\end{example}

\begin{example}\label{ex:2-grp_to_Bgrb}
  Let $\CCG = (\CCG_1 \rightrightarrows \CCG_0)$ be a {strict} Lie
  2-group with $\shom_{\CCG}(\unit, \unit) \cong \sU(1)$.
  Then there is a $\sU(1)$-bundle gerbe given by
  \begin{equation}
    \myxymatrix{
	  \CCG_1 \ar@{->}[d]_-{\sfs \times \sft} & \\
	  \CCG_0{\times_{h_0\CCG}}\CCG_0 \ar@<1.5pt>[r] \ar@<-1.5pt>[r] & \CCG_0 \ar@{->}[d]^{q} \\
	    & h_0 \CCG
  }
  \end{equation}
  where $h_0 $ assigns to a small category $\CCG$ the set of isomorphism classes of objects
  in $\CCG$, and $q:\CCG_0\thra h_0\CCG$ projects an object of $\CCG$ to
  its equivalence class.
  The bundle gerbe multiplication is given by vertical composition in the associated 2-category $\sB \CCG = (\CCG_1 \rightrightarrows \CCG_0 \rightrightarrows *)$, or equivalently the group product in $\CCG_1$.
  This example is found implicitly in~\cite{Waldorf:1201.5052}.
  We will come back to this particular class of $\sU(1)$-bundle gerbes in Section~\ref{sect:action_of_2-grps_on_their_BGrbs}.
  
In particular, applying this example to the string 2-group model from~\cite{Baez:2005sn}, one obtains the $\sU(1)$-bundle gerbes from Example~\ref{ex:BGrb_from_fus_LBdl}.
\end{example}

\section{The 2-category of line bundle gerbes}\label{sect:2-cat_of_BGrbs}

The central objects in this paper will be line bundle gerbes, which
will take the role of categorified prequantum line bundles. These are
associated to $\sU(1)$-bundle gerbes in the same way that line bundles
are associated to principal $\sU(1)$-bundles. Line bundle gerbes
define a 2-category, which we describe in great detail in this section.

\subsection{Line bundle gerbes}

A vector bundle $E$ associated to a principal $\sG$-bundle
$P$ subordinate to a surjective submersion $\sigma^Y:Y\thra M$ can be defined by composing the functor defining $P$, $\Phi^P:\check \CCC(Y\thra M)\rightarrow \sB\sG$, with the functor $\Psi^\varrho:\sB\sG\rightarrow \CatVect$ defining the representation $\varrho$ of interest. Analogously, we can define 2-vector bundles associated to principal 2-bundles, using 2-representation pseudofunctors as defined in Section~\ref{ssec:lie2grp_reps}. 
\begin{definition}\label{def:2-vector_bundle}
 Let $\CCG$ be a strict Lie 2-group and $\Psi:\sB\CCG\rightarrow
 \CatBimod$ a 2-representation of $\CCG$. Given a principal
 $\CCG$-bundle subordinate to a surjective submersion $\sigma^Y: Y\thra M$
 defined by a normalised pseudofunctor $\Phi:\check \CCC(Y\thra M)\rightarrow \sB\CCG$, the \uline{associated 2-vector bundle} corresponding to $\Psi$ is the composition $\Psi\circ \Phi$. 
\end{definition}

Let us specialise to the case $\CCG=\sB\sU(1)$ with 2-representation
$\Psi^{\varrho_{\rm f}}$ as given in Example~\ref{rem:def_rho_f}. A principal
$\sB\sU(1)$-bundle is simply a Hitchin-Chatterjee gerbe
$(P,\sigma^Y)$, and composing $\Phi$ with the 2-representation
$\Psi^{\varrho_{\rm f}}$ maps $Y^{[2]}$ to the $\FC$-bimodule
$\FC$. We thus obtain a trivial line bundle over $Y^{[2]}$
corresponding to the associated line bundle to $P$ with respect to the
fundamental representation $\varrho_{\rm f}$ of $\sU(1)$. Generalising this to arbitrary line bundles over $Y^{[2]}$ leads
to the more evident definition of line bundle gerbes associated to $\sU(1)$-bundle gerbes.
\begin{definition}\label{def:line_bundle_gerbes}
Given a $\sU(1)$-bundle gerbe $(P,\sigma^Y)$, its associated
\uline{line bundle gerbe} $(L,\sigma^Y)$ consists of the associated line
bundle $L$ to $P$ over $Y^{[2]}$ with respect to the fundamental
representation $\varrho_{\rm f}$ of $\sU(1)$ together with the induced \uline{bundle
  gerbe multiplication}\footnote{With a slight abuse of notation, we use the same symbol for the bundle
  gerbe multiplication on $\sU(1)$-bundle gerbes and their associated line bundle gerbes.}
 \begin{equation}
  \mu_{(y_1,y_2,y_3)}: L_{(y_1,y_2)}\otimes L_{(y_2,y_3)}
  \longrightarrow L_{(y_1,y_3)} \efor (y_1,y_2,y_3)\in Y^{[3]}~.
 \end{equation}
\end{definition}

\begin{remark}
This notion of a line bundle gerbe indeed extends that as a normalised pseudofunctor from the \v Cech 2-groupoid $\check \CCC(Y\thra M)$ to the 2-category of $\FC$-bimodules $\CatBimod^\FC$ constructed above Definition \ref{def:line_bundle_gerbes}.
\end{remark}

Next we add two additional structures to these line bundle gerbes:
hermitean metrics and connections. Unless stated otherwise, pullbacks
of vector bundles are endowed with the pullback metric and connection in the
following.

\begin{definition}
	\begin{enumerate}
		\item A \uline{hermitean line bundle gerbe} is a line
                  bundle gerbe $(L,\sigma^Y)$ together with a hermitean metric $h$ on $L$ and a multiplication $\mu$ which is an isometric isomorphism. 
		\item A connective structure on a $\sU(1)$-bundle
                  gerbe $(P,\sigma^Y)$ induces a 
                  \uline{connective structure} on
                  the associated line bundle gerbe $(L,\sigma^Y)$ consisting
                  of a connection $\nabla^L$ on the line bundle $L$ such
                  that the multiplication $\mu$ is parallel:
		\begin{equation}
		0 = (\nabla^L \mu)_{(y_1,y_2,y_3)}
		= \nabla^L\circ \mu_{(y_1,y_2,y_3)} - \mu_{(y_1,y_2,y_3)} \circ \big(
                \nabla^L\otimes \unit + \unit \otimes
                \nabla^L \big)~, 
		\end{equation}
		together with a curving 2-form $B$, cf.\ Section~\ref{ssec:connective_structure}.
		If $(L,\sigma^Y)$ is hermitean, we call a connective structure \uline{hermitean} if $\nabla^L h = 0$.
	\end{enumerate}
\end{definition}

In the following all line bundle gerbes will be hermitean and endowed
with a connective structure; we shall use the notation $\CL =
(L,\mu,h,\nabla^L,B,\sigma^Y)$ and often omit the adjective `line'
when there is no risk of confusion.

For any such  bundle gerbe $\CL$,
we define its \emph{dual  bundle gerbe} as 
\begin{equation}
	\CL^* = \big(L^*,\mu^{-{\rm
            t}},h^{-1},\nabla^{L^*},-B,\sigma^Y \big)~.
\end{equation}
The connection $\nabla^{L^*}$ is the connection induced on the dual
line bundle $L^*$ which is defined via
\begin{equation}
\big( \nabla^{L^*}_X\beta \big) (\varepsilon) = \pounds_X
\big(\beta(\varepsilon)\big) - \beta \big( \nabla^L_X \varepsilon) \ ,
\end{equation}
where $X\in \frX(Y^{[2]})$, $\beta$ and $\eps$ are sections of $L^*$
and $L$, respectively, and $\pounds$ denotes the Lie derivative. For
any morphism $\phi: E \rightarrow F$ of vector bundles, we denote by
$\phi^{\rm t}: F^* \rightarrow E^*$ its \emph{transpose}, in contrast
to the (Hilbert space) adjoint $\phi^*: F \rightarrow E$ which depends on the choice of metric.

The \emph{tensor product} of two  bundle gerbes $\CL_1 \otimes \CL_2$ is defined in the obvious manner as
\begin{equation}
\label{eq:def_otimes_objs}
\CL_1\otimes \CL_2:=\big( L_{1}\otimes L_{2},\, \mu_{1}\otimes
\mu_{2},\, h_{1}\otimes h_{2},\, \nabla^{L_1 \otimes L_2},\, B_{1} + B_{2} ,\, \sigma^{Y_1{\times_M}Y_2} \big)~,
\end{equation}
where $\nabla^{L_1 \otimes L_2} =
\nabla^{L_1}\otimes\unit+\unit\otimes\nabla^{L_2}$.

The \emph{standard trivial bundle gerbe of curving $\rho$} is the line bundle gerbe associated to the $\sU(1)$-bundle gerbe in Example~\ref{ex:very_trivial_gerbe}. It is described by the data
\begin{equation}\label{eq:trivial_line_bundle_gerbe}
\CI_\rho = \big( M \times \FC,\, \mu_\CI,\, h_\CI,\, \dd,\, \rho,\, \id_M \big)~,
\end{equation}
where $\mu_\CI$ is multiplication in the fibres of the trivial line
bundle $M \times \FC$, and the trivial metric, connection and surjective submersion are used.
There is, however, freedom in the choice of curving since for the
trivial surjective submersion $\id_M: M \rightarrow M$ every 
2-form $\rho \in \di\, \Omega^2(M)$ satisfies $\check{\delta} \rho =
\rho - \rho = 0$.

A hermitean  bundle gerbe with connective structure $\CL =
(L,\mu,h,\nabla^L,B,\sigma^Y)$ is \emph{trivial} if there is a hermitean
line bundle with connection subordinate to $\sigma^Y$ such that $\CL$ is isomorphic to the obvious
pullback. Two  bundle gerbes $\CL_{i} =
(L_{i},\mu_{i},h_{i},\nabla^{L_i},B_i,\sigma^{Y_i})$, $i=1,2$, are
\emph{stably isomorphic} if the tensor product $\CL_1^*\otimes \CL_2$
is trivial. Clearly two  bundle gerbes associated to stably
isomorphic $\sU(1)$-bundle gerbes are stably isomorphic, as stable isomorphisms factor through the representation included.

\subsection{1-morphisms}

A \emph{morphism} from a line bundle gerbe $(L,\sigma^Y)$ to a line bundle gerbe $(J,\sigma^Z)$ is a pair of maps $(f,g)$, where $f:Y\rightarrow Z$ is a map covering the surjective submersions, and $g:L\rightarrow J$ is a bundle map compatible with the bundle gerbe multiplication and covering the induced map $f^{[2]}:Y^{[2]}\rightarrow Z^{[2]}$. By restricting these morphisms to isometric connection
  preserving morphisms, we obtain a {category of hermitean 
  bundle gerbes} on $M$ with hermitean connective structure. However, this definition is not satisfactory, as the isomorphism classes of objects would not be in bijective correspondence with the Deligne cohomology group $H^2(M,\CD^\bullet(2))$: Whereas isometric connection preserving isomorphisms induce stable isomorphisms, the converse is not generally true.

On the other hand, \emph{morphisms of associated 2-vector bundles} from Definition~\ref{def:2-vector_bundle} are given by pseudonatural transformations between the underlying pseudofunctors. These include the morphisms above in the case that the line bundles $L$ and $J$ are trivial, but further allow for a larger class of morphisms. 

In~\cite{Waldorf:2007aa,Waldorf:0702652}, Waldorf extended the stable isomorphisms to a suitable notion of morphisms of bundle gerbes.

\begin{definition}[\cite{Waldorf:2007aa}]
	\label{def:lbgr_morph} 
	Given two  bundle gerbes $\CL_i =
        (L_i,\mu_i,h_i,\nabla^{L_i}, B_i,\sigma^{Y_i})$, $i=1,2$, a
        \uline{1-morphism from $\CL_1$ to $\CL_2$} is a quintuple $(E,\alpha,g,\nabla^E,\zeta^Z)$ of the following data:
	$\zeta^Z: Z \thra Y_1{\times_M}Y_2$ is a surjective
        submersion with projections $\zeta_i^Z=\mathrm{pr}_{Y_i} \circ \zeta^Z :Z\to
        Y_i$ inducing maps $Z^{[2]}\rightarrow Y^{[2]}_i$, $E \rightarrow Z$ is a
        hermitean vector bundle with metric $g$, $\alpha$ is an isometric\footnote{This
          constrains the bundle metric $g$ in terms of $h_i$ and
          $\alpha$.} isomorphism over $Z^{[2]}$ of hermitean line
        bundles with connection such that\footnote{The fibres of the line
          bundles $L_{i,(z_1,z_2)}$ arise from pullbacks from
          $Y_i^{[2]}$ to $Y^{[2]}_1\times_M Y^{[2]}_2$ and further to
          $Z^{[2]}$ along the obvious maps.} 
	\begin{equation}
		\alpha_{(z_1,z_2)}:L_{1,(z_1,z_2)} \otimes E_{z_2} \longrightarrow  E_{z_1} \otimes L_{2,(z_1,z_2)}
	\end{equation}
	 is compatible with the bundle gerbe multiplications on
         $\CL_1$ and $\CL_2$ in the sense that the diagram
	\begin{equation}
	\label{eq:comp_hom_bgrb}
	\begin{tikzpicture}[baseline=(current  bounding  box.center)]
		\node (L12L23E3) at(0,0) {$L_{1,(z_1,z_2)}\otimes L_{1,(z_2,z_3)}\otimes E_{z_3}$};
		\node (L13E3) at(8,0) {$L_{1,(z_1,z_3)}\otimes E_{z_3}$};
		\node (L12E2L23) at(0,-2.5) {$L_{1,(z_1,z_2)}\otimes E_{z_2}\otimes L_{2,(z_2,z_3)}$};
		\node (E1L12L23) at(0,-5) {$E_{z_1} \otimes L_{2,(z_1,z_2)}\otimes L_{2,(z_2,z_3)}$};
		\node (L13E1) at(8,-5) {$E_{z_1} \otimes L_{2,(z_1,z_3)}$};
		
		\draw[->](L12L23E3)--(L13E3) node[pos=.5,above]{\scriptsize{$\mu_{1,(z_1,z_2,z_3)} \otimes \unit$}};
		\draw[->](L13E3)--(L13E1) node[pos=.5,right]{\scriptsize{$\alpha_{(z_1,z_3)}$}};
		\draw[->](L12L23E3)--(L12E2L23) node[pos=.5,left]{\scriptsize{$\unit\otimes \alpha_{(z_2,z_3)}$}};
		\draw[->](L12E2L23)--(E1L12L23) node[pos=.5,left]{\scriptsize{$\alpha_{(z_1,z_2)} \otimes \unit$}};
		\draw[->](E1L12L23)--(L13E1) node[pos=.5,below]{\scriptsize{$\unit \otimes \mu_{2,(z_1,z_2,z_3)}$}};
	\end{tikzpicture}
	\end{equation}
	commutes over $(z_1,z_2,z_3)\in Z^{[3]}$, and $\nabla^E$ is a
        hermitean connection on $E$ compatible with $\alpha$, that is,
        $\alpha$ is parallel with respect to the induced connection.
	
	Given 1-morphisms $(E,\alpha,g,\nabla^E,\zeta^Z): \CL_1 \rightarrow \CL_2$ and $(E',\alpha',g',\nabla^{E'},\zeta^{Z'}): \CL_2 \rightarrow \CL_3$, their composition $(E,\alpha,g,\nabla^E,\zeta^Z)\bullet(E',\alpha',g',\nabla^{E'},\zeta^{Z'})$ is given by
	\begin{equation}\label{eq:comp_1hom}
	\big( \pr_Z^*E \otimes \pr_{Z'}^*E'\, , \, (\fpmap{}{\pr}{[2]*}{Z} \alpha) \circ (\fpmap{}{\pr}{[2]*}{Z'} \alpha'\,) \, , \,  (\fpmap{}{\zeta}{Z}{~} \circ \pr_Z) {\times_{Y_2}} (\fpmap{}{\zeta}{Z'}{~} \circ \pr_{Z'}) \big)
	\end{equation}
	for the induced surjective submersion $Z\times_{Y_2} Z'\thra
        Y_1\times_M Y_3$, together with the corresponding pullback connection and metric.
\end{definition}
For brevity, we will often refer to such a morphism $(E,\alpha,g,\nabla^E,\zeta^Z)$ simply as $(E,\alpha)$. 

\begin{remark}
These 1-morphisms are generalisations of stable isomorphisms,
which are recovered in the case $Z=Y_1\times_M Y_2$ and the vector
bundle $E$ is of rank~$1$. Relaxing the former restriction renders the
composition of morphisms associative which is desirable for purely
geometric constructions, while relaxing the latter condition enables a
description of bundle gerbe modules by morphisms. We shall return to these points in some more detail later on.
\end{remark}

\begin{example}
	Consider local  bundle gerbes $\CL_1$ and $\CL_2$ both
        subordinate to an open cover $\frU$ of $M$. Then a 
        morphism $(E,\alpha): \CL_1\to\CL_2$ is a family $(E_a
        \rightarrow U_a)$ of hermitean vector bundles of rank $n$ with connections together with bundle isomorphisms
	\begin{equation}
		\alpha_{ab}: L_{1,ab} \otimes E_{b} \isom E_{a} \otimes L_{2,ab}
	\end{equation}
	over $U_{ab}$. If the cover is good, we can just as well use trivial bundles so that the transition maps collapse to $\sU(n)$-valued functions $\alpha_{ab}: U_{ab} \rightarrow \sU(n)$.
	The compatibility condition \eqref{eq:comp_hom_bgrb}
        translates to a $\sU(1)$-twisted version of the cocycle
        identity over triple overlaps given by
	\begin{equation}\label{eq:comp_2_vector_bdl}
		g_{2,abc}\, \alpha_{ab}\, \alpha_{bc} = \alpha_{ac}\, g_{1,abc}~.
	\end{equation}
	The same result is obtained if we consider $\CL_{i}$ as
        associated 2-vector bundles to a local Hitchin-Chatterjee
        gerbe, cf.\ Definition~\ref{def:2-vector_bundle}. Then pseudonatural transformations between the defining pseudofunctors are encoded by twisted hermitean vector bundles as described in \eqref{eq:comp_2_vector_bdl}.
		
	The morphisms induced by $\alpha_{ab}$ are parallel, but they do not live on the bundles $E_a$.
	Instead, the line bundles constituting the bundle gerbes come into play, such that the corresponding local expression for the compatibility of $\alpha$ with the respective connections 
	\begin{equation}
	 \nabla^{L_{i,ab}}=\dd+A_{i,ab} \eand \nabla^{E_a}=\dd+a_a
	\end{equation}
	reads as
	\begin{equation}
		(A_{1,ab} + a_b) = {\rm Ad}(\alpha_{ab}^{-1}) \circ (A_{2,ab} + a_a) + \alpha_{ab}^{-1}\, \dd\alpha_{ab}~,
	\end{equation}
	where the last term is shorthand for the Maurer-Cartan form on $\sU(n)$ at $\alpha_{ab}$. For $\CL_1\cong \CI_0$ trivial, these constitute the relations defining \emph{twisted vector bundles with connections}, as studied for instance in~\cite{Karoubi:1012.2512,Rogers:2011zc}, and the definition of morphism in~\cite{Rogers:2011zc} agrees with ours restricted to these twisted vector bundles with connection.
	
	On curvature 2-forms $F_{\nabla^{L_{i,ab}}}=\dd A_{i,ab}$ we obtain
	\begin{equation}
		(F_{\nabla^{L_{1,ab}}} + F_{\nabla^{E_b}}) = {\rm
                  Ad}(\alpha_{ab}^{-1}) \circ (F_{\nabla^{L_{2,ab}}} + F_{\nabla^{E_a}})
	\end{equation}
	over $U_{ab}$. Substituting the curvatures on $L_{i,ab}$ by the respective curvings, this is equivalent to
	\begin{equation}
		\big( F_{\nabla^{E_a}} - (B_{2,a} - B_{1,a})\unit \big) = {\rm Ad}(\alpha_{ab}^{-1}) \circ \big( F_{\nabla^{E_b}} - (B_{2,b} - B_{1,b})\, \unit \big)~,
	\end{equation}
	where $\unit$ denotes the central embedding of the generator of $\au(1)$ into $\au(n)$.	In particular, since the adjoint representation forgets the
        $\sU(1)$-twists, the homomorphism bundles $\sHom(E,F)$ descend,
        for any pair of morphisms $(E,\alpha), (F,\beta): \CL_1 \rightarrow \CL_2$, to a vector bundle on $M$.
	From the local expression above, it follows that the family $(F_{\nabla^{E_a}} - (B_{2,a} - B_{1,a})\unit)$ then induces a 2-form on $M$ with values in this descent bundle.
	The machinery of this section enables us to see that these
        statements do in fact hold true for 1-morphisms
        between  bundle gerbes in the general case.
\label{ex:local1mor}\end{example}

\begin{remark}\label{rem:twist_vec_bundles}
The twisted vector bundles mentioned in Example~\ref{ex:local1mor}, which
describe morphisms from the trivial bundle gerbe $\CI_0$ into
$\CL_2$, can also be regarded as associated 2-vector bundles to a
principal $\CCG$-bundle with structure 2-group $\CCG=(\sU(1)\ltimes
\sU(n)\rightrightarrows \sU(n))$ with respect to the 2-representation
$\Psi^{\varrho_{\rm f}}$ induced by the fundamental representation
$\varrho_{\rm f}$ of $\sU(n)$, cf.\
Example~\ref{ex:u(1)u(n)-rep}. This observation is relevant for a
discussion of constraints on the curvature of the hermitean vector
bundles underlying morphisms of  bundle gerbes. We shall return to
this point later on.
\end{remark}

\begin{remark}
	The morphism sets introduced in~\cite{Waldorf:2007aa,Waldorf:0702652} are very closely related to the notion of bundle gerbe modules from~\cite{Carey:2002xp,Bouwknegt:2001vu}.
	The main differences are the introduction of an additional
        surjective submersion
        in~\cite{Waldorf:0702652,Waldorf:2007aa}, on which the vector
        bundle $E$ is defined, and the inclusion of infinite-rank Hilbert bundles in \cite{Carey:2002xp,Bouwknegt:2001vu}.
	In the latter case the focus is on twisted K-theory,
        so that the categorical structures of bundle gerbes, such as the composition of morphisms, are not relevant there.
	In fact, we can already see that the composition of morphisms
        relies crucially on the flexibility in the choice of the
        surjective submersion $\zeta^Z: Z \thra Y_1\times_M Y_2$.
	
	In general, allowing for infinite rank seems to cure certain drawbacks of the morphisms as defined above, but at the cost of several features which are desirable for our approach to 2-Hilbert spaces.
	Throughout this article, we will make several remarks on the infinite-rank case.
\end{remark}

The identity morphism on a bundle gerbe $\CL=(L,\mu,h,\nabla^L,B,\sigma^Y)$ is given by
\beq
\id_\CL= \big(L,\beta,h,\nabla^L,\id_{Y^{[2]}}\big) \ ,
\eeq
with $\beta_{(y_1,y_2,y_3,y_4)}=\mu^{-1}_{(y_1,y_2,y_4)}\circ \mu_{(y_1,y_3,y_4)}$ over $(y_1,y_2,y_3,y_4)\in Y^{[4]}$.
Similarly to the case of line bundles, there exists a canonical isomorphism
\begin{equation}
\label{eq:delta_G}
\delta_{\CL} = \big(L,\alpha,h,\nabla^L, \id_{Y^{[2]}} \big)\ \in\ \sisom_{\CatLBGrb^\nabla(M)}(\CL^* \otimes \CL, \CI_0)~,
\end{equation}
where
\beq
\alpha_{(y_1,y_2,y_3,y_4)}=\big(\unit\otimes\delta_{L,(y_2,y_3)}\otimes \delta_{L,(y_2,y_4)}\big)\circ\big(\mu^{-1}_{(y_1,y_2,y_3)}\otimes\mu^{\rm t}_{(y_2,y_3,y_4)}\otimes \unit\big) \ .
\eeq
The morphism $\delta_L$ is the canonical isomorphism from $L^* \otimes L$ to the trivial line bundle, inducing the natural pairing $L^*_{(y_1,y_3)} \otimes L_{(y_2,y_4)} \rightarrow \FC$.

\begin{example}
	The morphisms $\CI_{\rho_1} \rightarrow \CI_{\rho_2}$ consist of a surjective submersion $Z \thra M$, together with a descent datum for a hermitean vector bundle with connection on $Z$.
	Since these form a stack, every morphism induces a unique hermitean vector bundle with connection on $M$.
	
	In particular, we can consider trivial gerbes on $M= \FR^3$,
        and choose $\rho_1 = 0$ and $\rho_2$ such that $\dd\rho_2 =
        -2\pi\,\di\, 
        \dd x \wedge \dd y \wedge \dd z = -2\pi\,\di\, {\rm
          vol}_{\FR^3}$; that is, $\CI_{\rho_2}$ provides a bundle
        gerbe with connection on $\FR^3$ whose curvature 3-form $H$ is
        $-2\pi\, \di$ times the volume form on $\FR^3$.
\end{example}

\begin{remark}
In our definition of 1-morphisms of bundle gerbes we do not impose a condition on
      the trace of the curvature of the vector bundle $E \rightarrow Z$, in contrast to the original
      definition of~\cite{Waldorf:2007aa,Waldorf:0702652}.
There 1-morphisms were defined to satisfy
      \begin{equation}
            \tr \big(F_{\nabla^E} - (\fpmap{}{\zeta}{Z*}{2} B_2 - \fpmap{}{\zeta}{Z*}{1}
            B_1)\,\unit \big) = 0 \ .
\label{eq:tr1morphisms}\end{equation}
      However, the 1-morphisms as we chose to define them can be recast in the formalism
      of \emph{bimodules} of bundle gerbes~\cite{Waldorf:2007aa}.
      In our convention, an $\CL_1$--$\CL_2$-bimodule of curvature $\rho \in
      \di\, \Omega^2(M)$ would be a morphism $(E,\alpha): \CL_1 \rightarrow \CL_2 \otimes
      \CI_\rho$ which satisfies the condition \eqref{eq:tr1morphisms} on the trace of $F_{\nabla^E}$ with the
      respective curvings of its source and target.
      If we start from an arbitrary 1-morphism $(E,\alpha): \CL_1 \rightarrow \CL_2$,
      one can show that the 2-form $\omega^E = \tr (F_{\nabla^E} - (\fpmap{}{\zeta}{Z*}{2} B_2
      - \fpmap{}{\zeta}{Z*}{1} B_1)\,\unit)$ satisfies $\check\delta \omega^E = 0$, which follows from the fact that $\alpha$ is parallel.
      Therefore $\omega^E = (\sigma^{Z})^* \rho$ for some $\rho \in \di\,
      \Omega^2(M)$, and
      $(E,\alpha)$ provides an $\CL_1$--$\CL_2$-bimodule of curvature $\rho$.
\end{remark}

\begin{definition}
\label{def:flat_isomps_and_trivialisations}
	An isomorphism $(E,\alpha,g,\nabla^E,\zeta^Z) \in
        \sisom_{\CatLBGrb^\nabla(M)}(\CL_1,\CL_2)$ is \uline{flat} if it satisfies
	\begin{equation}
		F_{\nabla^E} = (\fpmap{}{\zeta}{Z*}{2} B_2 - \fpmap{}{\zeta}{Z*}{1} B_1)\,\unit~,
	\end{equation}
	where $B_1$ and $B_2$ are the curvings of $\CL_1$ and $\CL_2$, respectively.%
	\footnote{
		This refined definition is the notion of \emph{isomorphism} of bundle gerbes used in e.g.~\cite{Waldorf:2007aa,Waldorf:0702652,Waldorf:2010aa}.
	}
	
	A \uline{trivialisation} of a bundle gerbe $\CL$ with connective structure is a flat isomorphism $\CL \rightarrow \CI_\rho$ for some $\rho \in \di\,\Omega^2(M)$.
	A bundle gerbe is \uline{trivial} if it admits a trivialisation in this sense.
\end{definition}

Flat isomorphism classes of bundle gerbes form a group under the tensor
product. We then have the following statements.

\begin{theorem}[\cite{Waldorf:2007aa}]
	\label{st:BGrbs_and_Deligne_coho}
	\begin{enumerate}
		\item Two bundle gerbes are flat isomorphic if and only if they define the same Deligne class.
		\item The assignment of the Dixmier-Douady class induces a group isomorphism from the group of flat isomorphism classes of bundle gerbes to $H^2(M,\CD^\bullet(2))$.
		\item A bundle gerbe $\CL$ is trivialisable if and
                  only if its Dixmier-Douady class vanishes, that is,
                  ${\rm dd}(\CL) = 0$.
		\item If $(T,\beta) \in
                  \sisom_{\CatLBGrb^\nabla(M)}(\CL,\CI_\rho)$ and
                  $(T',\beta'\,) \in
                  \sisom_{\CatLBGrb^\nabla(M)}(\CL,\CI_{\rho'})$ are
                  two trivialisations of $\CL$, then $\rho - \rho' \in
                  2\pi\, \di\, \Omega^2_{{\rm cl},\RZ}(M)$.
	\end{enumerate}
\label{thm:triv}\end{theorem}

For any trivialisation $(T,\beta): \CL \rightarrow \CI_\rho$ we have
$F_{\nabla^T} - (\fpmap{}{\zeta}{Z*}{2} \rho - \fpmap{}{\zeta}{Z*}{1}
B)\,\unit = 0$, where $F_{\nabla^T}$ defines an integral cohomology class.
Therefore $\fpmap{}{\zeta}{Z*}{2} \rho - \fpmap{}{\zeta}{Z*}{1} B \in
2\pi\, \di\, \Omega^2_{{\rm cl},\RZ}(Z)$.
This shows that only certain $2$-forms $\rho$ are admissible as
trivialisations of $\CL$, and they depend on the curving $B$ of $\CL$.
From this we can again deduce that if there exist two trivialisations
$(T,\beta):\CL\to\CI_{\rho}$ and $(T',\beta'\, ):\CL\to\CI_{\rho'}$
of $\CL$ over the same surjective submersion $\zeta^Z$,
then\footnote{Over different surjective submersions $\zeta^Z \neq \zeta^{Z'}$ we form the fibred product
  $Z\times_M Z'$ and pull the remaining data back to this space.}
\begin{equation}
\fpmap{}{\zeta}{Z*}{2} \rho - \fpmap{}{\zeta}{Z*}{2} \rho'
= \big( \fpmap{}{\zeta}{Z*}{2} \rho - \fpmap{}{\zeta}{Z*}{1} B \big) -
\big( \fpmap{}{\zeta}{Z*}{2} \rho' - \fpmap{}{\zeta}{Z*}{1} B \big) \
\in \ 2\pi\, \di\, \Omega^2_{{\rm cl},\RZ}(Z)~,
\end{equation}
whence, as $\zeta^Z$ is a surjective submersion, $\rho - \rho' \in
2\pi\, \di\, \Omega^2_{{\rm cl},\RZ}(M)$ as well; this establishes statement
(4) in Theorem~\ref{thm:triv}.

\begin{remark}
	Let us come back to Definition~\ref{def:lbgr_morph} of 1-morphisms of bundle gerbes.
	As noted for example in~\cite{Bouwknegt:2001vu}, from any morphism of bundle gerbes $(E,\alpha): \CL_1 \rightarrow \CL_2$ with $\rk(E) = n$, one can construct an isomorphism
	\begin{equation}
		\big( \det(E),\, \det(\alpha) \big): \CL_1^{\otimes n} \isom \CL_2^{\otimes n}~.
	\end{equation}
	In particular, the difference
	\begin{equation}
		{\rm dd}(\CL_2^{\otimes n}) - {\rm dd}(\CL_1^{\otimes n})
		= n \, \big( {\rm dd}(\CL_2) - {\rm dd}(\CL_1) \big)
		= 0\ \in\ H^3(M,\RZ_M)
	\end{equation}
	of the Dixmier-Douady classes of the gerbes vanishes.
	Hence ${\rm dd}(\CL_2) - {\rm dd}(\CL_1)$ is a torsion element of $H^3(M,\RZ_M)$.
	
	However, this isomorphism is, in general, not \emph{flat}
        since we do not necessarily have $F^{\det(E)} = \tr(F_{\nabla^E}) = n\, (\fpmap{}{\zeta}{Z*}{2}  B_2 - \fpmap{}{\zeta}{Z*}{1}  B_1)$.
	This additional condition on the curvature is much stronger
        than the restriction on the difference of the Dixmier-Douady
        classes of the gerbes: Since the trace of that curvature is a
        representative of the first Chern class of $E$, and hence a closed form, we obtain
	\begin{equation}
		\fpmap{}{\zeta}{Z*}{2} \dd B_2 - \fpmap{}{\zeta}{Z*}{1} \dd B_1 = 0~.
	\end{equation}
	This implies that the \emph{representatives} $H_i$ of the Dixmier-Douady classes of $\CL_i$ obtained on $M$ via $(\sigma^{Y_i})^* H_i = \dd B_i$ coincide as forms on $M$, rather than just defining the same de Rham cohomology class.
	
	One possibility would be to allow for morphisms of infinite
        rank, that is, to construct morphisms from bundles of separable Hilbert spaces rather than just from finite-rank hermitean vector bundles.
	Much of the theory of bundle gerbes still goes through (see e.g.~\cite{Bouwknegt:2001vu} for an approach to sections).
	The trace condition, however, becomes incompatible with tensor products in the case where one allows for infinite rank and requires trace-class properties on the connection.
	As we will point out in
        Section~\ref{sect:2-Hilbert_spaces_from_Bgerbes}, more serious
        problems then arise in our construction of a 2-Hilbert space from a bundle gerbe.
\label{rem:ddtorsion}\end{remark}

\subsection{2-morphisms}

In~\cite{Waldorf:2007aa,Waldorf:0702652}, Waldorf introduced 2-morphisms corresponding to these morphisms of bundle gerbes, which give rise to a (strictly associative) 2-category of bundle
gerbes $\CatLBGrb^\nabla(M)$. The 1-morphisms employ vector bundles, which come with natural morphisms between them.
\begin{definition}[\cite{Waldorf:0702652, Waldorf:2007aa}]
	\label{def:2hom_bgrb}
	Let $(E,\alpha,g,\nabla^E,\zeta^Z)$ and
        $(E',\alpha',g',\nabla^{E'},\zeta^{Z'})$ be 1-morphisms
        between bundle gerbes $\CL_1$ and $\CL_2$. A \uline{2-morphism from
        $(E,\alpha)$ to $(E',\alpha'\, )$} is an equivalence
        class of pairs $(\phi,\omega^W)$, where $\omega^W: W
        \thra  \hat{Z}:= Z\times_{Y_{12}}Z'$ with $Y_{12}:=Y_1\times_M
        Y_2$ is a surjective submersion, and $\phi \in
        \shom_{\CatHVBdl^\nabla(M)}(\fpmap{}{\omega}{W*}{Z}E,
        \fpmap{}{\omega}{W*}{Z'}E'\,)$ is a morphism of hermitean vector
        bundles with connection compatible with $\alpha$ and
        $\alpha'$, where $\omega_Z^W:={\rm pr}_Z\circ\omega^W$; that is, there is a commutative diagram
	\begin{equation}
	\label{eq:comp_2hom_bgrb}
	\begin{tikzpicture}[baseline=(current  bounding  box.center)]
	\node (L1E) at(0,0) {$L_{1,(w_1,w_2)} \otimes E_{w_2}$};
	\node (EL2) at(8,0) {$E_{w_1} \otimes L_{2,(w_1,w_2)}$};
	\node (L1E') at(0,-2.5) {$L_{1,(w_1,w_2)}\otimes E'_{w_2}$};
	\node (E'L2) at(8,-2.5) {$E'_{w_1}\otimes L_{2,(w_1,w_2)}$};
	
	\draw[->] (L1E)--(EL2) node[pos=.5,above]{\scriptsize{$\alpha_{(w_1,w_2)}$}};
	\draw[->] (L1E')--(E'L2) node[pos=.5,below]{\scriptsize{$\alpha'_{(w_1,w_2)}$}};
	\draw[->] (L1E)--(L1E') node[pos=.5,left]{\scriptsize{$\mathbbm{1} \otimes \phi_{w_2}$}};
	\draw[->] (EL2)--(E'L2) node[pos=.5,right]{\scriptsize{$\phi_{w_1} \otimes \mathbbm{1}$}};
	\end{tikzpicture}
	\end{equation}
 in $\CatHVBdl^\nabla(W^{[2]})$, where $W^{[2]} = W{\times_{M}}W$.\footnote{Here the fibres of the line bundles
   $L_{i,(w_1,w_2)}$ come
   from pullbacks to $\hat Z^{[2]}$ and then further to $W^{[2]}$,
   while the fibres of the vector bundles $E_{w_i}$ and $E_{w_i}'$ arise from pullbacks to
   $W$ and then further to $W^{[2]}$.}
	
	Two pairs $(\phi,\omega^W)$ and $(\widetilde{\phi},{\omega}^{\widetilde{W}})$ are \uline{equivalent} if there is a commutative diagram
	\begin{equation}
	\begin{tikzpicture}[baseline=(current  bounding  box.center)]
	\node (X) at(0,0) {$X$};
	\node (W) at(-2,-1) {$W$};
	\node (tW) at(2,-1) {$\widetilde{W}$};
	\node (hZ) at(0,-2) {$\hat{Z}$};
	
	\draw[->] (X)--(W) node[pos=0.2,left] {\scriptsize{$\xi~$}};
	\draw[->] (X)--(tW) node[pos=0.2,right] {\scriptsize{$~~\widetilde{\xi}$}};
	\draw[->] (W)--(hZ) node[pos=0.6,left] {\scriptsize{$\omega^W~$}};
	\draw[->] (tW)--(hZ) node[pos=0.6,right] {\scriptsize{$~{\omega}^{\widetilde{W}}$}};
	\end{tikzpicture}
	\label{eq:2morequiv}\end{equation}
	of surjective submersions such that $\xi^*\phi = \widetilde{\xi}^* \widetilde{\phi}$.
\end{definition}
Next we want to define identity 2-morphisms; this requires a rather technical discussion, introducing two special isomorphisms of vector bundles, which we hide in Appendix~\ref{app:special_morphisms}.

\begin{proposition} [\cite{Waldorf:2007aa}, p.~42]
	Let $(E,\alpha,g,\nabla^E,\zeta^Z) \in \shom_{\CatLBGrb^\nabla(M)}(\CL_1,\CL_2)$ be a 1-mor\-phism of bundle gerbes.
	The pair $\unit_{(E,\alpha)} := (\dd_{(E,\alpha)}, \id_{Z\times_{Y_{12}} Z}) $, where $\dd_{(E,\alpha)}$ is defined in Lemma~\ref{st:dd-def_and_properties}, represents the identity 2-morphism on $(E,\alpha)$.
\end{proposition}

We can now package our definitions into a 2-category~\cite{Waldorf:2007aa,Waldorf:0702652}.
\begin{definition}\label{def:lbgrb_2-category}
 The \uline{2-category of line bundle gerbes $\CatLBGrb^\nabla(M)$} has objects hermitean line bundle gerbes with connection. Its 1-morphisms are those of Definition~\ref{def:lbgr_morph} and the 2-morphisms are those of Definition~\ref{def:2hom_bgrb}.
 
 Given 1-morphisms $(E,\alpha,\zeta^Z)$, $(E',\alpha',\zeta^{Z'})$,
 $(E'',\alpha'',\zeta^{Z''}):\CL_2\to \CL_3$, and $(F, \beta,
 \zeta^X)$, $(F', \beta', \zeta^{X'}):\CL_1 \to \CL_2$,
 the \uline{vertical composition} of 2-morphisms $(\phi,\omega^W): (E',\alpha',\zeta^{Z'}) \Rightarrow (E'',\alpha'',\zeta^{Z''})$ and $(\widetilde \phi,\omega^{\widetilde W}): (E,\alpha,\zeta^Z) \Rightarrow (E',\alpha',\zeta^{Z'})$ is the equivalence class of
  \begin{equation}
	  (\phi, \omega^W) \circ (\widetilde \phi,\omega^{\widetilde W}) = \big((\pr_W^* \phi) \circ (\pr_{\widetilde W}^* \widetilde \phi) \,,\,(\pr_Z\circ\omega^{W}\circ \pr_W) \times_{Y_{12}} (\pr_{Z''}\circ \omega^{\widetilde W}\circ\pr_{\widetilde W}) \big)~.
  \end{equation}

 \uline{Horizontal composition} of a pair of 2-morphisms $(\widetilde \phi,\omega^{\widetilde W}): (E,\alpha,\zeta^Z) \Rightarrow (E',\alpha',\zeta^{Z'})$ and $(\phi,\omega^W): (F, \beta, \zeta^X) \Rightarrow (F', \beta', \zeta^{X'})$ is given by%
  \begin{equation}
	  (\widetilde \phi,\omega^{\widetilde W}) \htimes (\phi,\omega^W) = (\psi,\omega^U) \ ,
  \end{equation}
with
\begin{equation}
\begin{aligned}
U&= (X\times_{Y_2}Z)\times_{Y_{13}} (W\times_{Y_2}\widetilde{W}\,)\times_{Y_{13}} (X'\times_{Y_2}Z'\,) \ , \\[4pt]
\omega^U&= \pr_{X\times_{Y_2}Z}\times_{Y_{13}}\pr_{X'\times_{Y_2}Z'} \ , \\[4pt]
\psi&= \dd_{(E',\alpha'\,)\htimes (F',\beta'\,)}\circ (\phi\otimes\widetilde{\phi}\,)\circ\dd_{(E,\alpha)\htimes(F,\beta)} \ ,
\end{aligned}
\end{equation}
where the horizontal composition of 1-morphisms is given in \eqref{eq:comp_1hom}.
\end{definition}

\begin{remark}\label{rem:trivial_line_bundle_gerbes}
The 2-category $\CatLBGrb^\nabla(M)$ has a full sub-2-category $\CatLBGrb^\nabla_{\rm triv}(M)$ having as objects all trivial bundle gerbes $\CI_\rho$ with $\rho \in \di\, \Omega^2(M)$. This 2-category is closed under the tensor product of gerbes, since $\CI_\rho \otimes \CI_{\rho'} = \CI_{\rho + \rho'}$. Hence $\CatLBGrb^\nabla_{\rm triv}(M)$ has the same 2-categorical structure as the full $\CatLBGrb^\nabla(M)$.
\end{remark}

The fact that $\CatHVBdl^\nabla$ is a stack has strong implications for the structure of the $2$-category $\CatLBGrb^\nabla(M)$. Before moving on, however, let us comment on the reasons for introducing the additional surjective submersions $\zeta^Z:Z\thra
Y_{12}$ in Definition~\ref{def:lbgr_morph} and $\omega^W:W\thra
\hat Z$ in Definition~\ref{def:2hom_bgrb}. To this end, let $\shom_{\CatLBGrb^{\nabla}_{\rm FP}(M)}(\CL_1,\CL_2)$ be the subcategory of the category of morphisms $\shom_{\CatLBGrb^{\nabla}(M)}(\CL_1,\CL_2)$ in which we allow only 1-morphisms with $Z=Y_{12}$ and only 2-morphisms with $W=\hat Z$. We then have the following statement.
\begin{theorem}[\cite{Waldorf:2007aa}, Theorem~2.4.1]
	\label{st:Red_is_an_equivalence}
The inclusion functor 
\begin{equation}
  \sfS: \shom_{\CatLBGrb^{\nabla}_{\rm FP}(M)}(\CL_1,\CL_2) \longrightarrow \shom_{\CatLBGrb^{\nabla}(M)}(\CL_1,\CL_2)
\end{equation}
is an equivalence of categories. By adding the identity morphism on objects in $\CatLBGrb^{\nabla}(M)$, we can extend this inclusion functor to a lax 2-functor of bicategories $\CatLBGrb^{\nabla}_{\rm FP}(M)\rightarrow \CatLBGrb^{\nabla}(M)$. We denote the inverse constructed in~\cite{Waldorf:0702652,Waldorf:2007aa} by $\sfR$. In particular, the respective $2$-isomorphism classes of these bicategories are identical.
\end{theorem}

Although the morphism categories are equivalent, $\CatLBGrb^{\nabla}(M)$ has the advantage that it is a 2-category in which horizontal composition of 1-morphisms is associative, while this is not the case for the bicategory $\CatLBGrb^{\nabla}_{\rm FP}(M)$. In this way, the extension of $\CatLBGrb_{\rm FP}^\nabla(M)$ to $\CatLBGrb^\nabla(M)$ evades the use of descent theory in the definition of compositions.
Further details on this matter can be found in Appendix~\ref{app:special_morphisms}.

\begin{example}\label{ex:equiv_sec_trivial_hvbld}
	Let us revisit the example of morphisms $\shom_{\CatLBGrb^\nabla(M)}(\CI_0,\CI_\rho)$.
	The subcategory $\shom_{\CatLBGrb_{\rm FP}^{\nabla}(M)}(\CI_0,\CI_\rho)$ consists of hermitean vector bundles with connection on $M$ and parallel morphisms of vector bundles.
	 Then 
	\begin{equation}
		\shom_{\CatLBGrb^\nabla(M)}(\CI_0,\CI_\rho) \cong \shom_{\CatLBGrb_{\rm FP}^{\nabla}(M)}(\CI_0,\CI_\rho) \cong \CatHVBdl^\nabla(M)
	\end{equation}
by Theorem~\ref{st:Red_is_an_equivalence}.
\end{example}

We can also give an easy characterisation of the (weakly) invertible morphisms in $\shom_{\CatLBGrb^\nabla(M)}(\CL_1,\CL_2)$ as follows.

\begin{proposition}[\cite{Waldorf:2007aa}]\label{prop:morphisms_weakly_invertible}
	A 1-morphism $(E,\alpha):\CL_1\to\CL_2$ of bundle gerbes is (weakly) invertible if and only if its vector bundle $E$ is of rank~$1$.
\end{proposition}

The notion of {stable isomorphisms} of bundle gerbes therefore just corresponds to isomorphisms in $\shom_{\CatLBGrb_{\rm FP}^{\nabla}(M)}(\CL_1,\CL_2)$.
A consequence of the equivalence of categories from Theorem~\ref{st:Red_is_an_equivalence} is therefore the following.

\begin{corollary}
	Two bundle gerbes $\CL_1,\, \CL_2 \in \CatLBGrb^\nabla(M)$ are 1-isomorphic\footnote{This means that there is a 1-isomorphism in $\CatLBGrb^\nabla(M)$ between them.} if and only if they are stably isomorphic.
\end{corollary}

As a further application, we can now construct the kernel of a 2-morphism.
\begin{proposition}
	\label{st:kernels_of_2mors}
	Let $(\phi,\omega^W): (E,\alpha) \Rightarrow (E',\alpha'\,)$ be a 2-morphism of bundle gerbes on~$M$.
	\begin{enumerate}
		\item The kernel of $\phi$ defines a hermitean vector bundle with connection on $W$, which descends to $Z$ along $\omega^W_Z = \pr_Z \circ \omega^W$.
		This descent bundle has a natural embedding $\imath$ into $E$, and we set
		\begin{equation}
		K_\phi = \imath \Big( D_{\omega_Z^W} \big( \ker(\phi),\, (\omega^{W[2]*}_Z \dd_{(E,\alpha)})_{|\ker(\phi)} \big) \Big)~.
		\end{equation}
		\item The vector bundle $K_\phi\to Z$, together with the restriction of $\alpha$ to $L_{1,(z_1,z_2)} \otimes K_{\phi,z_2}$ over $(z_1,z_2)\in Z^{[2]}$, is a 1-morphism of bundle gerbes from the source $\sfs(E,\alpha)$ to the target $\sft(E,\alpha)$.
		\item This 1-morphism of bundle gerbes is a kernel of the 2-morphism $(\phi,\omega^W)$.
	\end{enumerate}
\end{proposition}

\begin{proof}
	(1):
	The kernel of $\phi$ is a hermitean vector bundle with connection over $W$, since $\phi$ is a parallel homomorphism.
	It has a canonical embedding into $\omega_Z^{W*}E$ as a sub-bundle.
	From Lemma~\ref{st:dd_and_other_2-mors} it follows that the bundle $\ker(\phi)$ together with the morphism $\fpmap{}{\omega}{W[2]*}{Z} \dd_{(E,\alpha)}$ is a descent datum in the stack $\CatHVBdl^\nabla$ for the surjective submersion $W \thra Y_{12}$.
	As $Z \thra Y_{12}$ is a surjective submersion as well, there is an inclusion $W{\times_Z}W \hookrightarrow W\times_{Y_{12}}W$, so that by restriction we obtain a descent datum for $W \thra Z$.
	The embedding $\ker(\phi) \hookrightarrow \omega_Z^{W*}E$ is a descent morphism, so that the descent bundle carries a canonical embedding into $E$ as vector bundles over $Z$.
	
\noindent	(2):
	The compatibility condition~\eqref{eq:comp_2hom_bgrb} for 2-morphisms in $\CatLBGrb^\nabla(M)$ implies that $\alpha$ induces a morphism of vector bundles
	\begin{equation}
	\big(\fpmap{}{\omega}{W[2]*}{Z} \alpha\big)_{(w_1,w_2)} :\ L_{1,(w_1,w_2)} \otimes \ker(\phi)_{w_2} \isom \ker(\phi)_{w_1} \otimes L_{2,(w_1,w_2)}
	\end{equation}
	over $(w_1,w_2)\in W^{[2]}$.
	By Lemma~\ref{st:dd-def_and_properties}~(2), it is compatible with the descent datum induced by the identity on $L_i$ and $\fpmap{}{\omega}{W[2]*}{Z} \dd_{(E,\alpha)}$ on $\omega_Z^{W*}E$.
	Thus $\alpha$ restricts to give a bundle gerbe 1-morphism $(K_\phi,\alpha_{|K_\phi},\zeta^Z)$.
	
\noindent (3):
	For any other 2-morphism $(\widetilde{\phi},\omega^{\widetilde W}): (E'',\alpha''\,) \Rightarrow (E,\alpha)$ such that $(\phi,\omega^W) \circ (\widetilde{\phi},\omega^{\widetilde W}) = 0$, the pullback of $\widetilde{\phi}$ to $W{\times_Z}\widetilde W$ factors through the pullback of the kernel of $\phi$.
	This factor morphism of vector bundles is uniquely defined since the kernel is defined fibrewise, i.e.\ in the abelian category of vector spaces, and hence yields a unique 2-morphism $(E'',\alpha''\, ) \Rightarrow \ker(\phi,\omega^W)$.
\end{proof}

\begin{remark}
	We can similarly define the cokernel $\coker(\phi,\omega^W)$ as the descent of $\coker(\phi)$ along $\omega_{Z'}^W$ to $Z'$; due to the unitarity of $\alpha'$, we may use the orthogonal complement $\im(\phi)^\perp$ to construct the cokernel.
	A 2-morphism is monic (epic) if its kernel (cokernel) is the zero 1-morphism.
	Using orthogonal complements together with kernels of vector bundles and descent as above, we infer that every monic (epic) 2-morphism is a kernel (cokernel).
\end{remark}

The existence of kernels and cokernels of 2-morphisms, together with the enrichment over abelian groups, suggests that there may be an abelian structure on the morphism categories.
This point will be addressed in detail in the following.

\subsection{Monoidal structure}

As stated in Section~\ref{ssec:gq_outline}, the prequantum Hilbert space on a quantisable symplectic manifold is obtained from sections of the prequantum line bundle $L$. Thus we wish to generalise the notion of section from line bundles to line bundle gerbes. A definition of sections of a vector bundle $E$ which readily categorifies is that of a morphism from the trivial line bundle $M\times\FC$ to $E$. In particular, the evaluation
\begin{equation}
\mathsf{ev}_1: \shom_{\CatVBdl(M)}(M \times \FC, E) \overset{\cong}{\longrightarrow} \Gamma(M,E) \ , \qquad \varepsilon_x \longmapsto \varepsilon_x(1)
\end{equation}
is an isomorphism of modules over $C^\infty(M)$. In fact, the module action of $C^\infty(M)$ on the sections $\shom_{\CatVBdl(M)}(M \times \FC, E)$ can be viewed as being induced by pre-composition with elements of $\shom_{\CatVBdl(M)}(M\times \FC, M \times \FC) \cong C^\infty(M)$. Since $\shom_{\CatVBdl(M)}(M \times \FC, E)$ forms a vector space over the field of constant functions, we recover the vector space structure underlying the prequantum Hilbert space.

For the more general module action of $C^\infty(M)$ on $\shom_{\CatVBdl(M)}(E,F)$ for a pair of vector bundles $E$ and $F$ on $M$, we can use the monoidal structure on the category $\CatVBdl(M)$: We map a pair consisting of a section $\varepsilon \in \shom_{\CatVBdl(M)}(E,F)$ and a function $f \in \shom_{\CatVBdl(M)}(M\times \FC, M \times \FC)$ to
\begin{equation}
	 f \otimes \varepsilon \ \in \ \shom_{\CatVBdl(M)} \big( (M\times \FC) \otimes E,\, (M \times \FC) \otimes F \big) =  \shom_{\CatVBdl(M)}(E,F)~.
\end{equation}

In the following we make explicit the straightforward lift of the above constructions to the case of line bundle gerbes.
\begin{definition}
	The \uline{category of sections} of a line bundle gerbe $\CL$ is
	\begin{equation}
	\Gamma(M,\CL) = \shom_{\CatLBGrb^\nabla(M)}(\CI_0,\CL)~.
	\end{equation}
\end{definition}
This category is a module category over $\shom_{\CatLBGrb^\nabla(M)}(\CI_0,\CI_0)$ under the tensor product in $\CatLBGrb^\nabla(M)$, which is equivalent to the category $ \CatHVBdl^\nabla(M)$, cf.\ Example~\ref{ex:equiv_sec_trivial_hvbld}.\footnote{Waldorf showed in~\cite{Waldorf:2007aa} that in fact any morphism category in $\CatLBGrb^\nabla(M)$ is a module category over $\CatHVBdl^\nabla_0(M)$, the category of hermitean vector bundles with connection of vanishing traced curvature, with the module action stemming from the monoidal structure in $\CatLBGrb^\nabla(M)$.}

To obtain a 2-vector space underlying our eventual prequantum 2-Hilbert space, we need an additional monoidal structure. An obvious candidate here is the direct sum, which has some nice properties.
\begin{theorem}
	\label{st:Direct_sum_as_fctr}
	Let $(E,\alpha,g,\nabla^E,B,\zeta^Z)$ and $(E',\alpha',g', \nabla^{E'},B',\zeta^{Z'})$ be 1-morphisms in $\shom_{\CatLBGrb^\nabla(M)}(\CL_1,\CL_2)$. 
	\begin{enumerate}
	 \item Consider the assignment
	 \begin{equation}
	  (E,\alpha,g,\nabla^E,B,\zeta^Z)\oplus(E',\alpha',g', \nabla^{E'},B',\zeta^{Z'})=(F,\beta,h,\nabla^{F},B^F,\zeta^{\hat Z})
	 \end{equation}
	 with $(F,h,\nabla^F)$ the direct sum of the pullbacks of $(E,g,\nabla^E)$ and $(E',g',\nabla^{E'})$ to $\hat Z=Z\times_{Y_{12}}Z'$, $B^F$ the sum of the pullbacks of $B$ and $B'$, and 
	 \begin{equation}
	  \beta_{(\hat z_1,\hat z_2)}:=\big(\sfd^{\rm r}_{L_{2,(\hat z_1,\hat z_2)}}\big)^{-1}\circ (\alpha_{(\hat z_1,\hat z_2)}\oplus\alpha'_{(\hat z_1,\hat z_2)})\circ \sfd^{\rm l}_{L_{1,(\hat z_1,\hat z_2)}}~,
	 \end{equation}
	 where $\sfd^{\rm l,r}_L$ denote the distribution isomorphisms $\sfd^{\rm l}_L: L \otimes (E \oplus F) \rightarrow (L\otimes E)\oplus(L\otimes F)$ and $\sfd^{\rm r}_L: (E \oplus F) \otimes L \rightarrow (E\otimes L)\oplus(F\otimes L)$. This assignment turns $\shom_{\CatLBGrb^\nabla(M)}(\CL_1,\CL_2)$ into a symmetric monoidal category.
	  \item The tensor product on $\CatLBGrb^\nabla(M)$ is a monoidal functor on the morphism categories in $\CatLBGrb^\nabla(M)$ (in addition to the horizontal composition $\bullet$); that is, the tensor product induced on the morphism categories by the monoidal structure on $\CatLBGrb^\nabla(M)$ distributes over the direct sum in the sense that there exist natural transformations acting as
	  \begin{equation}\label{eq:left-distribution}
	  \sfd^{\rm l}_{F,E,E'}: F \otimes (E \oplus E'\, ) \twoisom (F \otimes E) \oplus (F \otimes E'\, )
	  \end{equation}
	  and
	  \begin{equation}\label{eq:right-distribution}
	  \sfd^{\rm r}_{E,E',F}: (E \oplus E'\, ) \otimes F \twoisom (E \otimes F) \oplus (E' \otimes F)~.
	  \end{equation}
	\end{enumerate}
\end{theorem}
\noindent The proof of Theorem~\ref{st:Direct_sum_as_fctr} is rather involved and hence is deferred to Appendix~\ref{app:Proof_of_direct_sum_thm}.

Next we work out further structure of the morphism categories $\shom_{\CatLBGrb^\nabla(M)}(\CL_1,\CL_2)$.
\begin{theorem}
	\label{st:hom_cats_are_semisimple_Abelian}
	Every morphism category $\shom_{\CatLBGrb^\nabla(M)}(\CL_1,\CL_2)$ is a semisimple abelian and cartesian symmetric monoidal category, with binary products and coproducts coinciding and given by the direct sum.
\end{theorem}

\begin{proof}
	Let us first show that the direct sum is both a product and a coproduct in\linebreak $\shom_{\CatLBGrb^\nabla(M)}(\CL_1,\CL_2)$. We can construct 2-monomorphisms from two summands $(E,\alpha,\zeta^Z)$ and $(E',\alpha',\zeta^{Z'})$ into the direct sum morphism $(E,\alpha,\zeta^Z)\oplus (E',\alpha',\zeta^{Z'})$ by pulling the summand back to $\hat Z$ using the usual embedding of a vector bundle into a direct sum.
	We now use the fact that given bundle gerbe morphisms $(E,\alpha,\zeta^Z)$, $(F,\beta,\zeta^Z)$ over the same surjective submersion $\zeta^Z$, and given a parallel morphism $\phi$ between these vector bundles intertwining $\alpha$ and $\beta$, then $\phi$ induces a unique 2-morphism $(E,\alpha,\zeta^Z) \Rightarrow (F,\beta,\zeta^Z)$.
	Similarly, the projections to the summands are constructed via the projections from the Whitney sum of the bundles over $\hat Z=Z{\times_{Y_{12}}}Z'$.
	Then, when pulled back to appropriate surjective submersions as indicated above, the situation reduces to that of ordinary vector bundles with connections, and it follows that the direct sum is in fact the product and coproduct.
	Hence the morphism categories are cartesian monoidal, and furthermore symmetric by Theorem~\ref{st:Direct_sum_as_fctr}.
	
	Semisimplicity follows from the fact that every 2-endomorphism $(\phi,\omega^W)$ induces a splitting of a morphism into the eigenspaces of $\phi$ as
	\begin{equation}
	(E,\alpha) \cong \bigoplus_{\lambda \in \sigma(\phi)}\, \ker \big( (\phi,\omega^W) - \lambda\,\unit_{(E,\alpha)} \big)~.
	\end{equation}
	By Proposition~\ref{st:kernels_of_2mors}, kernels of 2-morphisms again form 1-morphisms in $\CatLBGrb^\nabla(M)$.
	Therefore, as long as there exists a 2-morphism on $(E,\alpha)$ which is not a non-zero multiple of the identity 2-morphism, then $(E,\alpha)$ can be decomposed into its eigenspaces.
	This means that every 1-morphism decomposes (non-canonically) into 1-morphisms with the property that their only 2-endomorphisms are constant multiples of the identity 2-morphism on them, and the result now follows.
\end{proof}

\begin{remark}
	The existence of kernels of 2-morphisms, and therefore the fact that the category $\shom_{\CatLBGrb^\nabla(M)}(\CL_1,\CL_2)$ is abelian, is lost if we do not require the morphisms of vector bundles appearing in the representatives of 2-morphisms to be parallel.
\end{remark}

Finally, for every triple of bundle gerbes $\CL_1,\, \CL_2,\, \CL_3$ there is an equivalence of categories
\begin{equation}
\label{eq:Dual_of_BGrb_functor}
\begin{aligned}
 &\Delta_{\CL_2}: \shom_{\CatLBGrb^\nabla(M)} \big( \CL_1 \otimes \CL_2,\, \CL_3 \big) \isom \shom_{\CatLBGrb^\nabla(M)} \big( \CL_1,\, \CL_2^* \otimes \CL_3 \big)~,\\
  &\Delta_{\CL_2} = \big( - \otimes \, \id_{\CL_2^*} \big)\circ \delta^{-1}_{\CL_2} ~,
\end{aligned}
\end{equation}
where $\delta_{\CL}: \CL^* \otimes \CL \rightarrow \CI_0$ is the canonical isomorphism \eqref{eq:delta_G}. As a composition of functors with (weak) inverses, this is indeed an equivalence.
This provides us with an adjunction of endo-2-functors on $\CatLBGrb^\nabla(M)$,
\begin{equation}
- \otimes\, \CL_2  \ \dashv \ \CL_2^* \otimes -~,
\end{equation}
and thereby shows that the monoidal structure on $\CatLBGrb^\nabla(M)$ is closed.

The structures found so far are now summarised in the following statement.
\begin{theorem}\label{thm:4.28}
The 2-category $\CatLBGrb^\nabla(M)$ is a closed symmetric monoidal 2-category whose morphism categories are semisimple abelian and cartesian symmetric monoidal categories.
\end{theorem}

\subsection{Closed structure}
\label{sect:2-bundle_with_2-metric}

The next step towards a definition of prequantum 2-Hilbert spaces from sections of line bundle gerbes is the introduction of a local inner product on the 1- and 2-morphisms. 

Recall that for any hermitean line bundle $(L,h)$ over a manifold $M$, the Riesz representation theorem provides us with an isometric anti-isomorphism $\flat_h:L\rightarrow L^*$ over each point $x\in M$ as well as its inverse $\sharp_h:L^*\rightarrow L$. We can use this to define an inner product on morphisms between hermitean line bundles $(L_1,h_1)$ and $(L_2,h_2)$ as
\begin{equation}
\begin{aligned}
	&\frh: \shom_{\CatHVBdl(M)}(L_1,L_2) \otimes \shom_{\CatHVBdl(M)}(L_1,L_2) \longrightarrow C^\infty(M)~, \\
	&\frh(\psi,\phi)\big|_x = (\flat_{h_2}\circ\psi\circ \sharp_{h_1})(\tilde \ell\, ) \big(\phi(\ell)\big)~,
\end{aligned}
\end{equation}
where $\ell\in L_1|_x$ and $\tilde \ell\in L^*_1|_x$ such that $\tilde \ell(\ell)=1$.

By restricting to morphisms $\shom_{\CatHVBdl(M)}(M\times \FC, L)$ corresponding to sections $\Gamma(M,L)$ and choosing vectors $\ell$ of unit length $h_1(\ell,\ell)=1$, this inner product simplifies to
\begin{equation}
	\frh(\psi,\phi)\big|_x = h(\psi(\ell),\phi(\ell))~.
\end{equation}

Let us now categorify these observations. We wish to establish some functor $\Theta$ from the category $\shom_{\CatLBGrb^\nabla(M)}(\CL_1,\CL_2)$ to $\shom_{\CatLBGrb^\nabla(M)}(\CL^*_1,\CL^*_2)$. A natural choice on 1-morphisms $(E,\alpha,g,\nabla^E,\zeta^Z)\in \shom_{\CatLBGrb^\nabla(M)}(\CL_1,\CL_2)$ is given by 
\begin{equation}
\label{eq:def_Theta_morphisms}
	\Theta \big( E,\alpha,g,\nabla^E, \zeta^Z \big) = \big( E^*, \alpha^{-{\rm t}}, g^{-1}, \nabla^{E^*}, \zeta^Z \big)~.
\end{equation}
This satisfies the compatibility condition with the bundle gerbe multiplications on $\CL_1^*$ and $\CL_2^*$ as we see from
\begin{equation}
\begin{aligned}
	\alpha^{-{\rm t}}_{(z_1,z_3)} \circ \big( \mu_{1,(z_1,z_2,z_3)}^{-{\rm t}} \otimes \mathbbm{1} \big)
	&= \big( \alpha_{(z_1,z_3)} \circ (\mu_{1,(z_1,z_2,z_3)} \otimes \mathbbm{1} ) \big)^{-{\rm t}}\\[4pt]
	&= \big( ( \mathbbm{1} \otimes \mu_{2,(z_1,z_2,z_3)} ) \circ ( \alpha_{(z_1,z_2)} \otimes \mathbbm{1} ) \circ ( \mathbbm{1} \otimes \alpha_{(z_2,z_3)} ) \big)^{-{\rm t}}\\[4pt]
	&= \big( \mathbbm{1} \otimes \mu_{2,(z_1,z_2,z_3)}^{-{\rm t}} \big) \circ \big(  \alpha^{-{\rm t}}_{(z_1,z_2)} \otimes \mathbbm{1} \big) \circ \big( \mathbbm{1} \otimes \alpha^{-{\rm t}}_{(z_2,z_3)} \big)
\end{aligned}
\end{equation}
over $(z_1,z_2,z_3)\in Z^{[3]}$. Here we used the fact that pullbacks of morphisms of vector bundles respect taking the transpose and inverse.

Next let us check that $\Theta$ is compatible with mapping 2-morphisms to 2-morphisms. Consider a representative $(\phi,\omega^W): (E,\alpha) \Rightarrow (F,\beta)$ of a 2-morphism in $\CatLBGrb^\nabla(M)$. Then
\begin{equation}
(\phi_{w_1} \otimes \mathbbm{1}) \circ \alpha_{(w_1,w_2)} = \beta_{(w_1,w_2)} \circ (\mathbbm{1} \otimes \phi_{w_2})
\end{equation}
which implies
\begin{equation}
(\phi^{\rm t}_{w_1} \otimes\mathbbm{1}) \circ \beta^{-{\rm t}}_{(w_1,w_2)} = \alpha^{-{\rm t}}_{(w_1,w_2)} \circ ( \mathbbm{1}\otimes \phi^{\rm t}_{w_2})
\end{equation}
over $(w_1,w_2)\in W^{[2]}$. This is the compatibility relation for a {\em contravariant} functor. We can readily extend this contravariant functor on the morphism category to a 2-contravariant 2-functor\footnote{By this we mean a 2-functor which is contravariant on 2-morphisms.} on all of $\CatLBGrb^\nabla(M)$ by using the evident Riesz isomorphisms to map $\CL$ to $\CL^*$.

\begin{definition}
	\label{def:Theta}
	The \uline{Riesz functor} $\Theta: \CatLBGrb^\nabla(M) \rightarrow \CatLBGrb^\nabla(M)$ is the strict 2-contra\-variant 2-functor acting on objects, 1-morphisms, and 2-morphisms as
	\begin{equation}
	\begin{aligned}
	\Theta(\CL) &= \CL^*~,\\[4pt]
	\Theta \big( E, \alpha, g, \nabla^E, \zeta^Z \big) &= \big( E^*, \alpha^{-{\rm t}}, g^{-1}, \nabla^{E^*}, \zeta^Z \big)~,\\[4pt]
	\Theta \big(\phi , \omega^W\big) &= \big( \phi^{\rm t} ,\sfsw \circ \omega^W \big)~,
	\end{aligned}
	\end{equation}
	with $\sfsw$ the swap map
	\begin{equation}
	 \sfsw: Z {\times_{Y_{12}}} Z' \longrightarrow Z' {\times_{Y_{12}}} Z \ , \quad (z,z'\,) \longmapsto (z',z) \ .
	\end{equation}
\end{definition}

\begin{proposition}
	With the above definition, $\Theta: \CatLBGrb^\nabla(M) \rightarrow \CatLBGrb^\nabla(M)$ is a 2-contravariant 2-functor which respects the monoidal structure on $\CatLBGrb^\nabla(M)$ as well as those on the morphism categories.
\end{proposition}

\begin{proof}
	Functoriality on 1-morphisms follows from the fact that taking duals of vector bundles is compatible with tensor products, and that the trivial line bundle is self-dual.
	On 2-morphisms, compatibility with vertical composition is just the well-known fact that $(\phi \circ \rho)^{\rm t} = \rho^{\rm t} \circ \phi^{\rm t}$.
	Preservation of units follows because of $\dd_{(E,\alpha)}^{-{\rm t}}= \dd_{(E^*,\alpha^{-{\rm t}})} = \dd_{\Theta(E,\alpha)}$, cf.\ Lemma~\ref{st:dd-def_and_properties}, so that
	\begin{equation}
	\begin{aligned}
	\Theta (\unit_{(E,\alpha)}) &=\Theta \big( \dd_{(E,\alpha)},\id_{Z_{\times_{Y_{12}}}Z} \big)\\[4pt]
	&= \big(\dd_{(E,\alpha)}^{\rm t}, \sfsw \circ \id_{Z_{\times_{Y_{12}}}Z} \big)\\[4pt]
	&= \big( \dd_{(E^*,\alpha^{-{\rm t}})}^{-1},\sfsw \circ \id_{Z_{\times_{Y_{12}}}Z} \big)\\[4pt]
	&= \big(\dd_{(E^*,\alpha^{-{\rm t}})}, \id_{Z_{\times_{Y_{12}}}Z} \big)^{-1}\\[4pt]
	&= \big(\unit_{(E^*,\alpha^{-{\rm t}})} \big)^{-1}\\[4pt]
	&= \unit_{\Theta(E,\alpha)}~.
	\end{aligned}
	\end{equation}
	Compatibility of $\Theta$ with the monoidal structures on $\CatLBGrb^\nabla(M)$ and its morphism categories follows from the compatibility of tensor products and direct sums with taking duals, transposes, and inverses of morphisms of hermitean vector bundles with connections.
\end{proof}

The Riesz functor $\Theta$ differs significantly from the dual functor
      defined in~\cite{Waldorf:2007aa,Waldorf:0702652}. Not only does
      it act on the directions of 1- and 2-morphisms in the opposite way,
      but, more importantly, when acting on $\CatHVBdl^\nabla(M)$ as a subcategory of
      $\shom_{\CatLBGrb^\nabla(M)}(\CI_0,\CI_\rho)$ the Riesz functor yields the dual vector
      bundle rather than acting as the
      identity.
      This can again be understood as the difference between the Hilbert space adjoint
      of a map (which is obtained from Riesz duals) and the vector or Banach space
      transpose:
      Sending an operator to the former is an antilinear map, whereas sending it to the
      latter is a linear map.
      In particular, on scalars the first map is complex conjugation, which is crucial
      for constructing hermitean inner products, while the second map is the identity.
The functor $\Theta$ is more closely related to what is called the
\emph{adjoint} of a 1-morphism in~\cite{Gawedzki:2008um}. The difference in our definition is that $\Theta$ is
constructed to act 2-contravariantly on 2-morphisms, since we allow for
non-invertible 2-morphisms which prevents us from reversing the direction of
the 2-morphism in Definition~\ref{def:Theta} by adding an inverse on
$\phi^{\rm t}$.

\begin{lemma}
	The functors $\sR$ and $\sfS$ establishing the equivalence from Theorem~\ref{st:Red_is_an_equivalence} satisfy
	\begin{equation}
	\sR \circ \Theta = \Theta \circ \sR \eand \sfS \circ \Theta = \Theta \circ \sfS~.
	\end{equation}
They are compatible with the monoidal structure $\otimes$ on $\CatLBGrb^{\nabla}(M)$, as well as the monoidal structure $\oplus$ on the morphism categories.
\end{lemma}
\begin{proof}
	The functor $\sR$ is descent for each morphism.
	Descent is compatible with duals, direct sums, and tensor products of vector bundles with connections, and with transposes and inverses of descent morphisms.
\end{proof}

Using the Riesz functor, we can now define a ``categorified inner product'' on the morphism categories in $\CatLBGrb^\nabla(M)$ as
\begin{equation}
\lsb (E,\alpha), (F,\beta) \rsb = \Theta(E,\alpha) \otimes (F,\beta) \ \in \  \shom_{\CatLBGrb^\nabla(M)}(\CL_1^* \otimes \CL_3,\, \CL_2^* \otimes \CL_4)
\end{equation}
for $(E,\alpha)\in \shom_{\CatLBGrb^\nabla(M)}(\CL_1,\CL_2)$ and $(F,\beta)\in \shom_{\CatLBGrb^\nabla(M)}(\CL_3,\CL_4)$. This inner product extends to the following 2-contravariant 2-functor.
\begin{definition}
\label{def:internal_hom}
	The \uline{internal 2-hom-functor} on $\CatLBGrb^\nabla(M)$ is the 2-contravariant 2-functor given by
	\begin{equation}
	\begin{aligned}
	&\lsb - , - \rsb : \CatLBGrb^\nabla(M) \times \CatLBGrb^\nabla(M) \longrightarrow \CatLBGrb^\nabla(M)~,\\
	&\lsb - ,- \rsb = \otimes_{\CatLBGrb^\nabla(M)} \circ \big( \Theta \times \unit_{\CatLBGrb^\nabla(M)} \big)~. 
	\end{aligned}
	\end{equation}
\end{definition}
The 2-contravariant functoriality of $\lsb- ,-\rsb$ is immediate, since it is constructed from 2-contravariant 2-functors.
By the same argument, it is sesquilinear with respect to the tensor product and direct sum in the sense that there are natural 2-isomorphisms
\begin{equation}
\begin{aligned}
\lsb (E,\alpha) \otimes (E',\alpha'\,),\, (F,\beta) \rsb &\twoisom \Theta(E,\alpha) \otimes \lsb (E',\alpha'\,),\, (F,\beta) \rsb~,\\[4pt]
\lsb (E,\alpha) \oplus (E',\alpha'\,),\, (F,\beta) \rsb &\twoisom \lsb (E,\alpha),\, (F,\beta) \rsb \oplus \lsb (E',\alpha'\,),\, (F,\beta) \rsb~,\\[4pt]
\lsb (E,\alpha), (F,\beta) \rsb &\twoisom \Theta \Big( \lsb (F,\beta), (E,\alpha) \rsb \Big)~.
\end{aligned}
\end{equation}
The first natural 2-isomorphism stems from the associativity of the tensor product in\linebreak $\CatLBGrb^\nabla(M)$, the second amounts to the functorial distributivity of the tensor product over the direct sum, and the natural 2-isomorphism in the last identity just swaps the factors in the tensor product.
On 2-morphisms it acts as follows:
If $(\phi,\omega^W): (E',\alpha'\,) \Rightarrow (E,\alpha)$ and $(\psi,\omega^V): (F,\beta) \Rightarrow (F',\beta'\,)$ are 2-morphisms in $\CatLBGrb^\nabla(M)$, then
\begin{equation}
\lsb (\phi,\omega^W),\, (\psi,\omega^V) \rsb = \Theta(\phi,\omega^W) \otimes (\psi,\omega^V): \lsb (E,\alpha),\, (F,\beta) \rsb \Longrightarrow \lsb (E',\alpha'\,),\, (F',\beta'\,) \rsb~.
\end{equation}

The terminology `internal 2-hom-functor' is justified by the following statement.

\begin{theorem}
	\label{st:internal-2-hom-adjunction}
	For every 1-morphism $(F,\beta) \in \shom_{\CatLBGrb^\nabla(M)}(\CL_3,\CL_4)$ there is an adjoint pair
	\begin{equation}
		\xymatrixcolsep{1.5cm}
		\myxymatrix{
			\shom_{\CatLBGrb^\nabla(M)}(\CL_1,\CL_2) \ \ar@<0.2cm>[rr]^-{- \otimes (F,\beta)} & \perp & \ \shom_{\CatLBGrb^\nabla(M)}(\CL_1 \otimes \CL_3, \CL_2 \otimes \CL_4) \ar@<0.2cm>[ll]-^{\delta_{\CL_4} \circ \lsb (F,\beta) , \, - \rsb \circ \delta^{-1}_{\CL_3}}~.
		}
	\end{equation}
This endows the morphism category $\sHom_{\CatLBGrb^\nabla(M)}$, consisting of all morphism categories in $\CatLBGrb^\nabla(M)$, with an internal hom-functor.
\end{theorem}

\begin{proof}
	Let $(E,\alpha): \CL_1 \rightarrow \CL_2$, $(F,\beta): \CL_3 \rightarrow \CL_4$, and $(G,\gamma): \CL_1 \otimes \CL_3 \rightarrow \CL_2 \otimes \CL_4$ be 1-morphisms in $\CatLBGrb^\nabla(M)$.
	We have to show that there exist bijections
	\begin{equation}
	\begin{aligned}
	&\Psi_{E,F,G}: 2\shom_{\CatLBGrb^\nabla(M)} \big( (E,\alpha) \otimes (F,\beta),\, (G,\gamma) \big)\\
	&\hspace{3cm}\isom 2\shom_{\CatLBGrb^\nabla(M)} \Big( (E,\alpha),\, {\delta_{\CL_4}} \bullet \lsb (F,\beta) ,\, (G,\gamma) \rsb \bullet {\delta^{-1}_{\CL_3}} \Big)~.
	\end{aligned}
	\end{equation}
	For this, we have to establish a natural correspondence from
	\begin{equation}
	\xymatrixcolsep{0.3cm}
	\myxymatrix{
		\CL_1 \otimes \CL_3 \rrtwocell<10>^{(E,\alpha) \otimes (F,\beta)}_{(G,\gamma)}{\omit} & \Downarrow {\scriptstyle{(\phi,\omega^W)}} & \CL_2 \otimes \CL_4
	}
\end{equation}
to
\begin{equation}
	\xymatrixcolsep{0.cm}
	\myxymatrix{
		\CL_1 \ar@{->}[rr]^-{(E,\alpha)} \ar@{->}[dd]_-{\delta_{\CL_3}^{-1} \otimes \id_{\CL_1}} & & \CL_2 \\
		& \Downarrow {\scriptstyle{\Psi_{E,F,G}(\phi,\omega^W)}} & \\
		\CL_3^* \otimes \CL_1 \otimes \CL_3 \ar@{->}[rr]_-{\lsb(F,\beta),\,(G,\gamma)\rsb} & & \CL_4^* \otimes \CL_2 \otimes \CL_4 \ar@{->}[uu]_-{\delta_{\CL_4} \otimes \id_{\CL_2}}
	}
	\end{equation}
	The proof is complicated by the fact that $(E,\alpha)$ and $\lsb(F,\beta),(G,\gamma)\rsb$ are not defined over the same surjective submersion.
	To cure this, we employ the equivalence of $\CatLBGrb^\nabla(M)$ to its reduced version $\CatLBGrb^{\nabla}_{\rm FP}(M)$ given in Theorem~\ref{st:Red_is_an_equivalence}, which implies the existence of natural isomorphisms of 2-functors
	\begin{equation}
	\eta: \sfS \circ \sfR \twoisom \unit_{\CatLBGrb^\nabla(M)} \eand \eps = \unit: \sfR \circ \sfS \twoisom \unit_{\CatLBGrb^{\nabla}_{\rm FP}(M)}~.
	\end{equation}
	In particular, restricting to each morphism category yields natural isomorphisms of functors $\eta_{\CL_1\to\CL_2}$. This equivalence is compatible with the tensor product.
	
We write $Y_{ij\cdots}=Y_i\times_M Y_j\times_M\cdots$. For the first diagram, the reduced 1-morphisms $\sR(E,\alpha)$ and $\sR (F,\beta)$ live over $Y_{12}$ and $Y_{34}$, respectively, while $\sR (G,\gamma)$ is defined on $Y_{1324}$.
	Thus the reduced 2-morphism $\sR (\phi,\omega^W)$ amounts to a morphism of vector bundles
	\begin{equation}
	\myxymatrix{
		\sR (E \otimes F) = \pr_{Y_{12}}^* \sR E \otimes \pr_{Y_{34}}^* \sR F \ar@{->}[rr]^-{\sR \phi} \ar@{->}[dr] & & \sR G \ar@{->}[dl]\\
		& Y_{1324} &
	}
	\end{equation}
	Denote by $\widetilde{\sR \phi}: \pr_{Y_{12}}^* \sR E \rightarrow \pr_{Y_{34}}^* \sR F^* \otimes \sR G$ the morphism of vector bundles defined by $\widetilde{\sR \phi}(e)(f) = \sR \phi (e \otimes f)$.
	This defines a 2-morphism in $\CatLBGrb^\nabla(M)$, and by applying the reduction functor it gives rise to the descent
	\begin{equation}
	\xymatrixcolsep{2cm}
	\myxymatrix{
		\pr_{Y_{12}}^* \sR E \ar@{->}[rr]^-{\widetilde{\sR \phi}} \ar@{->}[dr] & & \pr_{Y_{34}}^* \sR\, \Theta F \otimes \sR G \ar@{->}[dl] \\
		& Y_{1324} \ar@{->}'[d][dd]^-{\pr_{Y_{12}}} & \\
		\sR E \ar@{->}[rr]^(.2){D_{\pr_{Y_{12}}}(\widetilde{\sR \phi})} \ar@{->}[dr] & & D_{\pr_{Y_{12}}}(\pr_{Y_{34}}^* \sR\, \Theta F \otimes \sR G) = \sR(\Theta F \otimes G) \ar@{->}[dl] \\
		& Y_{12} &
	}
	\end{equation}
We set
	\begin{equation}
	\begin{aligned}
	&\Psi_{E,F,G}(\phi,\omega^W) = \eta_{\Theta F \otimes G} \circ \big( D_{\pr_{Y_{12}}}(\widetilde{\sR \phi}),\,\id_{Y_{12}}  \big) \circ \eta_{E}^{-1}\\
	& \qquad \qquad \qquad \qquad \in \
        2\shom_{\CatLBGrb^\nabla(M)} \Big( (E,\alpha),\,
        \delta_{\CL_4} \circ \lsb (F,\beta) , \, (G,\gamma) \rsb \circ \delta^{-1}_{\CL_3} \Big)~,
	\end{aligned}
	\end{equation}
so that the morphism $\Psi_{E,F,G}(\phi,\omega^W)$ is the composition
	\begin{equation}
		\xymatrixrowsep{2cm}
		\xymatrixcolsep{1.5cm}
		\xymatrix{
			E \ar@{-->}[r]^-{\Psi_{E,F,G}(\phi,\omega^W)} \ar@{->}[d]_{\eta_E^{-1}} & \Theta F \otimes G\\
			\sfS \, \sfR E \ar@{->}[r]_-{\sfS\, \sfR \widetilde{\Psi(\phi,\omega^W)}}  & \sfS \, \sfR (\Theta F \otimes G) \ar@{->}[u]_-{\eta_{\Theta F \otimes G}}
		}
	\end{equation}
Naturality of this map is evident from the fact that descent, and hence $\sR$, are functors and therefore respect composition of 2-morphisms. Compatibility with the monoidal structures in $\CatLBGrb^\nabla(M)$ is a consequence of descent and $\sR$ being compatible with direct sums, tensor products, and duals.
\end{proof}

The restriction of the internal 2-hom-functor to sections of a fixed bundle gerbe $\CL$ deserves a special name.
\begin{definition}
	\label{def:2-bdl_metric}
	The functor
	\begin{equation}
	\frh = \sfl_{\delta_{\CL}} \circ \lsb- ,- \rsb \circ \sfr_{\delta^{-1}_{\CI_0}} : \Gamma(M,\CL)^{\rm op} \times \Gamma(M,\CL) \longrightarrow \Gamma(M,\CI_0)
	\end{equation}
is the \uline{line bundle gerbe metric} on $\CL$.
\end{definition}
The line bundle gerbe metric $\frh$ inherits the functoriality and, most notably, the sesqui\-linearity of the internal 2-hom-functor.
It maps pairs of sections of $\CL$ sequilinearly into the rig category $\shom_{\CatLBGrb^\nabla(M)}(\CI_0,\CI_0)$.
Hence it can be viewed as a categorified version of a bundle metric.
It is further possible to compose $\frh$ with the equivalence given by $\shom_{\CatLBGrb^\nabla(M)}(\CI_0,\CI_0) \cong \CatHVBdl^\nabla(M)$, if one desires to use the latter as the underlying rig category.

\begin{remark}
\label{rem:multiples_of_2-bdl_metric}
  Recall that in the case of line bundles, the inner product on sections arose from the hermitean structure on the bundle. In the case of line bundle gerbes, the hermitean structure on the defining line bundle enters via the condition that the line bundle isomorphisms $\alpha$ in sections $(E,\alpha)$ is isometric. 
  
  Also, note that we could generalise the line bundle metric by a positive function as mentioned in Section~\ref{ssec:gq_outline}. In the case of line bundle gerbes, this amounts to the freedom to compose $\frh$ with the tensor product by an endomorphism $(E,\alpha)$ of $\CI_0$ for which there exists a 2-isomorphism $\Theta(E,\alpha) \twoisom (E,\alpha)$.
\end{remark}

The following result shows that the line bundle gerbe metric is canonically natural in the 2-category $\CatLBGrb^\nabla(M)$.
\begin{proposition}
	\label{st:natural_property_of_2-bdl_metric}
	Let $\Gamma_{\rm par}(M,- )$ denote the parallel section functor
	\begin{equation}
		\Gamma_{\rm par}(M,-) = \shom_{\CatHVBdl^\nabla(M)}(M \times \FC,-)~,
	\end{equation}
	and let $\sfR$ be as in Theorem~\ref{st:Red_is_an_equivalence},
	acting on $\shom_{\CatLBGrb^\nabla(M)}(\CI_0,\CI_0)$ as the descent functor to $\CatHVBdl^\nabla(M)$.
	There is a canonical natural isomorphism
	\begin{equation}
	\label{eq:natural_property_of_2-bdl_metric}
	\eta:\ \Gamma_{\rm par}(M,- ) \circ \sfR \circ \frh \twoisom \big( 2\shom_{\CatLBGrb^\nabla(M)} \big)_{|\shom_{\CatLBGrb^\nabla(M)}(\CI_0,\CL)}~.
	\end{equation}
\end{proposition}

\begin{proof}
	Let $(E,\alpha,\zeta^{Z})$ and $(E',\alpha',\zeta^{Z'})$ be two sections of $\CL$. The bundle of homomorphisms from $E$ to $E'$ can be identified with $E^* \otimes E'$, and its sections which descend to sections of $D_{\zeta^Z{\times_M}\zeta^{Z'}}(E^* \otimes E'\,)$ are those compatible with the descent datum $(\pr_Z^*\alpha^{-{\rm t}}) \otimes (\pr_{Z'}^* \alpha'\, )$.
	Restricting this to sections which are parallel with respect to the induced connections on $E^* \otimes E'$ and $D_{\zeta^Z{\times_M}\zeta^{Z'}}(E^* \otimes E'\,)$ yields precisely the 2-morphisms from $(E,\alpha,\zeta^{Z})$ to $(E',\alpha',\zeta^{Z'})$ in $\CatLBGrb^\nabla(M)$.
This only yields representatives of 2-morphisms over $Z{\times_{Y^{[2]}}}Z'$, but according to Proposition~\ref{st:2-mor_prop} from Appendix~\ref{app:special_morphisms} this does indeed cover all 2-morphisms $(E,\alpha,\zeta^{Z}) \Rightarrow (F,\beta,\zeta^{Z'})$ bijectively.
\end{proof}

\begin{remark}
	If we had defined the 2-morphisms in $\CatLBGrb^\nabla(M)$ to be general morphisms of vector bundles, compatible only with the descent datum $(\pr_Z^* \alpha^{-{\rm t}}) \otimes (\pr_{Z'}^* \alpha'\, )$, then the isomorphism of Proposition~\ref{st:natural_property_of_2-bdl_metric} would read as
	\begin{equation}
	\Gamma(M,- ) \circ \sfR \circ \frh \twoisom \big( 2\shom_{\CatLBGrb^\nabla(M)} \big)_{|\shom_{\CatLBGrb^\nabla(M)}(\CI_0,\CL)}
	\end{equation}
	in terms of the ordinary section functor.
\end{remark}

\begin{remark}
	The line bundle gerbe metric was obtained via the constructions of a dual functor $\Theta$ and an adjoint functor $\lsb-,-\rsb$ of the tensor product of bundle gerbe morphisms.
	It produces as outputs the bundles of algebras and respective bundles of modules obtained without this categorical detour in~\cite{Karoubi:1012.2512,Schweigert:2014nia}.
	However, as we restrict ourselves to parallel homomorphisms here, we only obtain the corresponding subalgebras of the Azumaya algebras constructed there.
\end{remark}

\subsection{Prequantum line bundle gerbes}\label{ssec:multisymplectic}

We can now come to a categorification of statements \eqref{it:observables}, \eqref{it:prequantum_line_bundle} and \eqref{it:tautological_line_bundle} of Section~\ref{ssec:gq_outline}.

\begin{definition}
 A \uline{multisymplectic manifold} $(M,\omega)$ is a manifold $M$ endowed with a closed differential form $\omega$, which is non-degenerate in the sense that $\iota_X \omega=0$ is equivalent to $X=0$ in $\frX(M)$. If $\omega$ is of degree $p+1$, we also call $(M,\omega)$ \uline{$p$-plectic}.
\end{definition}
A 2-plectic structure now induces a Lie 2-algebra structure on functions and Hamiltonian 1-forms \cite{Baez:2008bu}. A \emph{Hamiltonian 1-form} $\alpha\in \Omega^1_{\rm Ham}(M)$ on a 2-plectic manifold $(M,\omega)$ is a 1-form such that $\iota_{X_\alpha}\omega=-\dd \alpha$ for some \emph{Hamiltonian vector field} $X_\alpha\in \frX(M)$. We then obtain a Baez-Crans 2-vector space in $2\CatVect_{\FC\rightrightarrows\FC}$ given by 
\begin{equation}
 \Pi_{M,\omega}:=\big(\Omega^1_{\rm Ham}(M)\ltimes C^\infty(M)\rightrightarrows \Omega^1_{\rm Ham}(M)\big)
\end{equation}
with source and target maps $\sfs(\alpha_1,f_1)=\alpha_1$, $\sft(\alpha_1,f_1)=\alpha_1+\dd f_1$. This 2-vector space becomes a semistrict Lie 2-algebra when endowed with the Lie bracket functor
\begin{equation}
 [(\alpha_1,f_1),(\alpha_2,f_2)]:=([\alpha_1,\alpha_2],0)=(-\iota_{X_{\alpha_1}}\iota_{X_{\alpha_2}}\omega,0)
\end{equation}
and the Jacobiator
\begin{equation}\label{eq:jacobiator_exact}
 J_{\alpha_1,\alpha_2,\alpha_3}:[\alpha_1,[\alpha_2,\alpha_3]]~\xrightarrow{~([\alpha_1,[\alpha_2,\alpha_3]]~,~\iota_{X_{\alpha_1}}\iota_{X_{\alpha_2}}\iota_{X_{\alpha_3}}\omega)~}~[[\alpha_1,\alpha_2],\alpha_3]+[\alpha_2,[\alpha_1,\alpha_3]]
\end{equation}
for all $f_{i}\in C^\infty(M)$ and $\alpha_{i}\in \Omega^1_{\rm Ham}(M)$, $i=1,2,3$. Since
\begin{equation}
 [X_{\alpha_1},X_{\alpha_2}]=X_{[\alpha_1,\alpha_2]}~,
\end{equation}
the Lie 2-algebra $\Pi_{M,\omega}$ is a central extension of the Lie algebra of Hamiltonian vector fields. We identify $\Pi_{M,\omega}$ as the {\em Lie 2-algebra of classical observables on $(M,\omega)$}. 

The categorified analogue of the prequantum line bundle of statement \eqref{it:prequantum_line_bundle} is then readily defined.
\begin{definition}\label{eq:def_prequantum_line_bundle_gerbe}
 A 2-plectic manifold $(M,\omega)$ is \uline{quantisable} if its 2-plectic form $\omega$ represents a class in integer cohomology. A \uline{prequantum line bundle gerbe} over a quantisable 2-plectic manifold is a hermitean line bundle gerbe with curvature 3-form $H=-2\pi \, \di \, \omega$.
\end{definition}
Under a complete higher quantisation, the Lie 2-algebra
$\Pi_{M,\omega}$ should then be mapped to a quasi-2-isomorphic Lie
2-algebra of quantum observables, which act on a suitably defined
2-Hilbert space from polarised sections of a prequantum gerbe with
curvature 3-form $-2\pi\,\di\,\omega$. We will come back to this point in Section~\ref{ssec:observables}.

Analogously to statement \eqref{it:tautological_line_bundle}, we would like to construct a tautological line bundle gerbe. If $M$ is 2-connected and the 2-plectic form $\omega$ represents an element in $H^3(M,\RZ_M)$, we can construct the corresponding tautological gerbe \cite{Murray:9407015,Carey:1995yc} which takes the role of a prequantum line bundle gerbe. We first recall the following definition.
\begin{definition}
 The \uline{transgression map} $\CT:\Omega^{p+1}(M)\rightarrow \Omega^p(\Omega M)$ is $\CT(\eta):=\int_{S^1}\,\ev^* \eta$, where $\ev:\Omega M\times S^1\rightarrow M$ is the evaluation map.
\end{definition}
This map is a chain map and therefore descends to de Rham cohomology; for details on the properties of the transgression map between differential forms on a manifold and those on its mapping spaces, see Appendix~\ref{sect:mapping_space_geometry}. We can now transgress the 2-plectic form $\omega$ to a closed 2-form $\CT(\omega)$ on $\Omega M$. The corresponding tautological $0$-gerbe on $\Omega M$ defines a line bundle $L_{\CT(\omega)}$ with connection over $\Omega M$, which in turn defines a line bundle gerbe $(L_{\CT(\omega)},\partial)$ over $M$ associated to a principal $\sU(1)$-bundle gerbe from Example~\ref{ex:BGrb_from_fus_LBdl} with Dixmier-Douady class $[\omega]$. For details, see \cite{Murray:9407015,Carey:1995yc,Johnson:2003aa}.

\begin{example}
 The tautological gerbe over $M=S^3\cong \sSU(2)$, with $\frac{1}{k}\, \omega$ the unit volume form ${\rm vol}_{S^3}$ on $S^3$ for some $k\in \RZ_{>0}$, yields the loop group central extensions discussed in~\cite{Pressley:1988aa}. The bundle gerbe from Example~\ref{ex:central_loop_extension} is given by
\begin{equation}
\begin{gathered}
 \xymatrix{
\widehat{\Omega_k\sSU(2)} \ar[d] & \\
\Omega\sSU(2)\ \ar@< 2pt>[r] \ar@< -2pt>[r] & \ \CP \sSU(2) \ar[d]^\dpar \\
 & \sSU(2)
}
\end{gathered}
\end{equation}
which defines the prequantum line bundle gerbe of the 2-plectic manifold $(S^3,k\, {\rm vol}_{S^3})$ with Dixmier-Douady class $k$.
\end{example}

\section{2-Hilbert spaces from line bundle gerbes}
\label{sect:2-Hilbert_spaces_from_Bgerbes}

All the geometric structure we have found so far on the 2-category of bundle gerbes naturally ties together with the theory of {\em 2-Hilbert spaces}, as developed e.g.\ in~\cite{Baez:9609018} and, differently, in~\cite{kapranov19942}. We will use aspects from both of these approaches, but give a definition of a 2-Hilbert space which differs from these in certain respects. We will not discuss a categorified analogue of Cauchy completeness as done in \cite{Bartlett:2008ji}.

\subsection{Module categories and 2-Hilbert spaces}
\label{sect:Module_cats_and_2Hspaces}

The difference between a Hilbert space and a mere vector space is the datum of a positive-definite non-degenerate inner product, which is a sesquilinear map into the field $\FC$.
As we replace the latter by the rig category $\CatVect$, an inner product on a 2-vector space $\CCV$ should assign a vector space to a pair of objects of $\CCV$, cf.~\cite{Baez:9609018}.
Therefore, we would like to have some functor $\langle -,- \rangle$ mapping into $\CatVect$ together with natural isomorphisms
\begin{equation}
\begin{aligned}
\langle \CV, U \otimes \CW \rangle &\cong U \otimes \langle \CV, \CW \rangle~,\\[4pt]
\langle \CV, \CW \rangle &\cong \langle \CW, \CV \rangle^*~,\\[4pt]
\langle \CU, \CV \oplus \CW \rangle &\cong \langle \CU, \CV \rangle \oplus \langle \CV, \CW \rangle~,
\end{aligned}
\label{eq:innprodfunctor}\end{equation}
for all $\CU,\CV,\CW \in \CCV$, $U \in \CatVect$, where we write the module action $\CatVect\times\CCV\to\CCV$ as $(U,\CV)\mapsto U\otimes\CV$.
We further require non-degeneracy in demanding that $\langle \CV, \CV \rangle = 0$ if and only if $\CV = 0$.
Positive definiteness is obtained automatically since there are no negatives in $\CatVect$.
The second property implies that switching the arguments of $\langle -,- \rangle$ is the same as taking the dual vector space after applying $\langle -,- \rangle$, i.e.
\begin{equation}
\langle -,- \rangle \circ \sfsw \cong * \circ \langle -,- \rangle~,
\end{equation}
where $\sfsw$ is the functor which swaps the factors in a tensor product. Since $\sfsw$ is a covariant functor, we conclude from \eqref{eq:innprodfunctor} that the functor $\langle -,- \rangle$ should be a map
\begin{equation}
\langle -,- \rangle: \CCV^{\rm op} \times \CCV \longrightarrow \CatVect~.
\end{equation}

\begin{example}
	The category $\CatVect$ is endowed with an inner product functor coming from its involution $*: \CatVect^{\rm op} \rightarrow \CatVect,\, V \mapsto V^*,\, \phi \mapsto \phi^{\rm t}$. Set
	\begin{equation}
	\langle U, V \rangle := U^* \otimes V = \shom_\CatVect(U,V)~.
	\end{equation}
	Analogously, every $\CatVect^n$ is endowed with an inner product given by this functor on every summand and zero on pairs of objects from different summands.
	This can be regarded as the categorification of the canonical inner product on $\FC^n$.
	By composition with the canonical inner product, any equivalence $\CCV \isom \CatVect^n$ induces an inner product on $\CCV$.
\end{example}

\begin{remark}
Up to now, our discussion worked for left $\CatVect$-module categories, but the module structure requires a further restriction.
	The category $\CatHilb_{\rm sep}$ of separable Hilbert spaces is not generically a $\CatVect$-module category.
	This is because objects of $\CatVect$ do not carry canonical inner products, so there is no module action of $\CatVect$ on $\CatHilb_{\rm sep}$.\footnote{At least none which does not factor through an inclusion $\CatVect \hookrightarrow \CatHilb_{\rm sep}$ of {\em choosing} an inner product on every vector space.}
	
	The category $\CatHilb_{\rm sep}$ is, however, a module category over itself, enriched over $\CatVect$ (in fact $\shom_{\CatHilb_{\rm sep}}(\CH_1,\CH_2) = \CB(\CH_1,\CH_2)$ is even a Banach space), and has a canonical $\CatHilb_{\rm sep}$-valued inner product given by $(\CH_1,\CH_2) \mapsto \CH_1^* \otimes \CH_2 =  L^2(\CH_1,\CH_2)$, the Hilbert-Schmidt operators from $\CH_1$ to $\CH_2$.
\label{rem:Hilbsepcat}\end{remark}

The remark above suggests that the natural choice for a rig category for the definition of 2-Hilbert spaces is the category of Hilbert spaces $\CatHilb$.
In the following we restrict $\CatHilb$ to the category of {\em finite-dimensional} Hilbert spaces, with morphisms being the morphisms of the underlying vector spaces. We are therefore interested in elements of the 2-category $2\CatVect_\CatHilb$.
\begin{definition}\label{def:2-Hilbert_space}
	A \uline{2-Hilbert space} is a $\CatHilb$-module category $\CCH$, together with a $\CatHilb$-sesquilinear inner product functor $\langle -,- \rangle: \CCH^{\rm op} \times \CCH \rightarrow \CatHilb$. We denote the 2-category of 2-Hilbert spaces, linear functors, and natural transformations by $2\CatHilb$. 
\end{definition}	
There is an obvious forgetful functor from $2\CatVect_\CatHilb$ into $2\CatVect_\CatVect$ and therefore every 2-Hilbert space has an underlying 2-vector space. This allows the following definition.
\begin{definition}\label{def:free_2_Hilbert_space}
A \uline{free 2-Hilbert space} is a 2-Hilbert space $\CCH$ whose underlying 2-vector space is free.
\end{definition}

The various statements of Section~\ref{ssec:2-vector-spaces} concerning 2-vector spaces hold analogously for 2-Hilbert spaces.

\begin{example}
	The category $\CatHilb$ is a 2-Hilbert space with inner product given by $\langle -,- \rangle = \shom_\CatHilb(-,-)$.
	In fact, every $\CatHilb^n$ is a 2-Hilbert space as well, just as $\FC^n$ is a Hilbert space obtained as the direct sum of $n$ copies of $\FC$.
\end{example}

\begin{remark}
	\label{rmk:2Hspaces_and_Hilbert-Schmidt_ops}
	Consider the category of separable Hilbert spaces $\CatHilb_{\rm sep}$ with its canonical module action on itself via the tensor product.
	An inner product on $\CatHilb_{\rm sep}$ is given by Hilbert-Schmidt operators, $\langle \CH_1,\CH_2 \rangle = \CH_1^* \otimes \CH_2$,
	but the category $\CatHilb_{\rm sep}$ is not abelian. For instance, on an infinite-dimensional separable Hilbert space $\CH$, not every monomorphism is a kernel as there exist monomorphisms $T: \CH \rightarrow \CH$ with dense image which are not epic:
	Since $\im(T)$ is not closed, it cannot be the kernel of a bounded operator.
	
	However, since $L^2(\CH_1,\CH_2) \subset \CK(\CH_1,\CH_2) = L^\infty(\CH_1,\CH_2)$, Hilbert-Schmidt operators have the spectral properties of compact operators.
	In particular, for every $T \in  L^2(\CH,\CH)$ there exists a decomposition into its eigenspaces, all of which are closed.
	Therefore, here {\em every} object $\CH$ decomposes (non-canonically) into sub-objects $\CH_i$ of norm $\langle \CH_i, \CH_i \rangle \cong \FC$, e.g.\ consider $T$ to be a projection onto a finite-dimensional subspace of $\CH$.\footnote{Hence, finding a decomposition amounts to finding an isometric isomorphism $\CH \cong \FC^n$, or $\CH \cong \ell^2(\RZ)$.}
\end{remark}

\subsection{The 2-Hilbert space of sections of a line bundle gerbe}\label{ssec:2Hilbert_space_line_bgr}

Let us now come to the 2-Hilbert space structure on the sections of line bundle gerbes. First, we categorify statements \eqref{it:vector_space} in Section~\ref{ssec:gq_outline}. We saw in Section~\ref{sect:2-cat_of_BGrbs} that there exists a monoidal structure on the 2-category $\CatLBGrb^\nabla(M)$ for a generic manifold $M$, and Theorem~\ref{st:Direct_sum_as_fctr} says that the morphisms in this category are monoidal categories. Consider a line bundle gerbe $\CL$. With $\CI_0 \otimes \CL = \CL$ it follows that $\shom_{\CatLBGrb^\nabla(M)}(\CI_0,\CI_0) = \Gamma(M,\CI_0)$ is a rig category, and that $\shom_{\CatLBGrb^\nabla(M)}(\CI_0,\CL) = \Gamma(M,\CL)$ is a module category over $\Gamma(M,\CI_0)$, with module action given by the tensor product in $\CatLBGrb^\nabla(M)$. We can now regard a Hilbert space $V$ as a trivial hermitean vector bundle with trivial connection, which is analogous to the embedding of $\FC$ as constant functions in $C^\infty(M)$. In particular, $M \times V = (M \times V) \otimes (M \times \FC)$. This extends to an embedding of $\CatHilb$ into $\CatHVBdl^\nabla(M)$ as rig categories,
\begin{equation}
I: \CatHilb \longrightarrow \CatHVBdl^\nabla(M) \ ,\quad
\big( \phi: V \rightarrow W \big) \longmapsto \big( I(\phi) : (M \times V) \rightarrow (M \times W) \big)~.
\end{equation}
Since the morphisms in $\CatHVBdl^\nabla(M)$ are taken to be connection preserving, $I$ is a fully faithful functor.
It is moreover injective on objects, so that it provides an inclusion of $\CatHilb$ into $\CatHVBdl^\nabla(M)$ as a full subcategory.
Together with Theorem~\ref{st:Red_is_an_equivalence}, we arrive at composable morphisms of rig categories
\begin{equation}\label{eq:embedding_Hilb}
\xymatrix{
\CatHilb\ \ar@{^{(}->}[r] & \CatHVBdl^\nabla(M) \isom \Gamma(M,\CI_0)~.
}
\end{equation}

\begin{lemma}
	\begin{enumerate}
		\item The inclusion $\CatHilb \hookrightarrow \Gamma(M,\CI_0)$ is fully faithful and injective on objects. Moreover, it makes every section category $\Gamma(M,\CL)$ into a module category over $\CatHilb$.
		\item Taking global sections defines a monoidal
                  functor of 2-categories\footnote{The first functor
                    is not expected to be an equivalence of
                    categories, since the analogous functor for vector
                    bundles would only be defined on line bundles but not generic vector bundles.}
		\begin{equation}
\xymatrix{
		\Gamma(M,-): \CatLBGrb^\nabla(M) \longrightarrow 2\CatVect_{\Gamma(M,\, \CI_0)} \cong 2\CatVect_{\CatHVBdl^\nabla(M)}\ \ar@{^{(}->}[r] & 2\CatVect_{\CatHilb}~,
		}
		\end{equation}
		which is represented by the trivial line bundle gerbe $\CI_0$ on $M$.
	\end{enumerate}
\end{lemma}

\begin{proof}
	The first statement has been shown above already.
	The second statement follows from the properties of the hom-functor on $\CatLBGrb^\nabla(M)$.
\end{proof}
Thus every line bundle gerbe $\CL$ on $M$ defines a $\CatHilb$-module. It remains to construct the inner product functor to arrive at the desired 2-Hilbert space, which is the categorification of statement \eqref{it:inner_product} of Section~\ref{ssec:gq_outline}.

Recall the line bundle gerbe metric $\frh: \Gamma(M,\CL)^{\rm op} \times \Gamma(M,\CL) \rightarrow \Gamma(M,\CI_0)$ from Definition~\ref{def:2-bdl_metric}, which by Proposition~\ref{st:natural_property_of_2-bdl_metric} yields a natural isomorphism
\begin{equation}
\eta: \Gamma_{\rm par}(M,-) \circ \sfR \circ \frh \twoisom \big( 2\shom_{\CatLBGrb^\nabla(M)} \big)_{|\shom_{\CatLBGrb^\nabla(M)}(\CI_0,\CL)}~.
\end{equation}
The functor $\sfR \circ \frh$ produces a finite-rank hermitean vector bundle with connection from a pair $(E,\alpha),\, (F,\beta)$ of sections of $\CL$.

We would like to map the image of the bundle gerbe metric to an element of $\CatHilb$, just as the integral maps the result of the hermitean metric of a line bundle to an element in $\FC$, cf.\ statement \eqref{it:inner_product} of Section~\ref{ssec:gq_outline}. The natural categorified analogue would appear to be the \emph{direct integral} over the image. Recall that the direct integral over a Hilbert bundle $H\rightarrow M$ is given by the set of square-integrable sections of $H$, 
\begin{equation}
 \int_M^\oplus \, \dd \mu_M(x) \ H_x=L^2_{\dd \mu_M}(M,H)~,
\end{equation}
where $\dd \mu_M$ is some volume form on $M$. This space becomes a
Hilbert space with the natural inner product
$(\varepsilon_1,\varepsilon_2):=\int_M\, \dd \mu_M(x)\
(\varepsilon_1(x),\varepsilon_2(x))$ for $\varepsilon_{1}, \varepsilon_{2}\in L^2_{\dd \mu_M}(M,H)$. However, there are two issues with this approach. Firstly, the resulting Hilbert space will be infinite-dimensional in general. Secondly, and more severely, we expect the connection on the hermitean vector bundle in the image of the inner product functor to enter the definition of the map from $\CatHVBdl^\nabla(M)$ to $\CatVect$.

A natural cure to both problems is to restrict to parallel sections. Then, however, it is sensible to restrict to local data and we can omit the direct integral altogether. We thus define
\begin{equation}
\label{eq:inner_product_on_2Hspace_of_BGrb}
\langle -,- \rangle = \Gamma_{\rm par}(M,-) \circ \sfR \circ \frh:\ \Gamma(M,\CL)^{\rm op} \times \Gamma(M,\CL) \longrightarrow \CatVect~.
\end{equation}
The image of this map is canonically endowed with a hermitean inner product. For this, note that $\langle (E,\alpha), (F,\beta) \rangle$ is a hermitean vector bundle with hermitean metric $h_{\langle (E,\alpha), (F,\beta) \rangle}$. For two parallel sections $\varepsilon_1,\varepsilon_2$, we then define
\begin{equation}
\prec \varepsilon_1,\varepsilon_2\succ_{\langle (E,\alpha), (F,\beta) \rangle}\ :=\ h_{\langle (E,\alpha),(F,\beta) \rangle _{|x}}(\varepsilon_1,\varepsilon_2)
\end{equation}
for any $x \in M$.
This is well-defined since $\varepsilon_1$ and $\varepsilon_2$ are parallel, and the connection on $\langle (E,\alpha), (F,\beta) \rangle$ preserves $h_{\langle (E,\alpha),(F,\beta) \rangle}$.
Sesquilinearity, positive definiteness, and non-degeneracy of $\prec -,-\succ_{\langle (E,\alpha),(F,\beta) \rangle}$ follow immediately from the respective \emph{fibrewise} properties of $h_{\langle (E,\alpha), (F,\beta) \rangle}$.
Altogether, we conclude that
\begin{equation}
	\big( \langle (E,\alpha), (F,\beta) \rangle,\, \prec-,-\succ_{\langle (E,\alpha), (F,\beta) \rangle} \big)
\label{eq:innerprodmor}\end{equation}
forms a finite-dimensional Hilbert space.
This notion of inner product is compatible with the embedding \eqref{eq:embedding_Hilb} of $\CatHilb$ into $\Gamma(M,\CI_0)$.

Let us now turn to morphisms between sections. Consider two 2-morphisms $(\phi,\omega^W): (E',\alpha'\, ) \Rightarrow (E,\alpha)$ and $(\psi,\omega^U): (F,\beta) \Rightarrow (F',\beta'\, )$ in $\Gamma(M,\CL) = \shom_{\CatLBGrb^\nabla(M)}(\CI_0,\CL)$.
Then we get a morphism
\begin{equation}
\langle (\phi,\omega^W), (\psi,\omega^U) \rangle: \langle (E,\alpha), (F,\beta) \rangle \longrightarrow \langle (E',\alpha'\,), (F',\beta'\,) \rangle
\end{equation}
initially in $\CatVect$, which immediately extends to $\CatHilb$. The functor $\langle -,- \rangle$ becomes valued in $\CatHilb$ if we endow the vector spaces it generates with the inner product of \eqref{eq:innerprodmor}.
We have thus constructed a functor
\begin{equation}\label{eq:inner_product_functor}
\langle -,- \rangle = \Gamma_{\rm par}(M,-) \circ \sfR \circ \frh:\ \Gamma(M,\CL)^{\rm op} \times \Gamma(M,\CL) \longrightarrow \CatHilb~.
\end{equation}
By the $\Gamma(M,\CI_0)$-sesquilinearity of $\frh$, and the respective linearity properties of $\sfR$ and $\Gamma_{\rm par}(M,-)$, this functor is $\CatHilb$-sesquilinear.

Moreover, for every $(E,\alpha) \in \Gamma(M,\CL)$ there exists a non-zero element in\linebreak $2\shom_{\CatLBGrb^\nabla(M)}((E,\alpha),(E,\alpha))$, namely the identity morphism $\unit_{(E,\alpha)}$ on $(E,\alpha)$. Consequently, for every non-zero section of $\CL$, the dimension of the Hilbert space $\langle (E,\alpha), (E,\alpha) \rangle$ is at least $1$ and the functor \eqref{eq:inner_product_functor} is also non-degenerate. 

We can additionally show that the resulting 2-Hilbert space is free in the sense of Definition~\ref{def:free_2_Hilbert_space}. Recall that a simple object $(E,\alpha)\in \Gamma(M,\CL)$ is one whose inner product is $2\shom_{\CatLBGrb^\nabla(M)}((E,\alpha),(E,\alpha))=\FC\, \unit_{(E,\alpha)}$. 
\begin{lemma}
    \label{st:decomposition_prop_in_Gamma(M,G)}
    Normalised and simple objects in $\Gamma(M,\CL)$ coincide, and $\Gamma(M,\CL)$ is free as a 2-vector space.
\end{lemma}
\begin{proof}
    The first assertion is a consequence of the existence of the natural isomorphism $\eta: \langle -,- \rangle \twoisom \big(2\shom_{\CatLBGrb^\nabla(M)} \big)_{|\Gamma(M,\CL)}$.
    Because of Theorem~\ref{st:hom_cats_are_semisimple_Abelian}, which states that $\Gamma(M,\CL) = \shom_{\CatLBGrb^\nabla(M)}(\CI_0,\CL)$ is semisimple and abelian, every object $\CV$ of $\Gamma(M,\CL)$ decomposes (non-canonically) into objects $\CE_i$, $i\in I$, with 1-dimensional inner product: $\CV=\bigoplus_{i\in I}\, \CE_i$. Each of these objects has basis $\CB_{*,i}:=(* \rightrightarrows *)$, and the disjoint union of categories $\CCB_I=\bigsqcup_{i\in I}\, \CB_{*,i}$ forms a 2-basis for $\Gamma(M,\CL)$.
\end{proof}

We can now conclude the following statement.
\begin{theorem}\label{thm:5.9}
	\begin{enumerate}
		\item The bundle 2-metric $\frh$, or any of its multiples, on a given bundle gerbe $\CL$ defines a $\Gamma(M,\CI_0)$- (and thus also $\CatHilb$-)sesquilinear, non-degenerate functor
		\begin{equation}
		\langle -,- \rangle:\, \Gamma(M,\CL)^{\rm op} \times \Gamma(M,\CL) \longrightarrow \CatHilb
		\end{equation}
		constructed as above.
		\item This makes $\big( \Gamma(M,\CL),\, \langle -,- \rangle \big)$ into a 2-Hilbert space.
		\item The assignment to a bundle gerbe $\CL$ of the 2-Hilbert space from item (2) defines a $\otimes$-monoidal functor $\CatLBGrb^\nabla(M) \rightarrow 2\CatHilb$.
	\end{enumerate}
\end{theorem}

\begin{proof}
Items (1) and (2) follow from our discussion above together with Lemma~\ref{st:decomposition_prop_in_Gamma(M,G)}. Item~(3) follows from the fact that this functor is just the functor $\shom_{\CatLBGrb^\nabla(M)}(\CI_0,-)$ followed by endowing the objects it produces with the inner product from item (1).
\end{proof}

\begin{example}
	The category of sections $\Gamma(S^1,\CI_0)$ of the trivial bundle gerbe $\CI_0$ on the circle $S^1$ is equivalent to the category of hermitean vector bundles with connection on $S^1$.
	Normalised objects in the corresponding 2-Hilbert space are those vector bundles  whose parallel endomorphisms are all proportional to the identity.
	In particular, every flat line bundle on $S^1$ is a simple object.
	Consequently, the 2-rank of $\Gamma(S^1,\CL)$ is at least equal to the cardinality of $h_0\CatHLBdl_0^\nabla(S^1)$, the isomorphism classes of flat hermitean line bundles on $S^1$.
	The set of these classes is isomorphic as an abelian group to
        $\sU(1) \cong S^1$, see e.g.~\cite{Waldorf:2010aa}.
	Hence even over this simple compact base manifold, the 2-rank of the 2-Hilbert space $\Gamma(S^1,\CI_0)$ is uncountably infinite.
\end{example}
We shall discuss a further example of more physical significance in great detail in Section~\ref{ssec:R3_2_Hilbert_space}.

\subsection{Observables in 2-plectic quantisation}\label{ssec:observables}

For a full quantisation, we are missing a notion of polarisation that we can apply on the sections forming the objects of the 2-Hilbert space. Nevertheless, we can make a few statements about the categorification of statement \eqref{it:rep_observables} of Section~\ref{ssec:gq_outline}. Ideally, the 2-vector space of observables $\Pi_{M,\omega}$ on a 2-plectic manifold $(M,\omega)$ as described in Section~\ref{ssec:multisymplectic} would be recovered as certain sections of the trivial line bundle gerbe $\CI_0$ over $M$: We expect an equivalence of categories between a subcategory of $\Gamma(M,\CI_0)\cong \CatHVBdl^\nabla(M)$ and the 2-vector space $\Pi_{M,\omega}= (\Omega^1_{\rm Ham}(M) \ltimes C^\infty(M) \rightrightarrows \Omega^1_{\rm Ham}(M))$. This is true to some extent as we discuss in the following.

The restriction of $\Pi_{M,\omega}$ to Hamiltonian 1-forms can be
eliminated by focusing on 2-plectic manifolds with dimensions a
multiple of $3$.\footnote{Another approach would be to take the
  proposal of \cite{Ritter:2015ffa} seriously and generalise
  quantisation of 2-plectic manifolds to quantisation of higher
  spaces, in particular appropriate Courant algebroids. This, however,
  is beyond the scope of the present paper.}

\begin{proposition}\label{prop:5.11}
As a category, the Lie 2-algebra $\Pi_{M,\omega}$ can be embedded in $\CatHVBdl^\nabla(M)$ as the full subcategory of topologically trivial hermitean line bundles with connection. Under this embedding, the additive abelian monoidal structure on $\Pi_{M,\omega}$ is taken to the multiplicative abelian monoidal structure $\otimes$ on $\CatHVBdl^\nabla(M)$.
\end{proposition}
\begin{proof}
 We map Hamiltonian 1-forms $\alpha\in \Omega^1_{\rm Ham}(M)$ to connection 1-forms on trivial hermitean line bundles with canonical bundle metric and the image of a morphism $(\alpha,f)\in\Omega^1_{\rm Ham}(M) \ltimes C^\infty(M)$ is mapped to a global gauge transformation with group element $\de^{-\di\, f}$. The second statement is then evident.
\end{proof}

As we know from Proposition~\ref{prop:morphisms_weakly_invertible}, however, this subcategory of line bundles encodes 1-automorphisms of $\CI_0$. We do not expect all higher endomorphisms of $\CI_0$ to correspond to automorphisms. Recall that a functor only yields an equivalence of categories if it is full, faithful and essentially surjective. As is readily seen, the embedding functor from $\Pi_{M,\omega}$ to $\CatHVBdl^\nabla(M)$ fails to be essentially surjective, and hence does not yield an equivalence of categories.

\section{Example: Higher prequantisation of \texorpdfstring{$\FR^3$}{R3}}
\label{sect:action_of_2-grps_on_their_BGrbs}

Let us now discuss an example in great detail. The most interesting
spaces from a string theory perspective are the 3-manifolds $\FR^3$,
$S^3$ and $T^3$. Since our framework is restricted to the case of
Dixmier-Douady classes which are torsion, we will focus on the case of
$\FR^3$: Here the degree~$3$ cohomology is trivial and the cohomology class of any 2-plectic
form, in particular the volume form ${\rm vol}_{\FR^3}$, is $1$-torsion. We start with a
general discussion, and then specialise to $\FR^3$.

\subsection{The 2-Hilbert space of a trivial prequantum line bundle gerbe}\label{ssec:general_trivial_pqlbg}

Consider the trivial prequantum line bundle gerbe $\CI_\rho$ over some manifold $M$ as defined in \eqref{eq:trivial_line_bundle_gerbe}. The sections of $\CI_\rho$ form a module over $\shom_{\CatLBGrb^\nabla(M)}(\CI_0,\CI_0) = \Gamma(M,\CI_0)$. As explained in Example~\ref{ex:equiv_sec_trivial_hvbld}, every section is 2-isomorphic to a hermitean vector bundle with connection on $M$, since
\begin{equation}
	\shom_{\CatLBGrb^\nabla(M)}(\CI_\rho,\CI_{\rho'}) \cong
        \shom_{\CatLBGrb^{\nabla}_{\rm FP}(M)}(\CI_\rho,\CI_{\rho'}) \cong \CatHVBdl^\nabla(M)~.
\end{equation}
In the latter smaller category, the morphisms are just the parallel
morphisms of vector bundles, and the direct sum is the Whitney sum of
vector bundles. In the subcategory $\CatLBGrb_{{\rm
    triv,FP}}^{\nabla}(M)$ of trivial line bundle gerbes (cf.\ Remark~\ref{rem:trivial_line_bundle_gerbes}), composition of 1-morphisms coincides with the tensor product of vector bundles. We therefore do not require additional surjective submersions in the 1-morphisms in order to achieve associativity of composition.
	
The Riesz functor $\Theta$ acts on $\CatLBGrb_{{\rm triv,FP}}^{\nabla}(M)$ by mapping a bundle gerbe to its dual bundle gerbe, a hermitean vector bundle with connection to its dual bundle, and a morphism of vector bundles to its transpose.
Note that $\CI_\rho^* = \CI_{-\rho}$.
Therefore the internal 2-hom-functor from Definition~\ref{def:internal_hom} reduces here to the usual internal hom-functor on $\CatHVBdl^\nabla(M)$,
\begin{equation}
  \lsb E,F\rsb = E^* \otimes F~,
\end{equation}
and already coincides with the line bundle gerbe metric, $\frh = \lsb -,-\rsb$.
The natural isomorphism $\shom_{\CatHVBdl^\nabla(M)}(E,F) \cong
\Gamma_{\rm par}(M,E^* \otimes F)$ for vector bundles then establishes
the natural isomorphism from
Proposition~\ref{st:natural_property_of_2-bdl_metric}.

The 2-Hilbert space $\Gamma(M,\CI_\rho)$ is thus equivalent to $\CatHVBdl^\nabla(M)$ endowed with the usual direct sum and module action of $\CatHVBdl^\nabla(M)$ given by the tensor product of vector bundles.
This induces a module action of $\CatHilb$ via the assignment to a
Hilbert space of the corresponding trivial hermitean vector bundle with connection over $M$.
The inner product of two objects is given by the vector space of all parallel morphisms between the respective vector bundles, 
\begin{equation}
	\langle E,F \rangle = \Gamma_{\rm par}(M,E^* \otimes F)~.
\end{equation}
This vector space is a Hilbert space with respect to the inner product
$\prec\varepsilon_1,\varepsilon_2\succ_{\langle E,F\rangle} = h_{E^* \otimes F_{|x}}(\varepsilon_1,\varepsilon_2)$ for an arbitrary $x \in M$.
Recall that $\prec\varepsilon_1,\varepsilon_2\succ_{\langle E,F\rangle} $ is independent of the choice of $x \in M$, since we require the morphisms to be parallel and the connections on the bundles to be metric-preserving.

\subsection{The 2-Hilbert space over \texorpdfstring{$\FR^3$}{R**3}}\label{ssec:R3_2_Hilbert_space}

Let now $\omega={\rm vol}_{\FR^3}=\dd x^1\wedge \dd x^2\wedge \dd x^3$
be the volume form on $M=\FR^3$ in standard cartesian coordinates
$(x^1,x^2,x^3)$. Then $(\FR^3,\omega)$ is a 2-plectic manifold. The Lie
2-algebra of observables $\Pi_{\FR^3,\omega}$ consists of the 2-vector
space $ \Omega^1(\FR^3)\ltimes C^\infty(\FR^3) \rightrightarrows
\Omega^1(\FR^3)$; in particular, all $1$-forms are Hamiltonian. The source and target maps in $\Pi_{\FR^3,\omega}$ are
\begin{equation}
\sfs(\alpha,f)=\alpha~,~~~\sft(\alpha,f)=\alpha+\dd f
\end{equation}
for $(\alpha,f)\in \Omega^1(\FR^3)\ltimes C^\infty(\FR^3)$. The
Hamiltonian vector field of $\alpha=\alpha_i\, \dd x^i$ reads
as\footnote{Throughout we use implicit summation over repeated upper
  and lower indices, and $\epsilon^{ijk}$ denotes the totally
  antisymmetric tensor of rank~$3$ with $\epsilon^{123}=+1$. We also
  sometimes abbreviate $\partial_i:=\der{x^i}$.}
\begin{equation}
 X_\alpha=X^i_\alpha \, \der{x^i}=-\epsilon^{ijk}\, \dpar_j\alpha_k\, \der{x^i}~,
\end{equation}
leading to the Lie bracket functor and Jacobiator
\begin{equation}
\begin{aligned}\label{eq:R^3_products}
 [(\alpha,f),(\beta,g)]&=([\alpha,\beta],0)=(\epsilon^{ijk}\,
 \dpar_i\alpha_k\, (\dpar_j\beta_l-\dpar_l\beta_j)\, \dd x^l,0)~,\\[4pt]
 J_{\alpha,\beta,\gamma}&=\big([\alpha,[\beta,\gamma]]~,~\epsilon^{ijk}\,
 \epsilon^{mnp}\, \dpar_m\alpha_n\, \dpar_j\beta_k\, (\dpar_i\gamma_p-\dpar_p\gamma_i)\big)~.
\end{aligned}
\end{equation}

Consider now the standard trivial line bundle gerbe $\CI_\rho$ with
curving $\rho\in \di\, \Omega^2(\FR^3)$ given by $\rho=-\tfrac{2\pi\, \di}{3!}\, \epsilon_{ijk} \, x^i\, \dd
x^j\wedge \dd x^k$. This is a prequantum line
bundle gerbe over $(\FR^3,\omega)$ according to Definition~\ref{eq:def_prequantum_line_bundle_gerbe}. Each of its sections in $\Gamma(\FR^3,\CI_\rho) = \shom_{\CatLBGrb^\nabla(\FR^3)}(\CI_0,\CI_\rho)$ consist of a surjective submersion $Y \rightarrow \FR^3$ together with a hermitean vector bundle with connection $E \rightarrow Y$ and a parallel isometric isomorphism $\alpha_{(y_1,y_2)}: E_{y_2} \isom E_{y_1}$ for all $(y_1,y_2)\in Y^{[2]}$. Over $(y_1,y_2,y_3)\in Y^{[3]}$ we have $\alpha_{(y_1,y_2)}\circ \alpha_{(y_2,y_3)}=\alpha_{(y_1,y_3)}$.

According to the discussion of
Section~\ref{ssec:general_trivial_pqlbg}, the category of sections of
$\CI_\rho$ is equivalent to the category of hermitean vector bundles with connection on $\FR^3$.
Since $\FR^3$ is contractible, every hermitean vector bundle on $\FR^3$ is isomorphic to the trivial hermitean vector bundle of the same rank.
Upon choosing such a trivialisation for every bundle, we conclude that
$\Gamma(\FR^3,\CI_\rho)$ is equivalent to the category having as
objects $\fru(n)$-valued 1-forms on $\FR^3$, for any $n \in \RZ_{>0}$, and
as morphisms $f: \omega \rightarrow \eta$ maps $f \in
C^\infty(\FR^3,\MM_{n,m}(\FC))$ valued in $n{\times}m$ complex
matrices such that $f \, \omega = \eta \, f +\dd f $.
Direct sum and tensor product reduce canonically to this category, and
the dual of a section is given by $\omega \mapsto -\omega^{\rm t}$.
The internal hom-functor on two sections is given by $\lsb \omega,
\eta \rsb = -\omega^{\rm t} \otimes \mathbbm{1} + \mathbbm{1} \otimes \eta$.
Hence the 2-Hilbert space inner product functor
$\langle\omega,\eta\rangle$ is canonically isomorphic to the vector space of all morphisms $f: \omega \rightarrow \eta$.
The Hilbert space inner product on the latter vector space is given by $\prec f,g\succ_{\langle\omega,\eta\rangle} = \tr(f^*(x)\, g(x))$ for an arbitrary $x \in \FR^3$. This category can be regarded as a higher-rank analogue of the Lie 2-algebra $\Pi_{\FR^3,\omega}$.

\subsection{Symmetries and the string 2-group}\label{ssec:symmetries}

Let us now work out the induced action of the spacetime isotropy group
$\sSO(3)$ on the 2-Hilbert space $\Gamma(\FR^3,\CI_\rho)$,
categorifying statement~\eqref{it:symmetry_group_action} of
Section~\ref{ssec:gq_outline}. The appropriate higher
analogue of the spin group is called the string group. There are
strict Lie 2-group models of the string group which indeed act on the 2-Hilbert space $\Gamma(\FR^3,\CI_\rho)$ described in Section~\ref{ssec:R3_2_Hilbert_space}. Instead of developing a full theory of equivariant prequantum line bundle gerbes, let us merely describe the 2-group action within our framework.

The string group of $\sSO(n)$ is the 3-connected cover of $\sSpin(n)$. This defines the group only up to homotopy, and there are various models for the string group, cf.\ \cite{Stolz:1996:785-800,Stolz:2004aa,Nikolaus:2011zg}.
\begin{definition}
 A \uline{smooth string group model} for a simple Lie group $\sG$ is a Lie group $\hat{\sG}$ together with a smooth homomorphism $q:\hat \sG \rightarrow \sG$ such that $\pi_k(\hat \sG)=0$ for $k\leq 3$ and $\pi_i(\hat \sG)\cong \pi_i(\sG)$ for $i>3$.
\end{definition}
A smooth string group model cannot be finite-dimensional. The situation is improved by looking at 2-group models, cf.\
\cite{Nikolaus:2011zg}. Let us give a slightly simpler definition
here. There is a canonical semistrict Lie 2-algebra $\aso(n)\oplus
\FR\rightrightarrows \aso(n)$ with source and target maps the trivial
projection, Lie bracket functor the usual Lie bracket and Jacobiator
given by a Lie group 3-cocycle
$J(g_1,g_2,g_3):=\big([g_1,[g_2,g_3]],k\, \langle
g_1,[g_2,g_3]\rangle\big)$ for $g_{i}\in \aso(n)$, $i=1,2,3$, and
$k\in\FR$, where $\langle-,-\rangle$ is the Killing form on $\aso(n)$
\cite{Baez:2003aa}. For $n=3$ and $n\geq 5$, $H^3(\sSpin(n))\cong
\RZ$, and we choose the group cocycle $k\, \langle g_1,[g_2,g_3]\rangle$ to be the generator $k=1$ of $H^3(\sSpin(n))$. In these cases, the semistrict Lie 2-algebra is called the \emph{string Lie 2-algebra}. We then define the following.
\begin{definition}
 A \uline{string 2-group model} is a Lie 2-group whose Lie
 2-algebra is equivalent to the string Lie 2-algebra.
\end{definition}
The model we shall be interested in here is the strict 2-group model
of \cite{Baez:2005sn}. Let $\Omega\sG$ and $\CP  \sG$ denote again the
loop and path spaces of some compact $1$-connected simple Lie group $\sG$, respectively,
based at the identity element $\unit\in\sG$. These fit into the short exact sequence 
\begin{equation}\label{eq:ses1}
 \unit\longrightarrow \Omega \sG\xrightarrow{~\sigma~} \CP \sG\xrightarrow{~\dpar~} \sG\longrightarrow \unit~,
\end{equation}
where $\sigma$ is the embedding and $\dpar$ is again the evaluation
map at the endpoint of a loop. Recall from
Example~\ref{ex:central_loop_extension} that there is a non-trivial central extension of $\Omega \sG$ by $\sU(1)$, giving the Kac-Moody group $\widehat{\Omega_k\sG}$ as a non-trivial principal $\sU(1)$-bundle over $\Omega \sG$: 
\begin{equation}\label{eq:ses2}
 1 \longrightarrow \sU(1)\longrightarrow \widehat{\Omega_k\sG} \xrightarrow{~\phi~} \Omega \sG\longrightarrow  \unit~.
\end{equation}
At the level of Lie algebras, this central extension is well under
control and can be lifted to Lie groups by using a
theorem\footnote{Here enters the condition that $\sG$ is simple: It
  guarantees that all invariant symmetric bilinear forms on its Lie
  algebra are proportional to each other.} due to Pressley and Segal~\cite{Pressley:1988aa}. 
\begin{proposition}[\cite{Baez:2005sn}]
 There is a string 2-group model which is given by the crossed module of Lie groups $\widehat{\Omega_1\sG}\xrightarrow{~\theta~}\CP \sG$, where the homomorphism $\theta$ is the composition $\sigma\circ \phi$ of homomorphisms from the short exact sequences \eqref{eq:ses1} and \eqref{eq:ses2}.
\end{proposition}

To understand the action of this strict string 2-group model on the
2-Hilbert space constructed in Section~\ref{ssec:R3_2_Hilbert_space},
we present the trivial line bundle gerbe $\CI_\rho$ over $\FR^3$ in a
slightly different way. We regard $\FR^3$ as the coset space
$\sSpin(3)\ltimes \FR^3/\sSpin(3)$ and choose the surjective
submersion $(\sigma^Y: Y\thra \FR^3)= (\CP (\sSpin(3)\ltimes \FR^3)\xthra{~\pi~} \sSpin(3)\ltimes \FR^3/\sSpin(3))$, where $\pi$ is the map to the orbit of the endpoint of a path. Then $Y^{[2]}=\Omega(\sSpin(3)\ltimes \FR^3)$, yielding a trivial $\sU(1)$-bundle gerbe $(P,\sigma^Y)$ with $P=\tau^*\widehat{\Omega_1\sSpin(3)}$, where we pulled back the $\sU(1)$-bundle defined by the Kac-Moody group along the trivial projection $\tau:\Omega(\sSpin(3)\ltimes \FR^3)\rightarrow \Omega \sSpin(3)$: 
  \begin{equation}\label{eq:symmetry_bdl_gerbe}
    \xymatrixcolsep{0.4cm}
    \myxymatrix{
	  \widehat{\Omega_1\sSpin(3)} \ar@{->}[d]
          &P=\tau^*\widehat{\Omega_1\sSpin(3)}\cong
          \widehat{\Omega_1\sSpin(3)} \ltimes \Omega \FR^3  \ar@{->}[d] & \\
	  \Omega\sSpin(3) & \Omega(\sSpin(3)\ltimes \FR^3)\ar@<1.5pt>[r] \ar@<-1.5pt>[r] \ar@{->}[l]_{\tau} & \CP (\sSpin(3)\ltimes \FR^3) \ar@{->}[d]^{\pi} \\
	    & & \FR^3\cong \sSpin(3)\ltimes \FR^3/\sSpin(3)
  }
  \end{equation}
The associated line bundle to $(P,\sigma^Y)$ with respect to the fundamental
representation $\varrho_{\rm f}$ of $\sU(1)$ defines a line bundle gerbe which is stably isomorphic to $\CI_\rho$.

The action of the string 2-group is now readily read off, using the pointwise products in path and loop space as well as the product in the crossed module of Lie groups forming the string 2-group model.
\begin{lemma}\label{lem:action}
The crossed module of Lie groups $\widehat{\Omega_1\sSpin(3)}\xrightarrow{~\theta~}\CP \sSpin(3)$ acts on $Y$ and $P$. The action on $Y$ is given by
 \begin{equation}
  g_1(g_2,p):=(g_1\, g_2,g_1 p) \ \in \ Y
 \end{equation}
 for $p\in \CP \FR^3$ and $g_{1},g_{2}\in \CP  \sSpin(3)$, which
 induces an action on $Y^{[2]}$. The latter action is covered by the action
 on the $\sU(1)$-bundle $P$ given by
\begin{equation}
 h_1(h_2,\ell):=(h_1 \, h_2,\phi(h_1)\, \ell) \ \in \ P
\end{equation}
for $\ell\in\Omega\FR^3$ and $h_1,h_2\in \widehat{\Omega_1\sSpin(3)}$.
\end{lemma}
This action on the $\sU(1)$-bundle gerbe $(P,\sigma^Y)$ over $\FR^3$ induces an action on the associated line bundle gerbe by the fundamental representation, which in turn induces an action on the corresponding 2-Hilbert space.
\begin{theorem}\label{thm:6.5}
 The strict string 2-group $\widehat{\Omega_1\sSpin(3)} \xrightarrow{~\theta~}\CP \sSpin(3)$ acts on the 2-Hilbert space obtained from sections of the line bundle gerbe associated to the $\sU(1)$-bundle gerbe $(P,\sigma^Y)$.
\end{theorem}

\section{The transgression functor}\label{sect:Transgression}

In this section we describe the transgression to loop space of the
various structures we have constructed on the 2-category
$\CatLBGrb^\nabla(M)$ of line bundle gerbes, and how it is used
to further our construction of higher geometric prequantisation.

\subsection{Overview}

Let us begin by briefly recalling some general facts about holonomy,
cf.\ \cite{Baez:2004in}. Consider, for simplicity, a trivial principal $\sU(1)$-bundle with
connection $(P,\nabla)$ over a manifold $M$. The path groupoid $\CCP
M$ of $M$ has as its set of objects the manifold $M$ and as its
morphisms parameterised paths between the points of $M$. Composition
of paths is well-defined up to some technical details, such as the
introduction of ``sitting instances''. The connection $\nabla$ on $P$
then induces a holonomy functor from $\CCP M$ to $\sB \sU(1)$ which
assigns to every path an element of $\sU(1)$ and to composed paths the
product of their group elements; it determines a $\sU(1)$-valued
function on the path space $\CP M$.\footnote{For non-trivial $\sU(1)$-bundles with connection $(P,\nabla)$, one would have to replace the target category $\sB \sU(1)$ with the category of $\sU(1)$-torsors~\cite{Baez:2004in}.} We can reduce the space of paths to loops without losing any information.

Let us now sketch the generalisation of this picture to higher
principal $n$-bundles on a manifold $M$. Here we consider the obvious
categorification $\CCP_n M$ of the path groupoid $\CCP M$ with higher
morphisms between paths corresponding to homotopies and homotopies
between homotopies. A connective structure on a
$\sB^{n-1}\sU(1)$-principal $n$-bundle then yields an $n$-functor from
$\CCP_n M$ to $\sB^n\sU(1)$ which assigns to every $(n-1)$-homotopy an
element of $\sU(1)$. We can again restrict to higher homotopies described by submanifolds without boundaries.

Instead of looking at higher homotopies, which are maps from the path
$n$-homotopies to $\sU(1)$, we can regard these as $(n-1)$-functors
from the $(n-1)$-groupoid of path $(n-1)$-homotopies on the path space
$\CP M$. It is not hard to show that such an $(n-1)$-functor defines a
principal $(n-1)$-bundle over the space of closed paths, cf.\ the
discussion in Section~\ref{ssec:gerbes}; that is, a principal
$n$-bundle with connective structure on $M$ gives rise to a principal
$(n-1)$-bundle on the free loop space $\CL M$. This transition from $M$ to
its loop space $\CL M$ is what is usually called {\em transgression}.

Let us now focus on the $\sU(1)$-bundle gerbes corresponding to
$n=2$. The {\em holonomy} of a gerbe with connective structure
$(P,\sigma^Y,A,B)$ is a function from $\CP\CP M$ to $\sU(1)$ and its
{\em transgression} is a principal $\sU(1)$-bundle with connection
over $\CL M$. 
Several constructions have been found which produce a transgression of a gerbe to loop space.
The first approach, due to Gawedzki~\cite{Gawedzki:1987ak}, used triangulations of surfaces to define higher-dimensional holonomy.
Starting from the stacky definition of a gerbe, Brylinski gave a construction which can be found in~\cite{0817647309}.
This inspired the related procedure for bundle gerbes elaborated on by
Waldorf in~\cite{Waldorf:0804.4835,Waldorf:2010aa}.
Both of these constructions yield a principal $\sU(1)$-bundle with connection on $\CL M$.
A transgression construction producing a line bundle on loop space has
been outlined by Carey, Johnson and Murray in~\cite{Carey:2002xp}.
Each of these constructions is functorial in a certain sense.

However, each of these constructions of a transgression functor has certain disadvantages.
While in particular cases the explicitness of Gawedzki's approach may
be welcome, the choices of triangulations make proofs rather
complicated, and tend to conceal the global categorical picture behind indices and local expressions.
The construction of Waldorf avoids both these problems, but the
considerations are limited to the 2-groupoid underlying the 2-category
$\CatLBGrb^\nabla(M)$.
Finally, the construction outlined by Carey, Johnson and Murray yields
a functor defined on the full 2-category of bundle gerbes on $M$, but
it is limited to simply-connected spaces and does not, so far, include the construction of a connection on the transgression line bundle.
Therefore it does not work on spaces like the $3$-torus, or generic
principal torus bundles, which would be desirable in constructions
related to T-duality in string theory.
Generally speaking, all considerations so far have not dealt with the
internal structures of the categories $\CatLBGrb^\nabla(M)$ and
$\CatHLBdl^\nabla(\CL M)$, such as the Riesz duals of morphisms, the
additive monoidal structures, or the bundle metrics.

Using the methods and results accumulated so far, we develop a transgression functor which is a hybrid of the latter two constructions.
This extends the transgression defined by Waldorf to the full
2-category of bundle gerbes, thus no longer producing a circle bundle
but rather a line bundle on loop space.
At the same time, it will extend the transgression as defined by
Carey, Johnson and Murray to include a connection and hermitean metric
on their line bundle, while working on multiply-connected spaces as well.
Furthermore, we will see that the line bundle gerbe metric from
Section~\ref{sect:2-bundle_with_2-metric} naturally induces a
hermitean metric on the transgression line bundle, that the Riesz
functor fits nicely into this framework, and that the direct sum of morphisms transgresses to the sum of the transgressed morphisms.

\subsection{Holonomy of bundle gerbes and their sections}
\label{sect:bgrb_hol}

We begin by recalling the definitions of holonomy of a line bundle gerbe and of D-brane holonomy.
We will not make explicit use of these notions in the sequel, but they serve as motivation for the content of the remainder of this section.

\begin{definition}[\cite{Waldorf:2007aa}]
	\label{def:Hol_of_BGrb}
	Let $\CL = (L,\mu,h,\nabla^L,B,\sigma^Y)$ be a line bundle gerbe over $M$, and $f: \Sigma \rightarrow M$ a smooth map from a closed oriented surface $\Sigma$ into $M$.
	This induces maps $\hat{f}: f^*(Y) \rightarrow Y$ covering $f$ on the bases.
	The pullback bundle gerbe
	\begin{equation}
		f^*\CL = \big(\hat{f}^{[2]*} L,\, \hat{f}^{[3]*}
                \mu,\, \hat{f}^{[2]*} h,\, \nabla^{\hat{f}^{[2]*}
                  L},\, \hat{f}^* B ,\, \sigma^{f^*(Y)} \big)
	\label{eq:pullback-bundle-gerbe}\end{equation}
	is trivialisable%
	\footnote{
		Recall that we define a trivialisation to be a {\em
                  flat} isomorphism, cf.\ Definition~\ref{def:flat_isomps_and_trivialisations}.
	}
	since $H^3(\Sigma,\RZ_\Sigma) = 0$ for dimensional reasons.
	Let $(T,\beta,g,\nabla^T,\zeta^Z) \in \sisom_{\CatLBGrb^\nabla(\Sigma)}(f^*\CL,\CI_\rho)$ be a trivialisation of $f^*\CL$.
	Then the \uline{surface holonomy} of $\CL$ around $(\Sigma,f)$ is
	\begin{equation}
	\hol_\CL(\Sigma, f) = \exp \Big( - \int_\Sigma\, \rho \Big)~.
	\end{equation}
\end{definition}

This definition is independent of the choice of trivialisation due to
item (4) of Theorem~\ref{st:BGrbs_and_Deligne_coho}: Given another
trivialisation in $\sisom_{\CatLBGrb^\nabla(\Sigma)}(f^*\CL, \CI_{\rho'})$ with $\rho'-\rho\neq 0$, the difference is a closed
$2$-form which vanishes mod $2\pi\,\di\,\RZ$ upon integration over the closed surface $\Sigma$.

Let us now come to D-brane holonomy. A {\em D-brane} in $M$ is an
embedded submanifold $\imath: Q \embd M$ together with a section
$(E,\alpha,\zeta^Z): \CI_\omega \rightarrow \imath^*\CL$ of $\CL$ over
$Q$. We call $Q$ the {\em worldvolume} of the D-brane and $\omega$ its {\em curvature}.

Let $f: D^2 \rightarrow M$ be a smooth map from the disc $D^2$ into $M$ with restriction $\partial f = f_{\vert \partial D^2}$ to the boundary of $D^2$. Assume that $\partial f: \partial D^2 \rightarrow Q$, i.e.\ that the image of the boundary of $D^2$ under $f$ is contained in $Q$.
In this situation, we first pull back the gerbe $\CL$ to $D^2$ along
$f$ to obtain $f^*\CL$ analogously to \eqref{eq:pullback-bundle-gerbe}. We also pull back $(E,\alpha,\zeta^Z)$ to $\partial D^2$,
\begin{equation}
(\dpar f)^* (E,\alpha,g,\nabla^E,\zeta^Z) = \big( \fpmap{}{\widehat{\dpar f}}{*}{} E,\, \fpmap{}{\widehat{\dpar f}}{[2]*}{} \alpha,\, \fpmap{}{\widehat{\dpar f}}{*}{} g,\, \nabla^{\fpmap{}{\widehat{\dpar f}}{*}{} E},\, \fpmap{}{\zeta}{\fpmap{}{\widehat{\dpar f}}{*}{} (Z)}{} \big)~.
\end{equation}
Since $H^3(D^2,\RZ_{D^2}) = 0$, there exists a trivialisation $(T, \beta, \zeta^W) \in \sisom_{\CatLBGrb^\nabla(D^2)}(f^*\CL, \CI_\rho)$.
The composition $(T, \beta)|_{\partial D^2} \bullet (\dpar f)^*(E,\alpha) \in
\shom_{\CatLBGrb^\nabla(\partial D^2)}(\CI_\omega,\CI_\rho)$ defines a vector
bundle $\sfR ((T, \beta) \bullet (\dpar f)^*(E,\alpha)) \in
\CatHVBdl^\nabla(\partial D^2)$, where $\sfR$ is the equivalence of categories $\sfR:
\shom_{\CatLBGrb^\nabla(\partial D^2)}(\CI_\omega, \CI_{\rho}) \isom
\CatHVBdl^\nabla(\partial D^2)$ from
Theorem~\ref{st:Red_is_an_equivalence}.

\begin{definition}[\cite{Waldorf:2007aa,Carey:2002xp,Gawedzki:2002se}]
	The \uline{(D-brane) holonomy} of $(\CL,E,\alpha)$ around $f$ is 
	\begin{equation}
	\label{eq:BGM_hol}
	\hol_{(\CL,E,\alpha)}(D^2, f) := \tr \big( \hol_{\sfR ((T,
          \beta) \bullet (\partial f)^*(E,\alpha))}(\partial D^2) \big) \ \exp
        \Big( - \int_{D^2} \, \rho \Big)~.
	\end{equation}
\end{definition}

This is, again, independent of the choice of trivialisation as we shall show now.
If $(S,\delta) \in \sisom_{\CatLBGrb^\nabla(D^2)}(f^*\CL, \CI_{\rho'})$ is another trivialisation of $f^* \CL$, then
\begin{equation}
(S,\delta) \bullet (\dpar f)^*(E,\alpha) \cong (S,\delta) \bullet (T, \beta)^{-1} \bullet (T, \beta) \bullet (\dpar f)^*(E,\alpha)
\end{equation}
and therefore
\begin{equation}
\begin{aligned}
	\sfR \big( (S,\delta) \bullet (\dpar f)^*(E,\alpha) \big) &\cong \sfR \big( (S,\delta) \bullet (T, \beta)^{-1} \bullet (T, \beta) \bullet (\dpar f)^*(E,\alpha) \big)\\[4pt]
	&\cong \sfR \big( (S,\delta) \bullet (T, \beta)^{-1} \big) \otimes \sfR \big( (T, \beta) \bullet (\dpar f)^*(E,\alpha) \big)~,
\end{aligned}
\end{equation}
where the first factor is a line bundle $J$ with curvature $F_{\nabla^J}= \rho' - \rho$.
Since the holonomy of a tensor product is the tensor product of the holonomies, and since isomorphic vector bundles have identical traced holonomies (Wilson loops), we obtain
\begin{equation}
	\tr \big( \hol_{\sfR ((S, \delta) \bullet (\partial f)^*(E,\alpha))}(\partial D^2) \big)
	= \hol_J(\partial D^2) \ \tr \big( \hol_{\sfR ((T, \beta) \bullet (\partial f)^*(E,\alpha))}(\partial D^2) \big)~.
\end{equation}
As $D^2$ is contractible, we may write $\hol_J(\partial D^2) = \exp( -
\int_{\partial D^2}\, \varepsilon^*A)$ for a global section $\varepsilon$ of the frame bundle of $J$.
By Stokes' theorem we obtain $\hol_J(\partial D^2) = \exp( -
\int_{D^2} \, \varepsilon^*F_{\nabla^J})$.
Using this, we may rewrite
\begin{equation}
\hol_J(\partial D^2) = \exp \Big(\, \int_{D^2} \, (\rho - \rho'\,) \, \Big)~.
\end{equation}
Combining this with the second factor in~\eqref{eq:BGM_hol} yields the invariance of $\hol_{(\CL,E,\alpha)}(D^2,f)$ under the change of trivialisation of $f^*\CL$.
Note that it is really only the {\em combination} of those two factors which is invariant, since $D^2$ is not a closed manifold.

\subsection{Transgression of bundle gerbes}
\label{sect:Transgression_definition}

For a line bundle gerbe $\CL$ on $M$, the $\sU(1)$ fibre over a loop $\gamma \in \CL M$ of the bundle on
loop space constructed by Waldorf is given by the set of 2-isomorphism
classes of trivialisations of $\gamma^*\CL$. This set is non-empty,
since line bundle gerbes over the circle are flat isomorphic to
$\CI_0$. Given any two such trivialisations, say $(T, \beta),\,
(T',\beta'\, ): \gamma^*\CL \isom \CI_0$, we have
\begin{equation}
	 (T, \beta) \cong (T,\beta) \bullet (T',\beta'\, )^{-1} \bullet (T' , \beta'\, )
\end{equation}
in $\sisom_{\CatLBGrb^\nabla(S^1)}(\gamma^*\CL,\CI_0)$.
Therefore $(T, \beta) \bullet (T',\beta'\, )^{-1}$ descends to a flat
hermitean line bundle $\sfR ((T, \beta) \bullet (T',\beta'\, )^{-1}) \in \CatHLBdl^\nabla_0(S^1)$ on the circle.
Isomorphism classes of flat hermitean line bundles on the circle are in turn classified by elements in $H^1(S^1,\sU(1)) = \sU(1)$, see Proposition~\ref{st:Deligne_ex_seq_2} and~\cite{Waldorf:2007aa,Taubes:2011aa}.
A classifying group isomorphism is given by
\begin{equation}
 \hol: h_0 \CatHLBdl^\nabla_0(S^1) \isom \sU(1)~.
\end{equation}
This means that the holonomy map $\hol$ above factors through the projection
to isomorphism classes, whence the desired group isomorphism is given
by the dashed arrow in the diagram
\begin{equation}
	\myxymatrix{
		\CatHLBdl^\nabla_0(S^1) \ar@{->}[r]^-{\hol}	\ar@{->}[d]_-{q}	&	\sU(1)\\
		h_0 \CatHLBdl^\nabla_0(S^1) \ar@{-->}[ur]
	}
\end{equation}
Hence the fibre thus constructed is indeed isomorphic to $\sU(1)$, and transitively acted on by $\sU(1)$.
We denote this transgression $\sU(1)$-bundle by $\CT^{\sU(1)}\CL$.

We would like to pass from the $\sU(1)$-bundle $\CT^{\sU(1)}\CL$ to
its associated line bundle; that is, we would like to transgress the {\em line} bundle gerbe on $M$ to a line bundle on $\CL M$ rather than the $\sU(1)$-bundle gerbe to a $\sU(1)$-bundle.
The line bundle can of course be obtained via the Borel construction from $\CT^{\sU(1)}\CL$, but we would like to relate it explicitly to the 2-line bundle defined by the bundle gerbe on $M$.
Consider a section $(E,\alpha,\nabla^E) \in
\shom_{\CatLBGrb^\nabla(M)}(\CI_0, \CL)$ of the bundle gerbe $\CL$.
For $\gamma \in \CL M$ and a representative $(T,\beta,\nabla^T)$ of the fibre of the $\sU(1)$-bundle on loop space, set
\begin{equation}
\label{eq:transgression_of_section}
 \Big( \CT^{(T,\beta,\nabla^T)} \big( E,\alpha,\nabla^E \big)
 \Big)_{|\gamma} = \Big[ \big[T,\beta,\nabla^T \big],\, \tr \big(
 \hol_{\sfR (T \bullet \gamma^*E)}(S^1) \big) \Big] \ ,
\end{equation}
where $[ T,\beta,\nabla^T ]$ denotes the 2-isomorphism class of the trivialisation.
If we choose another trivialisation $(T',\beta',\nabla^{T'})$, we obtain
\begin{align}
 &\Big( \CT^{(T',\beta',\nabla^{T'})} \big( E,\alpha,\nabla^E \big) \Big)_{|\gamma}\notag \\*
 &= \Big[ \big[ T',\beta',\nabla^{T'} \big],\, \tr \big( \hol_{\sfR (T' \bullet \gamma^*E)}(S^1) \big) \Big]\notag \\[4pt]
 &= \Big[ \big[ (T',\beta',\nabla^{T'}) \bullet (T,\beta,\nabla^T)^{-1} \bullet (T,\beta,\nabla^T) \big] ,\,
	 \hol_{\sfR (T' \bullet T^{-1})}(\gamma)\ \tr \big( \hol_{\sfR (T \bullet \gamma^*E)}(S^1) \big) \Big]\notag \\[4pt]
 &= \Big[ R_{\hol_{\sfR (T' \bullet T^{-1})}(\gamma)^{-1}}\, \big[T,\beta,\nabla^T \big] ,\,
 \hol_{\sfR (T' \bullet T^{-1})}(\gamma)\ \tr \big( \hol_{\sfR (T \bullet \gamma^*E)}(S^1) \big) \Big]\notag \\[4pt]
 &= \Big( \CT^{(T,\beta,\nabla^T)} \big( E,\alpha,\nabla^E \big)
   \Big)_{|\gamma} \ .
\end{align}
Via the identification of the right $\sU(1)$-action%
	\footnote{Here the action of $(T',\beta',\nabla^{T'}) \bullet (T,\beta,\nabla^T)^{-1}$ is from the left, so that it is naturally interpreted as the right action of the inverse, cf.\ also~\cite{Taubes:2011aa}.}
and the left action on $\FC$ as indicated above, the second to last line identifies with the equivalence relation $[R_g p, \varrho(g^{-1})(v)] = [p,v]$ in the generic Borel construction of an associated bundle to a (right) principal bundle.
Therefore the prescription in~\eqref{eq:transgression_of_section} maps a section of the bundle gerbe on $M$ to a section of the line bundle on loop space associated to the $\sU(1)$-bundle constructed by Waldorf.
The inverse in the definition of $\CT$ is necessary in order for the
$\sU(1)$-bundle to be a right principal bundle.

More generally, let $(E,\alpha,\nabla^E) \in
\shom_{\CatLBGrb^\nabla(M)}(\CL_1,\CL_2)$ be an arbitrary morphism
of bundle gerbes, and let $(T_i,\beta_i,\nabla^{T_i}): \gamma^*\CL_i \isom \CI_0$ for $i = 1,2$ be two trivialisations.
Defining
\begin{equation}
\label{eq:transgression_of_morphism}
\CT \big( E,\alpha,\nabla^E \big) \Big( \Big[ \big[ T_1,\beta_1,\nabla^{T_1} \big],\, \lambda \Big] \Big):= \Big[ \big[ T_2,\beta_2,\nabla^{T_2} \big],\, \tr \big( \hol_{\sfR (T_2 \bullet \gamma^*E \bullet T_1^{-1})}(S^1) \big)\, \lambda \Big]
\end{equation}
makes $\CT$ into a map
\begin{equation}
	 \CT: \shom_{\CatLBGrb^\nabla(M)}(\CL_1,\CL_2) \longrightarrow \shom_{\CatHLBdl(\CL M)}(\CT \CL_1, \CT \CL_2)~,
\end{equation}
where we use the notation
\begin{equation}
	\CT \CL = \CT^{\sU(1)} \CL \times_{\sU(1)} \FC
\end{equation}
for the transgression line bundle on $\CL M$.

While the domain of this map is a morphism category in $\CatLBGrb^\nabla(M)$, its range is only a set with some additional structure. If we recall that bundles which are isomorphic as bundles with connections have the same holonomy up to an inner automorphism on the holonomy group, and that 2-isomorphic sections of a trivial bundle gerbe descend to isomorphic vector bundles with connections, we see that the Wilson loops of such bundles are actually equal.
This means that 2-isomorphic 1-morphisms of bundle gerbes transgress to the same homomorphisms of line bundles over loop space.
Therefore $\CT$ can actually be defined on the (1-)category
$h_1\CatLBGrb^\nabla(M)$, which is the category obtained from the
full 2-category $\CatLBGrb^\nabla(M)$ by collapsing 1-morphisms to 2-isomorphism classes of 1-morphisms, and forgetting the 2-morphisms altogether.

Recall from Theorem~\ref{st:Direct_sum_as_fctr} that every morphism
category $\shom_{\CatLBGrb^\nabla(M)}(\CL_1,\CL_2)$ is a closed symmetric strict monoidal category enriched over abelian groups.
This structure of the morphism categories collapses to that of a commutative monoid upon taking the quotient by 2-isomorphisms.
Likewise, the category $h_1\CatLBGrb^\nabla(M)$ inherits a closed symmetric monoidal structure from the tensor product on $\CatLBGrb^\nabla(M)$.

\begin{lemma}
The category $h_1\CatLBGrb^\nabla(M)$ is a closed symmetric $\otimes$-monoidal category enriched over the category of commutative $\oplus$-monoids.
\end{lemma}

The structure of $h_1\CatLBGrb^\nabla(M)$ is in fact the same as that of $\CatHLBdl(\CL M)$, so that the transgression $\CT$ now has a good chance to actually be a functor.
Establishing this will be one of the main objectives of this
section. For this, let $\CatHLBdl^{\nabla}_{\rm fus}(\CL M)$ denote
the category of hermitean fusion line bundles with hermitean
connection, but with morphisms being just morphisms of vector bundles
(that is, not parallel or respecting the fusion product). This category is only a subcategory of the category of line bundles on loop space, namely the category of (diffeological) line bundles endowed with a {\em fusion product}, see e.g.~\cite{Waldorf:0911.3212,Waldorf:2010aa} and Example~\ref{ex:BGrb_from_fus_LBdl}.

\begin{theorem}
\label{st:Transgression_is_functorial}
Transgression, as defined above, is a functor
\begin{equation}
 \CT: h_1\CatLBGrb^\nabla(M) \longrightarrow \CatHLBdl^{\nabla}_{\rm fus}(\CL M)
\end{equation}
of closed and symmetric monoidal additive categories.
Its restriction to the groupoids\footnote{This means that we restrict the morphisms in the category to isomorphisms.} of the source and target categories is an equivalence.
\end{theorem}

The existence of a fusion product makes the line bundle induced over any based loop space $\Omega M$ together with the path fibration into a (diffeological) bundle gerbe over $M$.
Hence fusion line bundles are precisely those line bundles for which
this construction produces bundle gerbes on $M$; that is, the line
bundles on $\CL M$ which {\em regress} to a bundle gerbe on $M$.
It has been shown in~\cite{Waldorf:2010aa} that transgression and
regression form an equivalence of categories between fusion principal
$\sU(1)$-bundles with compatible connection on $\CL M$ and
$\sU(1)$-bundle gerbes (as opposed to {\em line} bundle gerbes) on
$M$; that is, on the two groupoids underlying the categories in Theorem~\ref{st:Transgression_is_functorial}.

The second statement of 
Theorem~\ref{st:Transgression_is_functorial} has thus already been
proven in~\cite{Waldorf:2010aa}. After restriction to its groupoid,
$\CatHLBdl^{\nabla}_{\rm fus}(\CL M)$ becomes equivalent to the
category of principal $\sU(1)$-bundles over $\CL M$, where the
equivalence is induced by the fundamental representation of
$\sU(1)$. This restriction is necessary, as all morphisms of principal
$\sU(1)$-bundles are isomorphisms. Similarly,
$h_1\CatLBGrb^\nabla(M)$ becomes equivalent to the category of
$\sU(1)$-bundle gerbes over $M$ which appears in~\cite{Waldorf:2010aa}.

\subsection{Transgression of algebraic data}
\label{sect:Transgression_of_algebraic_data}

We will first show that $\CT$ respects composition of morphisms in $\CatLBGrb^\nabla(M)$.
For this, consider two composable morphisms $(E,\alpha,\nabla^E) \in \shom_{\CatLBGrb^\nabla(M)}(\CL_1,\CL_2)$ and $(F,\beta,\nabla^F) \in \shom_{\CatLBGrb^\nabla(M)}(\CL_2,\CL_3)$. Let $\gamma \in \CL M$ be a loop in $M$, and let $(T_i,\zeta_i,\nabla^{T_i}) \in$\linebreak $ \sisom_{\CatLBGrb^\nabla(S^1)}(\gamma^*\CL_i,\, \CI_0)$ be trivialisations of the pullback along $\gamma$ of the bundle gerbes $\CL_i$ for $i = 1,2,3$.
From~\eqref{eq:transgression_of_morphism} we then obtain
\begin{align}
 &\CT (F,\beta,\nabla^F) \circ \CT (E,\alpha,\nabla^E)\, \Big[ \big[ T_1, \zeta_1, \nabla^{T_1} \big],\, \lambda \Big] \nt*
 &\qquad\qquad= \CT (F,\beta,\nabla^F) \Big[ \big[ T_2, \zeta_2, \nabla^{T_2} \big],\, \tr\big( \hol_{\sfR (T_2 \bullet \gamma^*E \bullet T_1^{-1})} (S^1) \big)\, \lambda \Big] \nt[4pt]
 &\qquad\qquad= \Big[ \big[ T_3, \zeta_3, \nabla^{T_3} \big],\, \tr\big( \hol_{\sfR (T_3 \bullet \gamma^*F \bullet T_2^{-1})} (S^1) \big)\, \tr\big( \hol_{\sfR (T_2 \bullet \gamma^*E \bullet T_1^{-1})} (S^1) \big)\, \lambda \Big] \nt[4pt]
 &\qquad\qquad= \Big[ \big[ T_3, \zeta_3, \nabla^{T_3} \big],\, \tr \Big( \hol_{\sfR (T_3 \bullet \gamma^*F \bullet T_2^{-1})} (S^1) \otimes \hol_{\sfR (T_2 \bullet \gamma^*E \bullet T_1^{-1})} (S^1) \Big)\, \lambda \Big] \nt[4pt]
 &\qquad\qquad= \Big[ \big[ T_3, \zeta_3, \nabla^{T_3} \big],\, \tr\big( \hol_{\sfR (T_3 \bullet \gamma^*F \bullet T_2^{-1} \bullet T_2 \bullet \gamma^*E \bullet T_1^{-1})} (S^1) \big)\, \lambda \Big] \nt[4pt]
 &\qquad\qquad= \CT \big( (F,\beta,\nabla^F) \bullet (E,\alpha,\nabla^E) \big)\, \Big[ \big[ T_1, \zeta_1, \nabla^{T_1} \big],\, \lambda \Big]~.
\end{align}
Therefore $\CT$ is compatible with compositions of morphisms.
From this compatibility we also deduce $\CT(\id_{\CL}) = \id_{\CT \CL}$.
This establishes the functoriality of $\CT$.

Taking the holonomy is compatible with Whitney sums and tensor products of vector bundles. Moreover, the trace transforms Whitney sums into sums and tensor products into products in $\FC$.
This, together with the compatibility of the reduction functor $\sfR$ with tensor product and direct sum, implies that $\CT$ maps direct sums of morphisms of line bundle gerbes into the sum of morphisms of line bundles,
\begin{equation}
 \CT \big( (E,\alpha,\nabla^E) \oplus (E',\alpha',\nabla^{E'}) \big) = \CT \big( E,\alpha,\nabla^E \big) + \CT \big( E',\alpha',\nabla^{E'} \big)~.
\end{equation}
Thus it respects the monoidal enriched structure on $h_1\CatLBGrb^\nabla(M)$.
Moreover, the tensor product of morphisms of bundle gerbes gets mapped to the tensor product of morphisms of line bundles,
\begin{equation}
 \CT \big( (E,\alpha,\nabla^E) \otimes (E',\alpha',\nabla^{E'}) \big) = \CT \big( E,\alpha,\nabla^E \big) \otimes \CT \big( E',\alpha',\nabla^{E'} \big)~.
\end{equation}
Consequently $\CT$ is compatible with the full monoidal structure on $h_1 \CatLBGrb^\nabla(M)$, as well as with the monoidal structures on the morphism categories.
As both these structures are functorial, they are compatible with taking 2-isomorphism classes of 1-morphisms, so that $\CT$ is compatible with the respective reduced structures on the (1-)category $h_1\CatLBGrb^\nabla(M)$.

Each line bundle $\CT \CL$ has a natural $\sU(1)$-structure induced by the bases of the fibres given by trivialisations of $\gamma^* \CL$.
This endows these line bundles naturally with a hermitean metric given by
\begin{equation}
h \Big( \Big[ \big[ T,\beta,\nabla^T \big],\, \nu \Big] ,\, \Big[ \big[ T,\beta,\nabla^T \big],\, \lambda \Big] \Big) = \overline{\nu} \, \lambda~.
\end{equation}
Let us evaluate this metric on a pair of transgressed sections of $\CL$.
We claim that this metric is actually induced by the line bundle gerbe metric $\frh$ constructed in Definition~\ref{def:2-bdl_metric} according to
\begin{equation}
 h\, \big( \CT (E,\alpha,\nabla^E),\, \CT (E',\alpha',\nabla^{E'}) \big)
 = \CT \circ \frh\, \big( (E,\alpha,\nabla^E),\, (E',\alpha',\nabla^{E'}) \big)~.
\end{equation}
For this, we need the following lemma.

\begin{lemma}
Let $(T,\beta,\nabla^T): \CL_1 \isom \CL_2$ be an isomorphism of bundle gerbes.
There exists a 2-isomorphism
\begin{equation}
 \delta_{\CL_2} \bullet \lsb (T,\beta,\nabla^T),\, (T,\beta,\nabla^T) \rsb \bullet \delta_{\CL_1}^{-1} \twoisom \id_{\CI_0}~.
\end{equation}
\end{lemma}

\begin{proof}
By Theorem~\ref{st:internal-2-hom-adjunction} we know that there is a natural bijection
\begin{equation}
\begin{aligned}
 &2\shom_{\CatLBGrb^\nabla(M)} \big(\id_{\CI_0} \otimes (T,\beta,\nabla^T),\, (T,\beta,\nabla^T) \big)\\ 
 & \qquad\qquad \isom 2\shom_{\CatLBGrb^\nabla(M)} \big(\id_{\CI_0},\, \delta_{\CL_2} \bullet \lsb (T,\beta,\nabla^T),\, (T,\beta,\nabla^T) \rsb \bullet \delta_{\CL_1}^{-1} \big)~.
\end{aligned}
\end{equation}
The naturality property of the bijection which establishes the adjunction of Theorem~\ref{st:internal-2-hom-adjunction} guarantees that isomorphisms get mapped to isomorphisms.
In the source of this map there is of course at least one 2-isomorphism given by the identity on $(T,\beta,\nabla^T)$.
Its image under the natural bijection of the 2-morphism sets then provides the required isomorphism: The natural pairing $T^* \otimes T \isom Y^{[2]} \times \FC$ yields an isomorphism as desired.
\end{proof}

An immediate consequence of this lemma is that $\delta_{\CL_2} \bullet (\Theta T \otimes T) \bullet \delta_{\CL_1}^{-1}$ has trivial holonomy.
Since the isomorphisms $\delta_{\CI_\rho}$ are defined using the trivial line bundle with the trivial connection, they do not contribute to the holonomy.
Hence we have
\begin{align}
 \hol_{\sfR  (\gamma^* \frh (E,E'\,))}(S^1)
  &=\hol_{\sfR  (\gamma^*(\delta_{\CL} \bullet (\Theta E \otimes E'\,) \bullet \delta_{\CI_0}^{-1}))}(S^1) \nt[4pt]
  &= \underbrace{\hol_{\sfR  ( \delta_{\CI_\rho} \bullet (\Theta T \otimes T) \bullet \delta_{\CL}^{-1})}(S^1)}_{= 1}
  \otimes \, \hol_{\sfR  (\gamma^*(\delta_{\CL} \bullet (\Theta E \otimes E'\,) \bullet \delta_{\CI_0}^{-1}))}(S^1) \nt[4pt]
  &= \hol_{\sfR  (\delta_{\CI_\rho} \bullet (\Theta T \otimes T) \bullet \gamma^*((\Theta E \otimes E'\,) \bullet \delta_{\CI_0}^{-1})}(S^1) \nt[4pt]
  &= \hol_{\sfR ( \Theta (T \bullet \gamma^*E))}(S^1) \otimes \hol_{\sfR (T \bullet \gamma^*E'\,)}(S^1) \nt[4pt]
  &= \hol_{\sfR (T \bullet \gamma^*E)}(S^1)^{-{\rm t}} \otimes \hol_{\sfR (T \bullet \gamma^*E'\,)}(S^1)~.
\end{align}
Now taking the trace of the left-hand side of this expression produces $\CT \circ \frh (E,E'\,)$ at the loop $\gamma$.
On the other hand, recalling that all connections are chosen to be hermitean so that $\hol_{\sfR (T \bullet \gamma^*E)}(S^1)^{-{\rm t}} = \overline{\hol_{\sfR (T \bullet \gamma^*E)}(S^1)}$, taking the trace of the right-hand side yields precisely $h (\CT E,\, \CT E'\,)$.
This proves the following result.

\begin{proposition}
The line bundle gerbe metric $\frh$ (and each of its multiples in the sense of Remark~\ref{rem:multiples_of_2-bdl_metric}) on the category of sections of $\CatLBGrb^\nabla(M)$ induces the canonical bundle metric on the transgression line bundles over loop space via
 \begin{equation}
  \CT \frh = h = \CT \circ \frh~.
 \end{equation}
\end{proposition}

\subsection{Transgression of connections}
\label{sect:connection_on_CTCG}

Let us now turn to the differential structure on a bundle gerbe, and try to make sense of the transgression functor on connections as well.
The results of~\cite{Waldorf:2010aa} make this plausible.

We start from a section of a bundle gerbe $\CL$, that is, a morphism $(E,\alpha,\nabla^E): \CI_0 \rightarrow \CL$.\footnote{We could also, for instance, consider local sections of $\CL$ in the sense of sections of $\phi^*\CL$ for a map $\phi: I^n {\times} S^1 \rightarrow M$.}
We would like to take the (covariant) derivative of its transgression to loop space.
To this end, let $X \in T_\gamma \CL M$ be a tangent vector to $\gamma$ in $\CL M = C^\infty(S^1,M)$.
According to our findings in Appendix~\ref{sect:mapping_space_geometry}, this corresponds to a vertical vector field $\bbG_* X \in \Gamma(\bbG \gamma,\, T^{\rm ver}(S^1{\times}M))$ on the embedded submanifold $\bbG \gamma \subset S^1 {\times} M$ given by the graph of $\gamma$.
Since smooth sections of vector bundles on manifolds form soft sheaves, we can extend $\bbG_*X$ to a global vertical vector field $\Xi$ on $S^1{\times}M$.
For $\epsilon>0$ small enough, the flow maps $\Phi^\Xi_\tau$ are all diffeomorphisms on $S^1{\times}M$ for $\tau \in (-\epsilon,\epsilon)$.

Define the map
\begin{equation}
	f^\Xi_\gamma: (-\epsilon, \epsilon) \times S^1 \longrightarrow M~,\quad
	(\tau,\sigma) \longmapsto \pr_M \circ \Phi^\Xi_\tau \big( \bbG \gamma(\sigma) \big)~.
\end{equation}
By construction this map satisfies
\begin{equation}
 \frac{\partial}{\partial \tau}_{|(\tau,\sigma)} f_\gamma^\Xi
  = {\dd f^\Xi_\gamma}_{|(\tau,\sigma)} \Big( \,
  \frac{\partial}{\partial \tau}_{|(\tau,\sigma)} \, \Big)
  = \frac{\partial}{\partial \tau}_{|(\tau,\sigma)} \pr_M \circ \Phi^\Xi_\tau \big( \bbG \gamma (\sigma) \big)
  = \Xi_{|f^\Xi_\gamma(\tau,\sigma)}
\end{equation}
as well as
\begin{equation}
 \frac{\partial}{\partial \sigma}_{|(\tau,\sigma)} f_\gamma^{\Xi}
 = {\dd f^{\Xi}_\gamma}_{|(\tau,\sigma)} \Big( \,
 \frac{\partial}{\partial \sigma}_{|(\tau,\sigma)} \, \Big)
 = \pr_{M*} \circ \Phi^\Xi_{\gamma*_{|f^\Xi_\gamma(\tau,\sigma)}}
 \Big( \, \frac{\partial}{\partial \sigma}_{|\sigma} \gamma \, \Big)~.
\end{equation}

Alternatively, we can view the adjunct map $f^{\Xi \dashv}_\gamma:
(-\epsilon,\epsilon) \rightarrow \CL M,\, \tau \mapsto
f^\Xi_\gamma(\tau,-)$ as the integral curve of the vector field
$\bbG^* \Xi$ on $\CL M$; in particular, this is a path in the loop space of $M$.
The parameter space $(-\epsilon, \epsilon) \times S^1$ is a cylinder with boundary, so that we can always find a trivialisation $(T,\beta,\nabla^T) \in \sisom_{\CatLBGrb^\nabla_2 ((-\epsilon,\epsilon){\times}S^1)} (f^{\Xi*}_\gamma \CL,\, \CI_\rho)$.
This induces a section of the principal $\sU(1)$-bundle on loop space
via the pullbacks along the embeddings $\imath_\tau: S^1
\hookrightarrow (-\epsilon, \epsilon) \times S^1,\, \sigma \mapsto
(\tau,\sigma)$ of $S^1$ into $(-\epsilon, \epsilon ) \times S^1$.
By dimensional reasons $\imath_\tau^* \rho = 0$ for all $\tau \in
(-\epsilon,\epsilon)$, and therefore
\begin{equation}
 \imath_\tau^* \big(T, \beta, \nabla^T \big): (f^{\Xi \dashv}_\gamma(\tau))^* \CL
 \isom \CI_0 \ \in \ \CatLBGrb^\nabla(S^1)
\end{equation}
defines an element in the transgression $\sU(1)$-bundle of $\CL$.
Now let $(E,\alpha,\nabla^E)$ be a section of $\CL$ along the surface $f^\Xi_\gamma$, that is, an element in $\shom_{\CatLBGrb^\nabla((-\epsilon,\epsilon){\times}S^1)} (\CI_0,\, f^{\Xi*}_\gamma \CL)$.
This might for instance be given as the the pullback to the surface of a global section of $\CL$.
Given such a section together with a trivialisation $(T,\beta,\nabla^T)$ of $f^{\Xi*}_\gamma \CL$ over $(-\epsilon,\epsilon){\times}S^1$, we obtain from this a section of $\CT \CL \rightarrow \CL M$ along the path $f^{\Xi\dashv}_\gamma$ in $\CL M$.
This section is defined by
\begin{equation}
	 \CT (E,\alpha,\nabla^E) : (-\epsilon,\epsilon) \longrightarrow (f^{\Xi\dashv}_\gamma)^* \CT \CL~,\quad
	 \tau \longmapsto \Big[ \big[ \imath_\tau^*(T,\beta,\nabla^T) \big],\, \tr \big( \hol_{\imath_\tau^* \sfR (T \bullet E)}(S^1) \big) \Big]~.
\end{equation}

The Wilson loop part of the section is a $\FC$-valued function on the interval.
The covariant derivative of $\CT (E,\alpha,\nabla^E)$ will consist of
the Lie derivative of this function together with an additional term
containing a local representative of the connection 1-form. For this, consider a principal bundle $P \rightarrow M$ with connection $\nabla$.
Let $f: I {\times} S^1 \rightarrow M$ be a smooth family of based
loops in $M$ with $f(\tau,0) = p$ for all $ \tau \in I$.
Then the holonomy of $f^*P$, starting at $p$, around the circle at $\tau \in I$ can be written as $\hol_{p, \jmath_\tau^*f^*P}(S^1)$, where $\jmath_\tau: S^1 \hookrightarrow I {\times} S^1,\, \sigma \mapsto (\tau, \sigma)$ is the inclusion of $S^1$ at parameter $\tau$.
The variation of the holonomy induced by the variation of the loop is
given by (cf.\ e.g.~\cite{Taubes:2011aa})
\begin{equation}
\label{eq:variation_of_holonomy}
	\frac{\dd}{\dd\tau}_{|\tau_0} \hol_{p, \jmath_\tau^*f^*P}(S^1) =
        {L_{\hol_{p, \jmath_{\tau_0}^*f^*P}(S^1)*}} \Big( - \int_{S^1}\,  \jmath_{\tau_0}^* \big( \iota_{\partial_\sigma} F_{{f}^*\nabla} \big) \Big)~,
\end{equation}
where $L_-$ here denotes left action by a group element. Employing this identity, we now consider holonomy in associated vector
bundles and compute the Lie derivative of the local representative of
a section to get
\begin{align}
\label{eq:Lie_derivative_of_transgressed_section}
 &\pounds_{\partial_\tau} \bigg( \tr \Big( \hol_{\jmath_\tau^* \sfR (T \bullet E)}(S^1) \Big) \bigg) (\tau) \nt*
  &\qquad = \tr \Big( - \hol_{\jmath_\tau^* \sfR (T \bullet E)}(S^1)\
    \int_{S^1} \, \CP_{\jmath_\tau^* \sfR (T \bullet E)} \circ \jmath_\tau^* \big( \iota_{\partial_\tau} F_{\nabla^{\sfR (T\bullet E)}} \big) \circ \CP_{\jmath_\tau^* \sfR (T \bullet E)}^{-1} \Big)\nt[4pt]
  &\qquad = - \tr \Big( \hol_{\jmath_\tau^* \sfR (T \bullet E)}(S^1)\
    \int_{S^1}\,  \jmath_\tau^* \big( \CP_{\sEnd(\sfR (T \bullet E))} \circ \iota_{\partial_\tau} \sfR ( F_{\nabla^{T\bullet E}} ) \big) \Big)\nt[4pt]
  &\qquad = - \tr \Big( \hol_{\jmath_\tau^* \sfR (T \bullet E)}(S^1)\
    \int_{S^1} \, \jmath_\tau^* \big( \CP_{\sfR (\sEnd(T)) \otimes \sfR (\sEnd(E))} \circ \iota_{\partial_\tau} \sfR ( (\rho - B)\,\unit + F_{\nabla^E}) \big) \Big)\nt[4pt]
  &\qquad = - \tr \Big( \hol_{\jmath_\tau^* \sfR (T \bullet E)}(S^1)\
    \int_{S^1}\,  \jmath_\tau^* \big( \CP_{\sfR (\sEnd(E))} \circ \iota_{\partial_\tau} \sfR ( F_{\nabla^E} - B\,\unit ) \big) \Big)\nt*
  &\qquad \hspace{4cm}- \tr \Big( \hol_{\jmath_\tau^* \sfR (T \bullet
    E)}(S^1)\ \int_{S^1}\,  \jmath_\tau^* \big( \iota_{\partial_\tau} \rho \big) \Big)~.
\end{align}
Here $\CP$ is the parallel transport along the circle parameterised by itself via the identity.
It starts from a chosen basepoint on the circle at which we wish to evaluate the holonomy but which we suppress here.
Its choice is irrelevant for the trace, since upon changing this basepoint the argument of the trace changes by adjoining an isomorphism between two fibres.
The connection used on the descent bundle $\sfR (T \bullet E)$ is the one inherited through the descent from the connections on $T$ and $E$.
Since $(T,\beta,\nabla^T) \bullet (E,\alpha,\nabla^E)$ descends to a
hermitean vector bundle with connection on $M$, and $F_{\nabla^{T
    \bullet E}}$ is compatible with the descent morphism, this 2-form
induces a 2-form $\sfR (F_{\nabla^{T \bullet E}})$ with values in
$\sEnd(\sfR(T \bullet E))$ on $M$; in fact one has $\sfR (F_{\nabla^{T \bullet E}}) = F_{\nabla^{\sfR(T \bullet E)}}$.

In the expression \eqref{eq:Lie_derivative_of_transgressed_section},
the first term transforms homogeneously upon changing
the local trivialisation $(T,\beta,\nabla^T)$ to $(T',\beta',\nabla^{T'})$, since the term multiplying the holonomy is independent of the choice of trivialisation.
In contrast, the second term changes as
\begin{equation}
\begin{aligned}
 &\tr \Big( \hol_{\jmath_\tau^* \sfR (T \bullet E)}(S^1)\ \int_{S^1} \, \jmath_\tau^* \big( \iota_{\partial_\tau} \rho \big) \Big)\\
 &\begin{aligned}
  \qquad = \hol_{\jmath_\tau^* \sfR (T \bullet T'^*)}(S^1)\ \bigg( &\tr \Big( \hol_{\jmath_\tau^* \sfR (T' \bullet E)}(S^1)\ \int_{S^1} \jmath_\tau^* \big( \iota_{\partial_\tau} \rho'\, \big) \Big)\\
  & \ + \, \tr \Big( \hol_{\jmath_\tau^* \sfR (T' \bullet E)}(S^1)\ \int_{S^1} \,
  \jmath_\tau^* \big( \iota_{\partial_\tau} (\rho - \rho'\, ) \big) \Big) \bigg)
 \end{aligned}\\[4pt]
 &\qquad = \hol_{\jmath_\tau^* \sfR (T \bullet T'^*)}(S^1)\ \tr \Big( \hol_{\jmath_\tau^*
   \sfR (T' \bullet E)}(S^1)\ \int_{S^1} \, \jmath_\tau^* \big( \iota_{\partial_\tau} \rho'\, \big) \Big)\\
  & \qquad \hspace{2cm} + \hol_{\jmath_\tau^* \sfR (T \bullet T'^*)}(S^1)^{-2}\ \pounds_{\partial_\tau} \big( \hol_{\jmath_\tau^* \sfR (T \bullet T'^*)}(S^1) \big)\ \tr \big( \hol_{\jmath_\tau^* \sfR (T' \bullet E)}(S^1) \big)~.
\end{aligned}
\end{equation}
In this expression the first term resembles the homogeneous part in the transformation of a local connection 1-form applied to a local representative of a section in an associated line bundle.
The second term is then precisely the corresponding inhomogeneous part.
From these considerations we obtain a natural candidate for a local representative of a connection 1-form $\CA$ on the transgression line bundle.
With the conventions made in the derivation, it reads as
\begin{equation}
\label{eq:connection_on_loop_space}
[T,\beta,\nabla^T]^* \CA_{|\gamma}(X) = \int_{S^1}\, \jmath_\tau^* \big( \iota_{\partial_\tau} \rho \big)~.
\end{equation}

So far we have only shown by somewhat heuristic means that~\eqref{eq:connection_on_loop_space} is a candidate for a connection on the transgression line bundle.
Let us now proceed to show that the choice~\eqref{eq:connection_on_loop_space} does indeed define a connection on $\CT \CL$.
In the approach thus far, we are still left to prove that the connection is defined independently of the choice of an extension $\Xi$ of $\bbG_* X$.
Instead we shall develop a more geometric approach in what follows.

The group $\sU(1)$ acts on the fibres of the $\sU(1)$-bundle on loop space via the isomorphism $\hol: h_0\CatHLBdl^\nabla_0(S^1) \isom \sU(1)$ to the isomorphism classes of flat line bundles on $S^1$.
A flat line bundle on $S^1$ is classified up to isomorphism by a
connection 1-form, which may be written as a global 1-form $A \in
\di\, \Omega^1(S^1)$.
Two such flat line bundles $A$ and $A'$ are isomorphic if and only if there exists a function $g \in C^\infty(S^1,\sU(1))$ such that $A' = A + \dd\log g$.
An element in the Lie algebra of $h_0\CatHLBdl_0^\nabla(S^1)$ is thus
represented by a 1-form $\eta \in \di\, \Omega^1(S^1)$.
Infinitesimal isomorphisms act as $\eta \mapsto \eta + \dd f$ for $f
\in \di\, C^\infty(S^1,\FR)$.
Thus there is an isomorphism of abelian Lie algebras (cf.\ \eqref{eq:variation_of_holonomy} and \cite{Taubes:2011aa})
\begin{equation}
\hol_*: \sLie(h_0 \CatHLBdl^\nabla_0(S^1)) \cong \di\, H^1_{\rm
  dR}(S^1) \isom \di\, \FR~, \quad
[\eta] \longmapsto - \int_{S^1} \, \eta~.
\end{equation}

Any trivialisation $(T,\beta,\nabla^T)$ of the pullback of $\CL$ to a space of the form $I {\times}S^1$ naturally yields a 2-form $\rho$ that lives on the parameter space.
By regarding $f^\Xi_\gamma$ as a path in $\CL M$ (that is, by passing to the adjunct map $f^{\Xi\dashv}$), $\rho$ becomes a map
\begin{equation}
 \rho^\dashv: TI \longrightarrow \di\, \Omega^1(S^1)~,\quad
  \lambda\, \partial_{\tau|\tau} \longmapsto \lambda\, \jmath_\tau^* (\iota_{\partial_\tau} \rho)~,
\end{equation}
where $\lambda \in \FC$.
After composing this with taking the de Rham equivalence class, we are left with a map
\begin{equation}
 [\rho^\dashv\, ]: TI \longrightarrow \sLie(h_0\CatHLBdl^\nabla_0(S^1))~,\quad
 \zeta_{|\tau} \longmapsto \big[\jmath_\tau^* (\iota_\zeta \rho) \big]~.
\end{equation}
Finally, we can compose with $\hol_* = - \int_{S^1}$ to obtain a map
valued in $\di\, \FR$. It reads as
\begin{equation}
 \hol_* \circ [\rho^\dashv\, ] \ \in \ \di\, \Omega^1(I)~,\quad
  \hol_* \circ [\rho^\dashv\, ] (\zeta_{|\tau}) = - \int_{S^1} \, \jmath_\tau^* \big( \iota_{\zeta} \rho \big)~.
\end{equation}

\begin{definition}
\label{def:horizontality_in_transgression_circle_bdl}
	A section of $\CT^{\sU(1)} \CL$ over a curve $f^\dashv: (-\epsilon,\epsilon) \rightarrow \CL M$ given by the 2-isomorphism class of $(T,\beta,\nabla^T): f^*\CL \isom \CI_\rho$ is \uline{horizontal} if and only if
\begin{equation}
\label{eq:horizontality_in_transgression_circle_bundle}
\begin{aligned}
	\hol_* \circ [\rho^\dashv\, ] = 0
\end{aligned}
\end{equation}
in $\di\, \Omega^1(I)$, or equivalently
$[\jmath_\tau^* (\iota_{\partial_{\tau}}\rho)] = 0$ in
$\di\, H_{\rm dR}^1(S^1)$ for all $\tau \in I$.
\end{definition}

With this definition we can now define a connection in the sense of a geometric distribution on the $\sU(1)$-bundle $\CT^{\sU(1)} \CL$ on loop space.
This then corresponds to a unique connection 1-form, which will turn
out to be precisely the one conjectured
in~\eqref{eq:connection_on_loop_space}.

\begin{theorem}
\label{st:geometry_of_lbdl_on_LM}
\begin{enumerate}
	\item Every path $f^\dashv: (-\epsilon,\epsilon) \rightarrow \CL M$ in the loop space has a unique horizontal lift $[T_h,\beta_h,\nabla^{T_h}]$ through any given point in $(\CT^{\sU(1)}\CL)_{|f^\dashv(0)}$.
	\item There is a $\sU(1)$-invariant splitting $T(\CT^{\sU(1)}\CL) = T^{\rm ver}(\CT^{\sU(1)}\CL) \oplus T^{\rm hor}(\CT^{\sU(1)}\CL)$.
	\item If $(T,\beta,\nabla^{T})$ is another section over $f^\dashv$, we define the bundle $J = \sfR (T \bullet T_h^{-1})$ in $\CatHLBdl^\nabla((-\epsilon,\epsilon){\times}S^1)$ with $F_{\nabla^J} = \rho - \rho_h$.
	It induces a map $\jmath_{(-)}^*J: I \rightarrow \CatHLBdl^\nabla(S^1)$, with $ \tau \mapsto \jmath_\tau^*J$, whose differential satisfies
	\begin{equation}
		\big( \hol \circ ( \jmath_{(-)}^*J ) \big)_{*|\tau_0} = \hol_{\jmath_{\tau_0}^*J}(S^1)\, \Big(- \int_{S^1} \, \jmath_\tau^* \big( \iota_{\partial_\tau}(\rho - \rho_h) \big)_{|\tau_0} \Big) \ .
	\end{equation}
	\item With respect to the horizontal distribution constructed above, define
	 \begin{equation}
		 \CA = \varphi^{-1} \circ \pr_{T^{\rm ver}(\CT^{\sU(1)}\CL)} \ ,
	\end{equation}
where $\varphi$ is the map assigning to an element of $\fru(1)$ its fundamental vector field on the total space of $\CT^{\sU(1)}\CL$. Then $\ker(\CA)=T^{\rm hor}(\CT^{\sU(1)}\CL)$.
	 In particular, $\CA$ is a well-defined connection 1-form on $\CT^{\sU(1)}\CL$, and its pullback along a section given by the 2-isomorphism class of a trivialisation $(T,\beta,\nabla^T)$ reads as
	 \begin{equation}
		 \big[ T,\beta,\nabla^T \big]^* \CA = \int_{S^1} \, \jmath_{(-)}^* \big( \iota_{(-)} \rho \big)~.
	 \end{equation}
\end{enumerate}
\end{theorem}

\begin{proof}
(1):
Let $f^\dashv: (-\epsilon,\epsilon) \rightarrow \CL M$ be a path in the loop space $\CL M$ of $M$.
As before, we denote its adjunct map by $f: (-\epsilon,\epsilon){\times}S^1 \rightarrow M,\, (\tau,\sigma) \mapsto f^\dashv(\tau)(\sigma)$.
Let $(S,\psi,\nabla^S): (f^\dashv(0))^*\CL \isom \CI_0$ be a trivialisation, that is, a representative of an element of $(\CT^{\sU(1)} \CL)_{|f^\dashv(0)}$.

First of all, let us show that there exists a horizontal section of $f^{\dashv *}(\CT^{\sU(1)} \CL)$.
Applying the exact sequence \eqref{eq:DD_exact_sequence} to the cylinder $C = (-\epsilon,\epsilon){\times}S^1$ gives $H^2(C,\CD^\bullet(2)) = 0$ due to  $\Omega^2_{{\rm cl},\RZ}(C) = \Omega^2(C)$, since every 2-form on $C$ is closed and there are no non-trivial 2-cycles in $C$.
Thus {\em all} bundle gerbes (with connection) on $C$ are isomorphic.
In particular, every gerbe is isomorphic to $\CI_0$ on $C$, and a corresponding isomorphism $(T_h,\beta_h,\nabla^{T_h})$ from the pullback bundle gerbe to $\CI_0$ provides a horizontal lift of $f^\dashv$ in the sense of Definition~\ref{def:horizontality_in_transgression_circle_bdl}.
To make this pass through $(S,\psi,\nabla^S)$, we tensor by the constant flat line bundle $\sfR (S \bullet \jmath_0^*T_h^{-1})$.

Next we have to prove that this trivialisation of $f^*\CL$ is unique up to 2-isomorphism.
For this, let us assume that $(T,\beta,\nabla^T): f^*\CL \isom \CI_\rho$ is another trivialisation of the pullback bundle gerbe satisfying the horizontality condition~\eqref{eq:horizontality_in_transgression_circle_bundle}, and assume that it is a horizontal lift through $(S,\psi,\nabla^S)$, that is, $\jmath_0^*(T,\beta,\nabla^T) \cong (S,\psi,\nabla^S)$.
The two sections induced by these trivialisations over the cylinder $C$ are $[\jmath_{(-)}^*(T,\beta,\nabla^T)]$ and $[\jmath_{(-)}^*(T_h,\beta_h,\nabla^{T_h})] $ in $ \Gamma(I, f^{\dashv*} \CT^{\sU(1)} \CL)$; for example, we have
\begin{equation}
\begin{aligned}
	\big[\jmath_{(-)}^*(T,\beta,\nabla^T)\big]: I &\longrightarrow f^{\dashv*} \CT^{\sU(1)} \CL \ , \\
	\tau &\longmapsto \big[\jmath_{\tau}^*(T,\beta,\nabla^T)\big]
        \ \in \ h_0 \sisom_{\CatLBGrb^\nabla(S^1)}\big( {f^{\dashv}(\tau)}^* \CT^{\sU(1)} \CL,\, \CI_0 \big)~.
\end{aligned}
\end{equation}
They are related via the right action of a transition function given as
\begin{align}
	\big[\jmath_{(-)}^*(T_h,\beta_h,\nabla^{T_h})\big] &= \big[\jmath_{(-)}^*(T_h,\beta_h,\nabla^{T_h}) \circ \jmath_{(-)}^*(T,\beta,\nabla^T)^{-1} \circ \jmath_{(-)}^*(T,\beta,\nabla^T)\big] \nt[4pt]
	&= R_{\hol_{\jmath_{(-)}^* \sfR (T \bullet T_h^{-1})}(S^1)^{-1}} \ \big[\jmath_{(-)}^*(T,\beta,\nabla^T)\big]~.
\end{align}
The holonomy in this transition function is just the holonomy of the hermitean line bundle $\sfR (T \bullet T_h^{-1})$ on the cylinder $C$ along the circles at the respective parameter values.
From the variation~\eqref{eq:variation_of_holonomy} of the holonomy upon changing the curve it is evaluated on%
	\footnote{Since we are in an abelian case here, the formula~\eqref{eq:variation_of_holonomy} extends to variations of the curve with varying initial point.
		Another way to see this would be to use an isomorphism $\jmath_{(-)}^*\sfR (T \bullet T_h^{-1}) \cong (C{\times}\mathbb{C}, \dd + \eta)$, which exists since $C$ retracts to $S^1$.
		Then the variational formula easily follows from $\hol_{C\times\FC}(S^1) = \exp(-\int_{S^1}\, \eta)$ and Stokes' theorem.}
we derive
\begin{align}
	&\frac{\dd}{\dd\tau}_{|\tau_0} \big( \hol_{\jmath_{\tau}^* \sfR (T \bullet T_h^{-1}) }(S^1) \big)^{-1} \nt*
	&\qquad\qquad= - \hol_{\jmath_{\tau_0}^* \sfR (T \bullet T_h^{-1})}(S^1)^{-2}\  \hol_{\jmath_{\tau_0}^* \sfR (T \bullet T_h^{-1}) }(S^1) \, \int_{S^1}\, - \jmath_{\tau}^* \big( \iota_{\partial_\tau} F_{\nabla^{\sfR (T \bullet T_h^{-1})}} \big)_{|\tau_0} \nt[4pt]
	&\qquad\qquad = \hol_{\jmath_{\tau_0}^* \sfR (T \bullet T_h^{-1})}(S^1)^{-1}\, \int_{S^1}\, \jmath_{\tau}^* \big( \iota_{\partial_\tau} (\rho - \rho_h) \big)_{|\tau_0}~.
\label{eq:Rconstrho}\end{align}
By the horizontality assumption, $[\jmath_\tau^*( \iota_{\partial_\tau}\rho)]= [\jmath_\tau^*( \iota_{\partial_\tau}\rho_h)] = 0$ for all $\tau \in (-\epsilon,\epsilon)$, and therefore the transition function is constant.
Moreover, the initial condition $\jmath_0^*(T_h,\beta_h,\nabla^{T_h}) \cong (S,\psi,\nabla^S) \cong \jmath_0^*(T,\beta,\nabla^T)$ is equivalent to the transition function having the value $1 \in \sU(1)$ at $\tau = 0$.
This shows that these two trivialisations define the same section of $f^{\dashv*} (\CT^{\sU(1)} \CL)$.

\noindent
(2):
The existence of the splitting is immediate from item (1).
Its $\sU(1)$-invariance can be seen as follows:
If $(T,\beta,\nabla^T): f^*\CL \isom \CI_\rho$ is a horizontal section, acting on it with a constant element of $\sU(1)$ amounts to tensoring $(T,\beta,\nabla^T)$ by the pullback of a line bundle from $S^1$ to $C$.
This, however, has no effect on $\rho$ since such line bundles are flat on $C$.

\noindent
(3):
This computation is the same as the one carried out in \eqref{eq:Rconstrho}.

\noindent
(4):
The unique connection 1-form equivalent to the geometric connection $T^{\rm hor}(\CT^{\sU(1)}\CL)$ acts on $\zeta \in \Gamma(I,TI)$ as
\begin{align}
	[T,\beta,\nabla^T]^* \CA (\zeta)
	&= \big[ (T,\beta,\nabla^T) \bullet (T_h,\beta_h,\nabla^{T_h})^{-1} \bullet (T_h,\beta_h,\nabla^{T_h}) \big]^* \CA (\zeta) \nt[4pt]
	&= \big( R_{\hol_{\jmath_{(-)}^* J}(S^1 )^{-1}} \ [T_h,\beta_h,\nabla^{T_h}] \big)^* \CA (\zeta) \nt[4pt]
	&= \hol_{\jmath_{(-)}^* J }(S^1) \ \dd \big( \hol_{\jmath_{(-)}^* J}(S^1)^{-1} \big) (\zeta) \nt[4pt]
	&= - \hol_{\jmath_{(-)}^* J}(S^1)^{-1} \ \hol_{\jmath_{(-)}^* J }(S^1) \, \int_{S^1}\, - \jmath_{(-)}^* \big( \iota_{\zeta} F_{\nabla^J} \big) \nt[4pt]
	&= \int_{S^1}\, \jmath_{(-)}^* \big( \iota_{\zeta} \rho \big)~,
\end{align}
as required.
\end{proof}

Let us next compute the curvature of this connection.
For this, we extend the construction used to derive the connection 1-form to two tangent vectors $X_0,\, X_1 \in T_\gamma \CL M$.
Let $\Xi_0,\, \Xi_1 \in \Gamma(S^1{\times}M, T^{\rm ver}(S^1{\times}M))$ be two extensions of the graphs $\bbG_* X_0$ and $\bbG_* X_1$ to vertical vector fields on $S^1{\times}M$.
Using the flows of these vector fields, we find $\epsilon_0 > 0$ and $\epsilon_1 > 0$ such that we can define a map
\begin{equation}
 f^{\Xi_1,\Xi_0}_\gamma: (-\epsilon_1,\epsilon_1) {\times}
 (-\epsilon_0,\epsilon_0) {\times} S^1 \longrightarrow M \ , \quad
 f^{\Xi_1,\Xi_0}_\gamma (\tau_1,\tau_0,\sigma) = \Phi^{\Xi_1}_{\tau_1}
 \circ \Phi^{\Xi_0}_{\tau_0} \circ \gamma(\sigma) \ .
\end{equation}
Its derivative along the additional parameters reads as
\begin{equation}
 (f^{\Xi_1,\Xi_0}_\gamma)_{*|(0,0,\sigma)}
 \Big(\,\frac\partial{\partial\tau_i}\, \Big)
 = \frac{\partial}{\partial \tau_i}_{|0} \Phi^{\Xi_i}_{\tau_i} \circ \gamma(\sigma)
 = \Xi_{i|(\sigma,\gamma(\sigma))}
 = \bbG_* X_{i|(\sigma,\gamma(\sigma))} \ ,
\end{equation}
for $i = 0,1$. The parameter space of this map is again homotopy
equivalent to $S^1$, so that we can always find trivialisations of $f^{\Xi_1,\Xi_0*}_\gamma \CL$.
Then since the exterior derivative commutes with pullbacks, application of Proposition~\ref{st:d_and_transgression} from Appendix~\ref{sect:mapping_space_geometry} to the restricted case of integration along the circle factor in $I^2{\times}S^1$ shows that 
\begin{equation}
\begin{aligned}
	\big( [T,\beta,\nabla^T]^* \dd \CA \big)_{|\tau} &= \big( \dd [T,\beta,\nabla^T]^* \CA \big)_{|\tau}\\[4pt]
	&= \dd \Big(\, \int_{S^1}\, \jmath_\tau^* \big(\iota_{(-)}
        \rho \big) \, \Big)\\[4pt]
	&= \int_{S^1}\, \jmath_\tau^* \big(\iota_{(-) \wedge (-)} \dd \rho\big)\\[4pt]
	&= \int_{S^1}\, \jmath_\tau^* \big( \iota_{(-) \wedge (-)} f^{\Xi_1,\Xi_0*}_\gamma H\big)\\[4pt]
&= f^{\Xi_1,\Xi_0 \dashv*}_\gamma \CT(H)_{|\tau} \ ,
\end{aligned}
\end{equation}
where $H\in \di\, \Omega^3(M)$ is the curvature 3-form of $\CL$ and $f^{\Xi_1,\Xi_0\dashv}_\gamma$ is the map from the product of the two parameter intervals to the loop space which is adjunct to $f^{\Xi_1,\Xi_0}_\gamma$.
This shows that
\begin{equation}
\label{eq:curvature_of_lbdl_on_LM}
 F_{\nabla^{\CT^{\sU(1)} \CL}} = \CT(H) \ .
\end{equation}

\begin{remark}
The sign in~\eqref{eq:curvature_of_lbdl_on_LM} is opposite to that
in~\cite{Waldorf:2010aa}, wherein the opposite sign is used in the flatness condition on trivialisations (Definition~\ref{def:flat_isomps_and_trivialisations}) of bundle gerbes.
\end{remark}

We can finally complete the proof of Theorem~\ref{st:Transgression_is_functorial}.

\noindent{\it Proof of Theorem~\ref{st:Transgression_is_functorial}}. By construction, the connection 1-form $\CA$ has holonomy transport along a loop $\phi^\dashv: S^1 \rightarrow \CL M$ given, with respect to a trivialisation $(T,\beta,\nabla^T)$ of $\phi^*\CL$ over $S^1{\times}S^1$, by
\begin{equation}
	\hol_{\CT^{\sU(1)} \CL}\big(\phi^\dashv(S^1)\big) = \exp \big( - [T,\beta,\nabla^T]^* \CA \big)
	= \exp \Big( - \int_{S^1{\times}S^1}\, \rho \Big)~.
\end{equation}
This is precisely the holonomy of the bundle gerbe $\CL$ around the torus $\phi: S^1{\times}S^1 \rightarrow M$ (cf.\ Definition~\ref{def:Hol_of_BGrb}).
Waldorf's arguments for the fusion properties of the connection still
hold true, so that we really obtain a functor with target category
$\CatHLBdl^{\nabla}_{\rm fus}(\CL M)$. \qed

\begin{remark}
\label{rem:fake-curvature_condition}
As explained in Remark~\ref{rem:twist_vec_bundles}, one can regard
morphisms of line bundle gerbes as twisted vector bundles which are
associated 2-vector bundles to certain principal 2-bundles. Thus one
might be tempted to demand that the {\em fake curvature}
vanishes, because this condition renders the 2-holonomy of the
2-bundle well-defined~\cite{Baez:0511710,Baez:2004in}. It reads as 
\begin{equation}\label{eq:fake_curv}
 F_{\nabla^E} - \big( \fpmap{}{\zeta}{Z*}{2} B_2 - \fpmap{}{\zeta}{Z*}{1} B_1 \big)\,\unit = 0~.
\end{equation}
In many contexts, however, this condition seems to be too strict.

Putting together \eqref{eq:Lie_derivative_of_transgressed_section} and~\eqref{eq:connection_on_loop_space}, we compute the covariant derivative of a transgressed section $\CT (E,\alpha,\nabla^E)$ of $\CL$ to be
\begin{align}
 &\nabla^{\CT \CL}_X\, \CT (E,\alpha,\nabla^E)\\*
 &\qquad = \Big[ \big[ T,\beta,\nabla^T \big],\,
  - \tr \Big( \hol_{\sfR (T \bullet E)}\big(\imath_\tau(S^1)\big)\ \int_{S^1}\, \CP_{\sfR (\sEnd(E))} \circ \iota_{\partial_\tau} \sfR \big( F_{\nabla^E} - B\, \unit \big) \Big) \Big]~, \notag
\end{align}
where the trivialisation is chosen in means adapted to the vector field $X$ as before.
This shows in another way how strong a condition the vanishing of the fake curvature is:
Sections $(E,\alpha,\nabla^E) \in \Gamma(M,\CL)$ satisfying this condition
transgress to parallel sections of $\CT \CL$; that is,
\begin{equation}
 \CT: \Gamma_{\rm fc}(M,\CL) = \shom_{\CatLBGrb^\nabla(M),{\rm fc}}(\CI_0,\CL) \longrightarrow \Gamma_{\rm par}(\CL M,\CT \CL)~,
\end{equation}
with $\shom_{\CatLBGrb^\nabla(M),{\rm fc}}$ denoting the full
subcategories of the morphism categories in\linebreak
$\CatLBGrb^\nabla(M)$ whose objects are those 1-morphisms in
$\CatLBGrb^\nabla(M)$ which satisfy the fake curvature condition \eqref{eq:fake_curv}.
\end{remark}

\begin{corollary}
	Flat isomorphisms of bundle gerbes transgress to isomorphisms
        of hermitean line bundles with connection, that is, isomorphisms in $\CatHLBdl^\nabla(\CL M)$.
\end{corollary}

\subsection{Grothendieck completion}
\label{sect:transgression_and_diff_K-theory}

In Theorem~\ref{st:Direct_sum_as_fctr} we saw that the direct sum of bundle gerbe 1-morphisms is a functor, making the categories $\shom_{\CatLBGrb^\nabla(M)}(\CL_1,\CL_2)$ into symmetric cartesian monoidal and semisimple abelian categories.
The transgression functor $\CT$ from Section~\ref{sect:Transgression_definition}, however, is defined on the 1-category $h_1 \CatLBGrb^\nabla(M)$ in which the symmetric monoidal structure of the morphism categories in $\CatLBGrb^\nabla(M)$ is collapsed to 2-isomorphism classes of 1-morphisms, making $\shom_{h_1 \CatLBGrb^\nabla(M)}(\CL_1,\CL_2)$ into a commutative monoid.
In Section~\ref{sect:Transgression_of_algebraic_data} we showed that $\CT$ respects that structure of the morphism sets in $h_1 \CatLBGrb^\nabla(M)$. On the other hand, the morphism sets in $\CatHLBdl^{\nabla}_{\rm fus}(\CL M)$ have the structure of an abelian group, and we can ask whether there is an extension of the transgression functor which acts between two categories with the same structural properties. This is made possible by the following construction.

\begin{definition}\label{def:Grothendieck_completion}
The \uline{Grothendieck 2-functor} is the strict 2-functor of 2-categories
\begin{equation}
 \Gr: \CatCat^{\mathsf{CMon}} \longrightarrow \CatCat^{\mathsf{AbGrp}}
\end{equation}
from the 2-category of categories enriched over commutative monoids to the 2-category of categories enriched over abelian groups (that is, preadditive categories), which is obtained as follows:
\begin{itemize}
\item[(i)] For $\mathscr{A}\in\CatCat^{\mathsf{CMon}}$, the \uline{Grothendieck completion} $\Gr(\mathscr{A})$ is the $\mathsf{AbGrp}$-enriched category with the same objects as $\mathscr{A}$ but its morphism sets
\begin{equation}
 \shom_{\Gr (\mathscr{A})}(a,b) = \Gr \big( \shom_{\mathscr{A}}(a,b) \big)
\end{equation}
are the Grothendieck groups of the commutative monoids $\shom_{\mathscr{A}}(a,b)$.
Composition of morphisms in $\Gr (\mathscr{A})$ is defined via
\begin{equation}
 [ f,g] \circ [h,k] = [f\circ h + g \circ k,\, f \circ k + g \circ h]~,
\end{equation}
for $f,g\in\shom_{\Gr (\mathscr{A})}(b,c)$ and $h,k\in \shom_{\Gr (\mathscr{A})}(a,b)$, and the identity on $a$ is given in terms of the original identity as $[\id_a,0]$.
\item[(ii)] On 1-morphisms in $\CatCat^{\mathsf{CMon}}$, that is, $\mathsf{CMon}$-enriched functors $\Phi: \mathscr{A} \rightarrow \mathscr{B}$, the functor $\Gr(\Phi)$ acts on objects in the same way as $\Phi$ but on morphisms as
\begin{equation}
 \Gr(\Phi)\big([f,g]\big) = \big[\Phi(f), \Phi(g)\big]~.
\end{equation}
\item[(iii)] The 2-morphisms in $\CatCat^{\mathsf{CMon}}$, that is, $\mathsf{CMon}$-enriched natural transformations $\eta: \Phi \Rightarrow \Psi$, are mapped to $\Gr(\eta)=[\eta,0]$.
\end{itemize}
\end{definition}

With this definition, one can check that $\Gr$ is indeed a well-defined strict 2-functor by checking all necessary compatibility conditions, which follow straightforwardly from those on the original enriched categories. For each $\mathscr{A}\in\CatCat^{\mathsf{CMon}}$ there is a canonical inclusion of categories $\imath_{\mathscr{A}}: \mathscr{A} \hookrightarrow \Gr (\mathscr{A})$ which sends an object to itself and a morphism $f\in\shom_{\mathscr{A}}(a,b)$ to its canonical image $[f,0]$ in $\shom_{\Gr(\mathscr{A})}(a,b)$.
If $\mathscr{A}$ is a monoidal category, then its monoidal structure carries over to $\Gr( \mathscr{A})$.

The usual Grothendieck functor from commutative monoids to abelian groups is adjoint to the forgetful functor in the opposite direction.
Such a forgetful 2-functor $\mathsf{U}:\CatCat^{\mathsf{AbGrp}}\to \CatCat^{\mathsf{CMon}}$ exists here as well.
\begin{proposition}\label{prop:Gruniversal}
Let $\mathscr{A} \in \CatCat^{\mathsf{CMon}}$.
For any category $\mathscr{B} \in \CatCat^{\mathsf{AbGrp}}$ and any functor $\Phi: \mathscr{A} \rightarrow \mathsf{U} (\mathscr{B})$ in $\CatCat^{\mathsf{CMon}}$ there exists a unique functor $\hat{\Phi}: \Gr(\mathscr{A}) \rightarrow \mathscr{B}$ in $\CatCat^{\mathsf{AbGrp}}$ such that
\begin{equation}
 \myxymatrix{
  \mathscr{A} \ar@{->}[r]^-{\imath_{\mathscr{A}}} \ar@{->}[d]_-{\Phi} & \Gr(\mathscr{A}) \ar@{-->}[dl]^-{\mathsf{U} (\hat{\Phi})}\\
  \mathsf{U} (\mathscr{B}) & 
 }
\end{equation}
is a strictly commutative diagram in $\CatCat^{\mathsf{CMon}}$; that is, there is an adjoint pair of 2-functors
\begin{equation}
	\myxymatrix{
		\CatCat^{\mathsf{CMon}} \ar@<0.2cm>[rr]^-{\Gr}	&	\perp	&	\CatCat^{\mathsf{AbGrp}} \ar@<0.2cm>[ll]^-{\mathsf{U}}
	}
\end{equation}
\end{proposition}

\begin{proof}
On objects, $\hat{\Phi}$ has to coincide with $\Phi$ in order to obtain a strictly commutative diagram.
Fixing two objects in $\mathscr{A}$, the assertion follows from the corresponding universal property of the ordinary Grothendieck functor applied to each monoid of morphisms.
\end{proof}

Let us now specialise to the $\mathsf{CMon}$-enriched categories $h_1 \CatLBGrb^\nabla(M)$ and $\CatHLBdl_{\rm fus}^{\nabla}(\CL M)$ of interest to us.
Applying the Grothendieck 2-functor $\Gr$ has a non-trivial effect only on $h_1 \CatLBGrb^\nabla(M)$, since $\CatHLBdl_{\rm fus}^{\nabla}(\CL M)$ is already a category in $\CatCat^{\mathsf{AbGrp}}$. Strictly speaking, Theorem~\ref{st:Transgression_is_functorial} asserts that transgression $\CT$ is a functor
\begin{equation}
 \CT: h_1 \CatLBGrb^\nabla(M) \longrightarrow \mathsf{U} \big( \CatHLBdl_{\rm fus}^{\nabla}(\CL M) \big)
\end{equation}
between monoidal categories enriched over commutative monoids.
Hence using Proposition~\ref{prop:Gruniversal} we can immediately infer the following result.

\begin{theorem}\label{thm:7.14}
	The transgression functor, as defined in Section~\ref{sect:Transgression_definition}, has a unique extension to a functor
	\begin{equation}
		\hat{\CT}: \Gr \big( h_1 \CatLBGrb^\nabla(M) \big) \longrightarrow \CatHLBdl_{\rm fus}^{\nabla}(\CL M)
	\end{equation}
	of closed and symmetric monoidal additive categories enriched over abelian groups.
\end{theorem}

\begin{remark}
	The morphism sets in the Grothendieck completion $\Gr (h_1 \CatLBGrb^\nabla(M))$ of the category of bundle gerbes on $M$ are certain subgroups of the twisted differential K-theory groups $\hat{K}^0(M,\CL)$.
	In particular, $\Gr(\shom_{\CatLBGrb^\nabla(M)}(\CI_0,\CL)) = \Gr(\Gamma(M,\CL))$ is the subgroup of $\hat{K}^0(M,\CL)$ with vanishing auxiliary form $\rho=0$, defined by $\CL$-twisted hermitean vector bundles with connection on $M$, cf.~\cite{Bunke:1011.6663,Park:1602.02292}.
\end{remark}

\subsection{Kostant-Souriau prequantisation}
\label{sect:KS_prequantisation}

Let us finally relate transgression to higher geometric
prequantisation. Consider the space of smooth (diffeological) sections
of the transgression line bundle $\CT \CL$ on loop space. Although the
definition of a Hilbert space structure on these sections is subtle,
we can nevertheless take their covariant derivatives using the
transgressed connection on $\CT \CL$ from
Section~\ref{sect:connection_on_CTCG} and make some statements.

Let us start from a quantisable 2-plectic manifold $(M,\omega)$ and
let $\CL$ be a prequantum line bundle gerbe. The transgression
$\CT(\omega)$ of $\omega$ to loop space $\CL M$ yields a closed 2-form,
which has as its kernel the vector fields generating
reparameterisations of the loops, cf.\
\cite{Wurzbacher:1995gb,0817647309,Saemann:2012ab}. It is therefore
called \emph{weakly symplectic}, and it would turn into a proper
symplectic form upon restriction to unparameterised loop space or knot
space $\CK M$. Here we differ from the usual discussions of
quantisation of loop space as done e.g.\
in~\cite{Pressley:1988aa,Popov:1992am}: Instead of using the weakly
symplectic structure on loop space obtained from the canonical 1-form
on $\CL M$, we work with a weakly symplectic structure induced by
transgression from a curvature 3-form on $M$.

We take the Hamiltonian 1-forms introduced in
Section~\ref{ssec:multisymplectic}, form the corresponding Hamiltonian
vector fields, and pull them back to $\CL M$. From the results in
Appendix~\ref{sect:mapping_space_geometry}, we observe that for
$\alpha,\beta\in\Omega^1_{\rm Ham}(M)$ one has
\begin{equation}
 \iota_{X_\alpha}\omega=-\dd \alpha
\end{equation}
if and only if
\begin{equation}
\iota_{X_\alpha}\CT(\omega)=-\dd \, \CT (\alpha) ~,
\end{equation}
where we denoted the pullback of a vector field $X$ on $M$ to $\CL M$ by the same symbol. This implies that
\begin{equation}
 \CT(\{\alpha,\beta\}_\omega)
 :=\CT(-\iota_{X_\alpha}\iota_{X_\beta}\omega)=-\iota_{X_\alpha}\iota_{X_\beta}\CT(\omega)
 =\{\CT(\alpha) ,\CT(\beta) \}_{\CT(\omega)}~,
\end{equation}
that is, transgression is a Lie 2-algebra homomorphism from
$\Pi_{M,\omega}$ to the strict discrete Lie 2-algebra $C^\infty(\CK
M)\rightrightarrows C^\infty(\CK M)$ with Poisson bracket induced by $\CT(\omega)$.

We can now take the covariant
derivative of a section along the pullback to $\CL M$ of a Hamiltonian
vector field and
arrive at a higher version of statement (8) from Section~\ref{ssec:gq_outline}.
\begin{proposition}\label{prop:KS_prequantisation_map}
Let $X_\alpha$ be the Hamiltonian vector field of $\alpha\in\Omega^1_{\rm Ham}(M)$. Then the map
\begin{equation}
 \CQ:\ \alpha \longmapsto \nabla^{\CT \CL}_{X_\alpha} + 2\pi \, \di\, \CT(\alpha)
\end{equation}
defines a representation\footnote{By a {\em representation of a Lie 2-algebra} on a vector space we simply mean a Lie 2-algebra homomorphism into the endomorphism Lie algebra of this vector space.} of the Lie 2-algebra $\Pi_{M,\omega} =
\big(\Omega^1_{\rm Ham}(M)\ltimes C^\infty(M)\rightrightarrows
\Omega^1_{\rm Ham}(M) \big)$ from Section~\ref{ssec:multisymplectic} on $\Gamma(\CL M,\CT \CL)$, which factors through $h_0\Pi_{M,\omega}$.
\end{proposition}

\begin{proof}
For $\alpha,\beta \in\Omega^1_{\rm Ham}(M)$ and $\psi\in \Gamma(\CL M,\CT \CL)$ we compute
\begin{align}
	\big[ \CQ(\alpha),\, \CQ(\beta) \big] \psi
	&= \big( \nabla^{\CT \CL}_{X_\alpha} + 2\pi \, \di\, \CT(\alpha) \big) \circ \big( \nabla^{\CT \CL}_{X_\beta} + 2\pi\, \di\, \CT(\beta) \big) \psi - (\alpha \leftrightarrow \beta) \nt[4pt]
	&= \big[ \nabla^{\CT \CL}_{X_\alpha},\, \nabla^{\CT \CL}_{X_\beta} \big] \psi
	+ 2\pi\, \di\, \big( \nabla^{\CT \CL}_{X_\alpha}\, \CT(\beta) - \nabla^{\CT \CL}_{X_\beta}\, \CT(\alpha) \big) \psi \nt[4pt]
	&= \iota_{X_\alpha\wedge X_\beta}F_{\nabla^{\CT \CL}}\, \psi + \nabla^{\CT \CL}_{[X_\alpha,X_\beta]} \psi
	+ 2\pi\, \di\, \big( \CT (\iota_{X_\alpha}\, \dd\beta) - \CT (
          \iota_{X_\beta}\, \dd\alpha) \big) \psi \nt[4pt]
	&= \iota_{X_\alpha\wedge X_\beta}\, \CT(H)\, \psi
	+ \nabla^{\CT \CL}_{X_{\{\alpha,\beta\}_\omega}}\, \psi
	+ 4\pi\, \di\, \CT \big( \{ \alpha,\beta \}_\omega \big) \psi~,
\end{align}
where we used the fact that transgression commutes with the exterior
derivative (see Appendix~\ref{sect:mapping_space_geometry}). Now we
use $\iota_{X_\alpha\wedge X_\beta}\, \CT(H) = \int_{S^1}\, \ev^* \iota_{X_\alpha \wedge X_\beta} H = \int_{S^1}\,\ev^*( - 2\pi\, \di\, \{ \alpha,\beta \}_\omega) = - 2\pi \, \di\, \CT( \{ \alpha,\beta \}_\omega)$ to finally get
\begin{align}
	\big[ \CQ(\alpha),\, \CQ(\beta) \big] \psi = \nabla^{\CT \CL}_{X_{\{\alpha,\beta\}_\omega}}\, \psi
	+ 2\pi\, \di\, \CT \big( \{ \alpha,\beta \}_\omega \big) \psi = \CQ \big( \{\alpha,\beta \}_\omega \big)\,\psi~,
\end{align}
which shows that $\CQ$ is indeed a representation of $h_0\Pi_{M,\omega}$
on $\Gamma(\CL M,\CT \CL)$.
\end{proof}

As the Jacobiator in $\Pi_{M,\omega}$ is an exact 1-form on $M$, cf.\ \eqref{eq:jacobiator_exact}, it vanishes under the transgression map.
Therefore the nonassociativity in the Lie 2-algebra $\Pi_{M,\omega}$ is invisible from the loop space perspective.
This is consistent with the heuristic statement that a higher quantum theory on $M$ should correspond to an ordinary quantum theory on $\CL M$.

It would be desirable to define a higher version of the Kostant-Souriau prequantisation map directly on $M$ instead of on $\CL M$. A direct generalisation of the quantisation map $Q$, however, relies on the notion of a covariant derivative of a section of a bundle gerbe. At this point, there is no such construction evident.
One of the core obstructions to simply generalising the usual differential geometric formalism of a covariant derivative is the fact that the set of sections of $\CL$ is not a vector space.
Even if we put a topology and smooth structure on its collection of objects, so that we can define smooth curves of objects, the derivative of such a curve would not be an object of the 2-vector space.
A hint on how one might want to proceed is contained in Section~\ref{ssec:observables}, where an embedding of the Lie 2-algebra of observables into the endomorphisms on the 2-Hilbert space $\Gamma(M,\CL)$ was given. We leave further exploration of this point to future work.

\subsection{Dimensional reduction}

We conclude this section with a few comments concerning dimensional reductions. As stated in Section~\ref{sec:motivation}, one of our motivations for studying prequantum line bundle gerbes stems from the conjecture that quantised 2-plectic manifolds arise naturally in M-theory. The reductions of the underlying M-theory configurations to string theory involves in many cases a reduction of the 2-plectic manifold to a symplectic manifold. It is therefore important that our constructions similarly exhibit ``nice'' reductions to those of ordinary geometric quantisation.

For simplicity, let us restrict to the case $M=\FR^3$ of Section~\ref{ssec:R3_2_Hilbert_space}. We choose the $x^3$-direction as the M-theory direction and find that the 2-plectic form $\omega=\dd x^1\wedge \dd x^2\wedge \dd x^3$ reduces as expected to a symplectic form, $\omega_{\rm red}=\iota_{\der{x^3}}\omega=\dd x^1\wedge \dd x^2$, and the algebra of observables reduces to that on $\FR^2$ \cite{Ritter:2013wpa}.

The prequantum line bundle gerbe $\CI_\rho$ with curving
$\rho=-\tfrac{2\pi\,\di}{3!}\, \epsilon_{ijk} \, x^i\, \dd x^j\wedge
\dd x^k\in \di\,\Omega^2(M)$ reduces to a prequantum line bundle $I_r$
with connection $r=-\tfrac{2\pi\,\di}{2}\, \epsilon_{ij3}\, x^i\, \dd
x^j$ of curvature $\dd r= -2\pi\,\di\, \omega_{\rm red}$. We saw that
the sections of $\CI_\rho$ are (trivial) hermitean vector bundles over
$\FR^3$ with connection. The latter are given in terms of global
$\au(n)$-valued 1-forms $a$. Restricting to forms which are constant
along $x^3$ and contracting these with $\der{x^3}$ yields sections of
$I_a\otimes (\FR^2\times\au(n))$. In order to arrive at functions valued in $\au(1)$, we can prescribe an additional trace over the sections in the dimensional reduction.

The action of the symmetry group can be likewise reduced. In the diagram \eqref{eq:symmetry_bdl_gerbe}, the trivial line bundle $I_r$ now corresponds to a function from $\Omega(\sSpin(2)\ltimes \FR^2)$ to $\sU(1)$. The action of the string 2-group $\widehat{ \Omega_1\sSpin(2)}\xrightarrow{~\theta~} \CP\sSpin(2)$ therefore collapses to an action of $\CP\sSpin(2)$ onto $\CP(\sSpin(2)\ltimes \FR^2)$ and $\Omega(\sSpin(2)\ltimes \FR^2)$, which trivially reduces to the usual action of $\sSpin(2)$ on the base manifold $\FR^2$.

Things become a little more elegant on loop space (or more precisely knot space), the actual domain of boundaries of open M2-branes. The 2-plectic form $\omega$ is transgressed to 
\begin{equation}
 \CT(\omega)=\int_{S^1}\, \dd \tau~\frac{1}{2}\, \epsilon_{ijk}\,
 \xd^{k}(\tau) \, \dd x^{i}(\tau) \wedge \dd x^{j}(\tau) ~,
\end{equation}
where $x^i(\tau)$ are canonical local coordinates on $\CL M$ and the dot denotes a $\tau$-derivative. Note that $\CT(\omega)$ is reparameterisation invariant and thus indeed descends to knot space $\CK M$. As discussed in Section~\ref{sect:KS_prequantisation}, one can consider prequantisation on $\CL M$. The reduction from M-theory to string theory arises by compactifying the fibres of the trivial bundle $\FR^3\rightarrow \FR^2$ to $\FR^2\times S^1\rightarrow \FR^2$ and subsequently forcing the boundaries of M2-branes to be contained in one of these fibres; that is, all loops in $\CL M$ are aligned in the $x^3$-direction,
\begin{equation}
 x^i(\tau)=x^{i}(0)+\delta^{i3}\, \tau \ ,
\end{equation}
and we denote the resulting loop space by $\CL M_{\rm red}\cong
\FR^2\times S^1$. If we restrict ourselves to the zero modes $x^i=
x^i(0)$ along the fibres, then functions on $\CL M_{\rm red}$ turn into functions on $\FR^2$. Under this reduction, the transgressed 2-plectic form $\CT(\omega)$ reduces to
\begin{equation}
 \CT(\omega)\big|_{\CL M_{\rm red}}=\tfrac{1}{2}\, \epsilon_{ij3}\, \dd x^i\wedge \dd x^j= \omega_{\rm red}~,
\end{equation}
the usual symplectic form on $\FR^2$. The prequantum line bundle over $\CL M$, restricted to $\CL M_{\rm red}$ by pullback along the embedding, is trivial and its sections over $\CL M_{\rm red}$ become functions on $\FR^2$.

\section{Outlook: The non-torsion case}\label{sec:outlook}

Our approach, as it stands, requires considerable technical extension for prequantum bundle gerbes with non-torsion Dixmier-Douady class. While a full analysis of this case is beyond the scope of this paper, we present here some options for generalising our constructions in a suitable fashion.

If we had allowed bundle gerbe morphisms to have separable Hilbert
spaces as their fibre, then we could try to extend the notion of
morphisms of bundle gerbes to situations where the difference of the
Dixmier-Douady classes of source and target bundle gerbes is
non-torsion. In particular, we could have considered the natural
bundle gerbes over $S^3$ and $T^3$. Then a candidate for an inner product with the correct linearity behaviour would be the subset of 2-morphisms whose morphism of vector bundles consists of fibrewise Hilbert-Schmidt operators, as indicated in Remark~\ref{rmk:2Hspaces_and_Hilbert-Schmidt_ops}.

If we still assume 2-morphisms to be constructed from parallel morphisms of vector bundles, then kernels still exist and a morphism of bundle gerbes would still decompose into the eigenspaces of each object in $\langle (E,\alpha), (E,\alpha) \rangle$.
But since the inner product is chosen to consist of sections made from Hilbert-Schmidt operators, every such eigenbundle is of {\em finite rank} because $L^2(\CH) \subset \CK(\CH)$.
Since kernels are again morphisms between the source and target bundle gerbe of the original morphism, this would imply that there exists a finite-rank morphism between them.
This is in contradiction with the assumption that the difference of their Dixmier-Douady classes is non-torsion, due to the argument of Remark~\ref{rem:ddtorsion}.

Therefore the only morphisms between bundle gerbes, the difference of whose Dixmier-Douady classes is non-torsion, would have zero inner product with themselves; that is, there would exist only zero-norm objects in the 2-Hilbert space $\Gamma(M,\CL)$ whenever $\CL$ has non-torsion Dixmier-Douady class.

A possible way out of this problem would be to drop the parallel requirement on the 2-morphisms in $\CatLBGrb^\nabla(M)$.
We could then still define the inner product $\langle -,- \rangle$ on $\Gamma(M,\CL)$ as before by just replacing parallel sections with general sections.
This precisely amounts to taking the fibrewise Hilbert-Schmidt operators on the descent $\sfR(\Theta E \otimes F) \cong (\sfR\, \Theta E)^* \otimes \sfR F$. 
The resulting vector space $\Gamma(M,\sfR(\Theta E \otimes F)) \cong  L^2(\sfR\, \Theta E, \sfR F)$ is the space of sections of a Hilbert bundle, which is in particular a (possibly infinite-rank) hermitean vector bundle with connection.
Denoting the hermitean metric induced on $(\sfR\, \Theta E)^* \otimes \sfR F$ by $h_{(\sfR\, \Theta E)^* \otimes \sfR F}$, the vector space $L^2(\sfR\, \Theta E, \sfR F)$ is itself endowed with an inner product given by
\begin{equation}
	\prec \eps_1, \eps_2\succ\, = \int_M\, \dd \mu_M(x)~ h_{(\sfR\, \Theta E)^* \otimes \sfR F}(\eps_1(x),\eps_2(x))~.
\end{equation}
We have also amended the definition of 2-morphisms and the inner product accordingly so that there still exists a natural isomorphism
\begin{equation}
	\eta: \langle -,- \rangle \twoisom \big(2\shom_{\CatLBGrb^\nabla(M)} \big)_{|\shom_{\CatLBGrb^\nabla(M)}(\CI_0,\CL)}~.
\end{equation}
Although these conventions would allow us to circumvent the problems occuring whenever $\mathrm{dd}(\CL_2) - \mathrm{dd}(\CL_1)$ is non-torsion, this comes at the price of losing the semisimple abelian property of the morphism categories since non-parallel morphisms of vector bundles do not generally have kernels.
From the point of view of higher geometric quantisation this choice might nevertheless be reasonable, since the prequantum Hilbert space of geometric quantisation consists of {\em all} smooth sections of the prequantum line bundle $L$ independently of their compatibility with the connection on $L$.
These sections however would not lie in the domain of the
transgression functor, which is at the heart of our analysis in
Section~\ref{sect:Transgression}; they can either be interpreted as
not being related to the ordinary quantisation of the transgressed
system, or alternatively as representing physical properties of the
original system on $M$ which cannot be detected in geometric
quantisation on loop space $\CL M$.

One can easily imagine further generalisations of the whole framework we presented in this paper. Two particular aspects come to mind. Firstly, it has been suggested that ordinary manifolds are not the appropriate domain for higher quantisation, and one should switch instead to 2-spaces \cite{Ritter:2015ffa}. This would solve a number of minor issues, such as the restriction of the observable Lie 2-algebra to Hamiltonian 1-forms. 

Secondly, one might want want to work with less restrictive 2-vector spaces. From the representation theory of 2-groups, it is known that Kapranov-Voevodsky 2-vector spaces are problematic, see \cite{Barrett:2004zb,Elgueta:2007:53-92,Freidel:2012:0-0}. One could e.g. admit a larger variety of 2-bases, allowing for algebroids or more general structures. By Definition~\ref{def:2-vector_bundle}, this would lead to a generalised notion of line bundle gerbe coming with a more general set of sections. 

Finally, a more radical generalisation such as working with the sheaves of spectra of~\cite{Bunke:1311.3188} might be necessary. In the latter framework, the discussion of twisted differential K-theory seems to be less restrictive.

\subsection*{Acknowledgements}

We thank Danny Stevenson and Konrad Waldorf for helpful discussions and correspondence. Part of this work was done while R.J.S.\ was
visiting the Centro de Matem\'atica, Computa\c{c}\~{a}o e
Cogni\c{c}\~{a}o of the Universidade de Federal do ABC in S\~ao Paulo,
  Brazil during June--July 2016, whom he warmly thanks for support and hospitality during his stay there.
This work was supported in part by the Action MP1405 QSPACE from 
the European Cooperation in Science and Technology (COST). 
The work of S.B. was supported by a James Watt Scholarship. 
The work of C.S. and R.J.S.\ is supported in part by the Consolidated Grant ST/L000334/1 
from the UK Science and Technology Facilities Council (STFC). The work of R.J.S. was supported in part by
the Visiting Researcher Program
Grant 2016/04341-5 from the Funda\c{c}\~{a}o de Amparo \'a Pesquisa do
Estado de S\~ao Paulo (FAPESP, Brazil).

\appendices

\subsection{Special morphisms of line bundle gerbes}
\label{app:special_morphisms}

Here we collect some technical constructions and statements necessary in several places in the main text. Lemmas~\ref{st:t_mu-lemma} and \ref{st:dd-def_and_properties} are relevant for the definition of identity 2-morphisms.

For a surjective submersion $\sigma^Y:Y\thra M$, denote by $\Delta: Y \rightarrow Y{\times_Y}Y\subset Y^{[2]},\, y \mapsto (y,y)$ the canonical diffeomorphism of $Y$ into the diagonal, which induces further maps $\Delta_{112}:Y^{[2]}\rightarrow Y\times_YY\times_MY\subset Y^{[3]}$ and $\Delta_{122}:Y^{[2]}\rightarrow Y\times_MY\times_YY\subset Y^{[3]}$.
\begin{lemma}
	\label{st:t_mu-lemma}
	Let $\CL=(L,\mu,\sigma^Y)$ be a line bundle gerbe over $M$. There is a canonical isomorphism of line bundles with connection $t_\mu:\Delta^*L\rightarrow Y\times \FC$ given by\footnote{For ease of notation, here we do not indicate the pullback of $L$ along the inclusion $Y\times_Y Y\subset Y^{[2]}$.}
	\begin{equation}
	 \Delta^*L=\Delta^*L\otimes \Delta^*L\otimes\Delta^* L^*\xrightarrow{~\Delta^*\mu\otimes \unit~}\Delta^*L\otimes \Delta^*L^*=Y\times \FC~,
	\end{equation}
	with 
	\begin{equation}\label{eq:prop1_t}
	t_{\mu,y_1} \otimes \mathbbm{1} = (\Delta_{112}^* \mu)_{(y_1,y_1,y_2)}\eand
	\mathbbm{1} \otimes t_{\mu,y_2} = (\Delta_{122}^* \mu)_{(y_1,y_2,y_2)}~.
	\end{equation}
	Moreover, $t_\mu^{-{\rm t}} = t_{\mu^{-{\rm t}}}$.
\end{lemma}

\begin{proof}
  The definition of the isomorphism and the proof of \eqref{eq:prop1_t} is found in \cite{Waldorf:2007aa}. The remaining property is a direct consequence of the fact that $t_\mu\otimes t_{\mu^{-{\rm t}}}$ is the identity morphism on $\Delta^*L\otimes \Delta^* L^*=Y\times \FC$.
\end{proof}

Another special isomorphism of vector bundles is obtained from every 1-morphism of line bundle gerbes as follows.
\begin{lemma}
	\label{st:dd-def_and_properties}
	Let $(E,\alpha,g,\nabla^E,\zeta^Z) \in \shom_{\CatLBGrb^\nabla(M)}(\CL_1,\CL_2)$. There is an isomorphism of hermitean vector bundles with $\dd_{(E,\alpha),(z_1,z_2)}: E_{z_1}\rightarrow E_{z_2}$ over $(z_1,z_2)\in Z\times_{Y_{12}} Z$ defined by 
		\begin{equation}
		\label{eq:def_dd}
\dd_{(E,\alpha)} = \big( \pr_{Y_1}^*t_{\mu_1} \otimes \mathbbm{1}_E \big) \circ \big(\alpha_{\vert Z{\times_{Y_{12}}}Z}\big)^{-1} \circ \big( \mathbbm{1}_E \otimes \pr_{Y_2}^*t_{\mu_2} \big)^{-1} 
		\end{equation} 
	with the following properties:	
	\begin{enumerate}
		\item Over $(z_1,z_2,z_3)\in
                  Z{\times_{Y_{12}}}Z{\times_{Y_{12}}}Z$, the
                  isomorphism $\dd_{(E,\alpha)}  $ satisfies a cocycle relation
		\begin{equation}
		\label{eq:compat_dd}
		\dd_{(E,\alpha),(z_1,z_2)}^{-1} \circ \dd_{(E,\alpha),(z_2,z_3)}^{-1}
		=\dd_{(E,\alpha),(z_1,z_3)}^{-1}~.
		\end{equation}
		\item Over $(z_1,z_2,z_3,z_4)\in(Z\times_{Y_{12}} Z)\times_M(Z\times_{Y_{12}} Z)$ there is a commutative diagram
		\begin{equation}
		\label{eq:comp_alpha_dd}
		\begin{tikzpicture}[baseline=(current  bounding  box.center)]
		\node(p13L1E) at(0,0) {$L_{1,(z_1,z_3)}\otimes E_{z_3}$};
		\node(p13EL2) at(7,0) {$E_{z_1}\otimes L_{2,(z_1,z_3)}$};
		\node(p24L1E) at(0,-4) {$L_{1,(z_2,z_4)}\otimes E_{z_4}$};
		\node(p24EL2) at(7,-4) {$E_{z_2}\otimes L_{2,(z_2,z_4)}$};
		
		\draw[->] (p13L1E)--(p13EL2) node[pos=.5,above] {\scriptsize{ $\alpha_{(z_1,z_3)}$}};
		\draw[->] (p13EL2)--(p24EL2) node[pos=.5,right] {\scriptsize{ $\dd_{(E,\alpha),(z_1,z_2)} \otimes \mathbbm{1}$}};
		\draw[->] (p13L1E)--(p24L1E) node[pos=.5,left] {\scriptsize{ $\mathbbm{1} \otimes \dd_{(E,\alpha),(z_3,z_4)}$}};
		\draw[->] (p24L1E)--(p24EL2) node[pos=.5,below] {\scriptsize{ $\alpha_{(z_2,z_4)}$}};
		\end{tikzpicture}
		\end{equation}
		\item The isomorphism $\dd_{(E,\alpha)}$ satisfies the identity
		\begin{equation}
		\dd_{(E,\alpha)}^{-{\rm t}} = \dd_{(E^*,\alpha^{-{\rm t}})} = \dd_{\Theta(E,\alpha)} ~.
		\end{equation}
	\end{enumerate}
\end{lemma}

\begin{proof}
	The proofs of items (1) and (2) can again be found in~\cite{Waldorf:2007aa}.
	Item (3) follows immediately from Lemma~\ref{st:t_mu-lemma}.
\end{proof}

The isomorphism $\dd_{(E,\alpha)}$ is compatible with 2-morphisms.
\begin{lemma}
	\label{st:dd_and_other_2-mors}
	For every 2-morphism $(\phi,\omega^W): (E,\alpha,\zeta^Z) \Rightarrow (E',\alpha',\zeta^{Z'})$ the diagram 
	\begin{equation}
	\xymatrixrowsep{2cm}
	\xymatrixcolsep{3cm}
	\myxymatrix{
		E_{w_1}\ar@{->}[r]^-{\dd_{(E,\alpha),(w_1,w_2)}} \ar@{->}[d]_-{\phi_{w_1}}
		& E_{w_2}\ar@{->}[d]^-{\phi_{w_2}}\\
		E'_{w_1}\ar@{->}[r]_-{\dd_{(E',\alpha'\,),(w_1,w_2)}}
		& E'_{w_2}
	}
	\end{equation}
commutes over $(w_1,w_2)\in W {\times_{Y_{12}}} W$.
\end{lemma}

\begin{proof}
	We compute
	\begin{equation}
	\begin{aligned}
	\phi_{w_2}&\circ \dd_{(E,\alpha),(w_1,w_2)}\\
	&\qquad= \phi_{w_2} \circ \big( \pr_{Y_1}^*t_{\mu_1} \otimes \mathbbm{1}_E \big)_{w_2} \circ \big(\alpha_{\vert Z{\times_{Y_{12}}}Z}\big)^{-1}_{(w_1,w_2)} \circ \big( \mathbbm{1}_E \otimes \pr_{Y_2}^*t_{\mu_2} \big)^{-1}_{w_1}\\[4pt]
	&\qquad=\big( \pr_{Y_1}^*t_{\mu_1} \otimes \mathbbm{1}_{E'} \big)_{w_2} \circ (\unit\otimes\phi_{w_2}) \circ \big(\alpha_{\vert Z{\times_{Y_{12}}}Z}\big)^{-1}_{(w_1,w_2)} \circ \big( \mathbbm{1}_E \otimes \pr_{Y_2}^*t_{\mu_2} \big)^{-1}_{w_1}\\[4pt]
	&\qquad=\big( \pr_{Y_1}^*t_{\mu_1} \otimes \mathbbm{1}_{E'} \big)_{w_2} \circ \big(\alpha'_{\vert Z{\times_{Y_{12}}}Z}\big)^{-1}_{(w_1,w_2)} \circ (\phi_{w_1}\otimes\unit) \circ \big( \mathbbm{1}_E \otimes \pr_{Y_2}^*t_{\mu_2} \big)^{-1}_{w_1}\\[4pt]
	&\qquad=\dd_{(E',\alpha'\,),(w_1,w_2)}\circ \phi_{w_1}~,
	\end{aligned}
	\end{equation}
	where we used Lemmas~\ref{st:t_mu-lemma} and \ref{st:dd-def_and_properties}.
\end{proof}

We can use this compatibility to reduce the representatives of
2-morphisms of Definition~\ref{def:2hom_bgrb} to 2-morphisms defined
with respect to the \emph{smallest} surjective submersion $\hat
Z=Z{\times_{Y_{12}}} Z'$.\footnote{It is smallest because of the universal property of a pullback.} 
\begin{proposition}
	\label{st:2-mor_prop}
	Every 2-morphism $(\phi,\omega^W)$ from $(E,\alpha,\zeta^Z)$ to $(E',\alpha',\zeta^{Z'})$ has a representative of the form $(\phi_0,\omega^{\hat{Z}})$.
\end{proposition}

\begin{proof} 
  The cocycle relation \eqref{eq:compat_dd} satisfied by $\dd_{(E,\alpha)}$ on $Z\times_ZZ\times_Z Z$, together with the fact that $\dd_{(E,\alpha)}^{-1}$ is an isomorphism, implies
  \begin{equation}
  \big(\dd_{(E,\alpha)}^{-1} \big)_{\vert Z{\times_Z}Z} = \mathbbm{1}_E~.
  \end{equation}
Next let us restrict $\phi$ to $W{\times_{\hat{Z}}}W\subset W^{[2]}=W{\times_{M}}W$. On this space $\omega^{W[2]}$ maps to the diagonal $\hat{Z}{\times_{\hat{Z}}}\hat{Z}$, from where the projections $\fpmap{}{\pr}{[2]}{Z}$ and $\fpmap{}{\pr}{[2]}{Z'}$ map to the diagonals $Z{\times_Z}Z$ and $Z'{\times_{Z'}}Z'$, respectively.
	
Because of~\eqref{eq:compat_dd} and~\eqref{eq:comp_alpha_dd}, both quadruples $(\fpmap{}{\omega}{W*}{Z}E,\, \fpmap{}{\omega}{W*}{Z}g,\, \nabla^{\fpmap{}{\omega}{W*}{Z}E},\, \fpmap{}{\omega}{W[2]*}{Z} \dd_{(E,\alpha)})$ and $ (\fpmap{}{\omega}{W*}{Z'}E',\, \fpmap{}{\omega}{W*}{Z'}g',\, \nabla^{\fpmap{}{\omega}{W*}{Z'}E'},\, \fpmap{}{\omega}{W[2]*}{Z'} \dd_{(E',\alpha'\,)})$ are objects of $\CatDes(\CatHVBdl^\nabla,\omega^W)$.
	In fact, they lie in the \emph{image} of the inverse descent functor $D^{-1}_{\omega^W}$,
	\begin{equation}
	\begin{aligned}
	\big( \fpmap{}{\omega}{W*}{Z}E,\, \fpmap{}{\omega}{W*}{Z}g,\, \nabla^{\fpmap{}{\omega}{W*}{Z}E},\, \fpmap{}{\omega}{W[2]*}{Z} \dd_{(E,\alpha)} \big) &= D_{\omega^W}^{-1} \big( \fpmap{}{\pr}{*}{Z}E,\, \fpmap{}{\pr}{*}{Z}g,\, \nabla^{\fpmap{}{\pr}{*}{Z}E}\big) \ , \\[4pt]
	\big( \fpmap{}{\omega}{W*}{Z'}E',\, \fpmap{}{\omega}{W*}{Z'}g,\, \nabla^{\fpmap{}{\omega}{W*}{Z'}E},\, \fpmap{}{\omega}{W[2]*}{Z'} \dd_{(E',\alpha'\,)} \big) &= D_{\omega^W}^{-1} \big( \fpmap{}{\pr}{*}{Z'}E',\, \fpmap{}{\pr}{*}{Z'}g',\, \nabla^{\fpmap{}{\pr}{*}{Z'}E'} \big)~.
	\end{aligned}
	\end{equation}
	Moreover, $\dd_{(E,\alpha)}$ is an isometric and
        connection-preserving morphism by its construction as a composition of morphisms with these properties.
	Lemma~\ref{st:dd_and_other_2-mors} now shows that 
	\begin{equation}
	 \phi \ \in \ \shom_{\CatDes(\CatHVBdl^\nabla,\omega^W)}\big( D_{\omega^W}^{-1}(\fpmap{}{\pr}{*}{Z}E),D_{\omega^W}^{-1}(\fpmap{}{\pr}{*}{Z'}E'\,) \big)
	\end{equation}
	is a descent morphism. Combining this with the fact that these objects are in the actual (and not just essential) image of the inverse descent functor, and that $\CatHVBdl^\nabla$ is a stack, we deduce that there exists $\phi_0 \in \shom_{\CatHVBdl^\nabla(\hat{Z})}(\fpmap{}{\pr}{*}{Z}E,\fpmap{}{\pr}{*}{Z'}E'\,)$ which the inverse descent functor $D_\omega^{-1}$ maps to $\phi$, that is, $\omega^{W*}\phi_0 = \phi$.
	
Finally, we have to ensure that $\phi_0$ satisfies the compatibility relation~\eqref{eq:comp_2hom_bgrb} for 2-morphisms in $\CatLBGrb^\nabla(M)$. From $\omega^{W*}\phi_0 = \phi$ we have
	\begin{equation}
	\big( \phi_{0,w_1} \otimes \mathbbm{1} \big) \circ \alpha_{(w_1,w_2)}
	= \alpha'_{(w_1,w_2)}\circ \big(\unit\otimes \phi_{0,w_2} \big)
	\end{equation}
	over $W^{[2]}$, and since $\omega^W$ is a surjective submersion it follows that \eqref{eq:comp_2hom_bgrb} is satisfied over $\hat Z^{[2]}$.
	By the definition of the equivalence relation on pairs representing 2-morphisms, we can choose $X = W$ in \eqref{eq:2morequiv} which shows that $(\phi,\omega^W)$ and $(\phi_0,\omega^{\hat{Z}})$ with $\omega^{\hat{Z}}=\id_{\hat Z}$ are equivalent.
\end{proof}

\begin{remark}
If we had defined the 2-morphisms with $W = \hat{Z}$ fixed, we would have had to use descent theory in the construction of the horizontal composition.
The more general Definition~\ref{def:2hom_bgrb} of 2-morphisms trades this complication for the introduction of redundancy in equivalence classes.
However, Lemma~\ref{st:2-mor_prop} implies that these two definitions
yield precisely the same 2-morphisms. In particular, all morphisms of
twisted vector bundles, which are precisely the bundle gerbe modules of local bundle gerbes, can be defined over \emph{any} given common refinement of the respective open covers.
\end{remark}

\setcounter{equation}{0}
\setcounter{theorem}{0}

\subsection{Proof of Theorem~\ref{st:Direct_sum_as_fctr}}
\label{app:Proof_of_direct_sum_thm}

(1):
	First we have to check that $(E,\alpha) \oplus (E',\alpha'\,)$ is an object in $\shom_{\CatLBGrb^\nabla(M)}(\CL_1,\CL_2)$, which amounts to satisfying~\eqref{eq:comp_hom_bgrb} or equivalently
	\begin{equation}\label{eq:B_1}
	 \beta_{(\hat z_1,\hat z_3)}\circ (\mu_{1,(\hat z_1,\hat z_2,\hat z_3)}\otimes\unit) =(\unit\otimes \mu_{2,(\hat z_1,\hat z_2,\hat z_3)})\circ(\beta_{(\hat z_1,\hat z_2)}\otimes \unit)\circ(\unit\otimes \beta_{(\hat z_2,\hat z_3)})~.
	\end{equation}
The distribution isomorphisms $\sfd^{{\rm l},{\rm r}}_L$ are compatible with pullbacks and satisfy
	\begin{equation}
	\sfd^{\rm l}_{L\otimes L'} = \sfd^{\rm l}_L \circ (\unit\otimes \sfd^{\rm l}_{L'}) \eand \sfd^{\rm r}_{L\otimes L'} = \sfd^{\rm r}_{L'} \circ (\sfd^{\rm r}_L\otimes\unit) ~.
	\end{equation}
	For $\mu \in \shom_{\CatHLBdl^\nabla(M)}(L\otimes L',L''\,)$, multilinearity yields
	\begin{equation}
	\begin{aligned}
	\sfd^{\rm l}_{L''} \circ (\mu \otimes \mathbbm{1}) &= \big((\mu \otimes \mathbbm{1}) \oplus (\mu \otimes \mathbbm{1})\big) \circ\sfd^{\rm l}_{L\otimes L'}~,\\[4pt]
	\sfd^{\rm r}_{L''} \circ (\mathbbm{1} \otimes \mu) &= \big((\mathbbm{1} \otimes \mu) \oplus (\mathbbm{1} \otimes \mu)\big) \circ\sfd^{\rm r}_{L\otimes L'}~.
	\end{aligned}
	\end{equation}
	These morphisms are naturally morphisms in $\CatHVBdl^\nabla(M)$, that is, they preserve metrics and connections on the respective bundles.
	Finally, for $\alpha \in \shom_{\CatHVBdl^\nabla(M)}(E,E'\,)$ and $ \beta \in \shom_{\CatHVBdl^\nabla(M)}(F,F'\,)$ we have
	\begin{equation}
	\begin{aligned}
	\sfd^{\rm l}_L \circ \big( \mathbbm{1} \otimes ( \alpha \oplus \beta)\big) &= \big( (\mathbbm{1} \otimes \alpha) \oplus (\mathbbm{1} \otimes \beta) \big) \circ \sfd^{\rm l}_L~,\\[4pt]
	\sfd^{\rm r}_L \circ \big(( \alpha \oplus \beta) \otimes \mathbbm{1} \big) &= \big( (\alpha \otimes \mathbbm{1}) \oplus (\beta \otimes \mathbbm{1}) \big) \circ \sfd^{\rm r}_L~.
	\end{aligned}
	\end{equation}
	By repeatedly applying these identities we then readily verify \eqref{eq:B_1} as 
	\begin{align}
	&\sfd^{\rm r}_{L_{2,(\hat z_1,\hat z_3)}} \circ \beta_{(\hat z_1,\hat z_3)}\circ (\mu_{1,(\hat z_1,\hat z_2,\hat z_3)}\otimes \unit)\circ \big(\sfd^{\rm l}_{L_{1,(\hat z_1,\hat z_2)}\otimes L_{1,(\hat z_2,\hat z_3)}}\big)^{-1}\notag \\
	&= \big(\alpha_{(\hat z_1,\hat z_3)} \oplus \alpha'_{(\hat z_1,\hat z_3)} \big) \circ \sfd^{\rm l}_{L_{1,(\hat z_1,\hat z_3)}}\circ (\mu_{1,(\hat z_1,\hat z_2,\hat z_3)}\otimes \unit)\circ \big(\sfd^{\rm l}_{L_{1,(\hat z_1,\hat z_2)}\otimes L_{1,(\hat z_2,\hat z_3)}}\big)^{-1} \notag \\[4pt]
	&= \big(\alpha_{(\hat z_1,\hat z_3)} \oplus \alpha'_{(\hat z_1,\hat z_3)} \big) \circ \big( (\mu_{1,(\hat z_1,\hat z_2,\hat z_3)}\otimes \unit)\oplus (\mu_{1,(\hat z_1,\hat z_2,\hat z_3)}\otimes \unit) \big)\notag \\[4pt]
	&= \big( (\unit\otimes \mu_{2,(\hat z_1,\hat z_2,\hat z_3)})\oplus (\unit\otimes \mu_{2,(\hat z_1,\hat z_2,\hat z_3)}) \big)\circ \big((\alpha_{(\hat z_1,\hat z_2)}\otimes \unit)\oplus(\alpha'_{(\hat z_1,\hat z_2)}\otimes \unit)\big) \notag \\
	&\hspace{8.5cm}\circ \big((\unit\otimes\alpha_{(\hat z_2,\hat z_3)})\oplus(\unit\otimes\alpha'_{(\hat z_2,\hat z_3)})\big)\notag \\[4pt]
	&=\sfd^{\rm r}_{L_{2,(\hat z_1,\hat z_3)}} \circ (\unit\otimes\mu_{2,(\hat z_1,\hat z_2,\hat z_3)})\circ(\beta_{(\hat z_1,\hat z_2)}\otimes\unit)\circ(\unit\otimes\beta_{(\hat z_2,\hat z_3)}) \notag\\
	&\hspace{8.5cm}\circ \big(\sfd^{\rm l}_{L_{1,(\hat z_1,\hat z_2)}\otimes L_{1,(\hat z_2,\hat z_3)}}\big)^{-1}~.
	\end{align}
	
	Next we turn to the sum of 2-morphisms. Let 
	\begin{equation}
	\begin{aligned}
	(\phi,\id_{\hat{Z}}) & \ \in \ 2\shom_{\CatLBGrb^\nabla(M)}\big((E,\alpha,g_E,\nabla^E,B^E,\zeta^Z), (E',\alpha',g_{E'},\nabla^{E'},B^{E'},\zeta^{Z'})\big)~,\\[4pt]
	(\psi,\id_{\hat{X}}) & \ \in \ 2\shom_{\CatLBGrb^\nabla(M)}\big((F,\beta,g_F,\nabla^F,B^F,\zeta^X),(F',\beta',g_{F'},\nabla^{F'},B^{F'},\zeta^{X'})\big)
	\end{aligned}
	\end{equation}
	be two 2-morphisms in $\CatLBGrb^\nabla(M)$.
	By Proposition~\ref{st:2-mor_prop} these choices cover all possible 2-morphisms between these 1-morphisms.
	We define the direct sum of these 2-morphisms over $W =
        \hat{Z}{\times_{Y_{12}}}\hat{X} =
        Z{\times_{Y_{12}}}Z'{\times_{Y_{12}}}X{\times_{Y_{12}}}X'$ by
        setting $(\phi,\id_{\hat{Z}}) \oplus (\psi,\id_{\hat{X}}):= \big( \chi, \id_{\hat{Z}{\times_{Y_{12}}}\hat{X}} \big)$ with 
	\begin{equation}
	\chi_{(\hat z_1,\hat z_2,\hat x_1,\hat x_2)} = \phi_{(\hat z_1,\hat z_2)}\oplus \psi_{(\hat x_1,\hat x_2)}: E_{\hat z_1} \oplus F_{\hat x_1} \longrightarrow E'_{\hat z_2} \oplus F'_{\hat x_2}~.
	\end{equation}
	It satisfies the compatibility condition~\eqref{eq:comp_2hom_bgrb} for $2$-morphisms as one readily verifies.
	This definition of the sum of 2-morphisms involves only tensor products, direct sums and composition of morphisms, so that it is strictly associative.
	All the morphisms and operations involved in this construction are compatible with connections, whence the direct sum of $2$-morphisms is $2\shom_{\CatLBGrb^\nabla(M)}$-valued as required.
	Moreover, any other representative of a $2$-morphism of this form is a pullback of a representative of the above form along a surjective submersion (e.g.\ $W \rightarrow \hat{Z}$) according to Proposition~\ref{st:2-mor_prop}.
	All the elements involved in the definition of the direct sum of $2$-morphisms as above are compatible with pullbacks, whence we can straightforwardly generalise this definition to generic representatives of $2$-morphisms.
	
	Finally, for the identity 2-morphisms we consider the following commutative diagram over $\hat{Z}{\times_{Y_{12}}}\hat{Z}$:
	\begin{equation}
	\xymatrixcolsep{3pc}
	\myxymatrix{
	E_{\hat z_1}\oplus F_{\hat z_1}\ar@{->}[d]\ar@{->}[r]^{\id_{E_{\hat z_1}\oplus F_{\hat z_1}}} & 
	E_{\hat z_1}\oplus F_{\hat z_1}\ar@{->}[d] \\
	(E_{\hat z_1}\oplus F_{\hat z_1})\otimes L_{2,(\hat z_1,\hat z_2)}\ar@{->}[d]\ar@{->}[r] &
	(E_{\hat z_1}\otimes L_{2,(\hat z_1,\hat z_2)})\oplus (F_{\hat z_1}\otimes L_{2,(\hat z_1,\hat z_2)})\ar@{->}[d] \\
 	L_{1,(\hat z_1,\hat z_2)}\otimes(E_{\hat z_2}\oplus F_{\hat z_2})\ar@{->}[d]\ar@{->}[r] &
 	(L_{1,(\hat z_1,\hat z_2)}\otimes E_{\hat z_2})\oplus (L_{1,(\hat z_1,\hat z_2)}\otimes F_{\hat z_2})\ar@{->}[d] \\
 	E_{\hat z_2}\oplus F_{\hat z_2}\ar@{->}[r]_{\id_{E_{\hat z_1}\oplus F_{\hat z_1}}} &
 	E_{\hat z_2}\oplus F_{\hat z_2} 
 	}
	\end{equation}
	The inner two horizontal morphisms are distribution isomorphisms for the pullbacks of $L_i$ or the trivial bundle of rank~$1$, respectively.
	In the top two vertical morphisms we have the action of $t_{\mu_2}^{-1}$, while in the bottom two vertical morphisms that of $t_{\mu_1}$ appears.
	In order to obtain the respective $\dd_{(E,\alpha)}$ as the vertical compositions, the middle two vertical morphisms are given by the respective sums.
	By definition the middle square is commutative; the outer ones are commutative by the properties of the distribution morphisms spelled out above. This yields 
	\begin{equation}
	\dd_{(E,\alpha)\oplus(F,\beta)} = \dd_{(E,\alpha)}\oplus\dd_{(F,\beta)}~,
	\end{equation}
	which completes the proof of functoriality.

\noindent
(2):
	Let us now address the distributivity of the tensor product over the direct sum.
	For this, we introduce a third morphism of bundle gerbes $(F,\beta,g_F,\nabla^F,\zeta^U):\CL_3 \rightarrow \CL_4$. We need to find a 2-isomorphism $(\sfd^{\rm l}_{F,E,E'},\omega^W)$ between 
	\begin{equation}
		(F,\beta,\zeta^U) \otimes \big( (E,\alpha,\zeta^Z) \oplus (E',\alpha',\zeta^{Z'}) \big)
	\end{equation}
	and 
	\begin{equation}
	 \big( (F,\beta,\zeta^U) \otimes (E,\alpha,\zeta^Z) \big) \oplus \big( (F,\beta,\zeta^U) \otimes (E',\alpha',\zeta^{Z'}) \big)~, 
	\end{equation}
	which are given by vector bundles over $U {\times_M} (Z
        {\times_{Y_{12}}} Z'\, )$ and $( U {\times_M} Z)
        {\times_{Y_{1234}}} ( U {\times_M} Z' \,)$,
        respectively. The simplest choice of surjective submersion is
        given by $\omega^W=\id_W$ for
	\begin{equation}
	\begin{aligned}
	W &= \big( U {\times_M} (Z {\times_{Y_{12}}} Z'\, ) \big)
        {\times_{Y_{1234}}} \big( ( U {\times_M} Z)
        {\times_{Y_{1234}}} ( U {\times_M} Z' \, ) \big)\\[4pt]
	&\cong \big( U {\times_{Y_{34}}}U {\times_{Y_{34}}} U \big)
        {\times_M} \big( ( Z {\times_{Y_{12}}} Z) {\times_{Y_{12}}} (
        Z' {\times_{Y_{12}}} Z' \, ) \big) \ ,
	\end{aligned}
	\end{equation}
	which by Proposition~\ref{st:2-mor_prop} can be made without
        loss of generality. We now define 
	\begin{equation}
		\sfd^{\rm l}_{F,E,E'} := \big( ( \dd_{(F,\beta)}
                \otimes \dd_{(E,\alpha)}) \oplus ( \dd_{(F,\beta)}
                \otimes \dd_{(E',\alpha'\,)} ) \big) \circ \sfd^{\rm l}_{F}~,
	\end{equation}
	with the evident pullbacks suppressed. We need to verify that this isomorphism of vector bundles satisfies~\eqref{eq:comp_2hom_bgrb}, making it a 2-isomorphism. We readily compute 
	\begin{equation}
	 \begin{aligned}
	  \sfd^{\rm l}_{F,E,E'}\circ &\big(\beta\otimes(\alpha\oplus
          \alpha'\, )\big)\\ 
&=\big( ( \dd_{(F,\beta)} \otimes
          \dd_{(E,\alpha)}) \oplus ( \dd_{(F,\beta)} \otimes
          \dd_{(E',\alpha'\, )}) \big) \circ \sfd^{\rm l}_{F}\circ
          \big(\beta\otimes(\alpha\oplus \alpha'\, )\big)\\[4pt]
	  &=\big( ( \dd_{(F,\beta)} \otimes \dd_{(E,\alpha)}) \oplus (
          \dd_{(F,\beta)} \otimes \dd_{(E',\alpha'\, )}) \big) \circ
          \big((\beta\otimes\alpha)\oplus(\beta\otimes\alpha'\, )\big)\circ \sfd^{\rm l}_{F}\\[4pt]
	  &=\big((\beta\otimes\alpha)\oplus(\beta\otimes\alpha'\,
          )\big)\circ \big( ( \dd_{(F,\beta)} \otimes
          \dd_{(E,\alpha)}) \oplus \big( \dd_{(F,\beta)} \otimes
          \dd_{(E',\alpha'\, )}) \big) \circ \sfd^{\rm l}_{F}\\[4pt]
	  &=\big((\beta\otimes\alpha)\oplus(\beta\otimes\alpha'\, )\big)\circ \sfd^{\rm l}_{F,E,E'}~,
	 \end{aligned}
	\end{equation}
	again with all pullbacks suppressed. Analogously we define
        $\sfd^{\rm r}_{E,E',F}$.

It remains to show that $\sfd^{\rm l}_{F,E,E'}$ is a natural transformation. Consider additional morphisms $(G,\eta,\zeta^X),\,  (G',\eta',\zeta^{X'}): \CL_1 \rightarrow \CL_2$ and $(H,\rho,\zeta^V): \CL_3 \rightarrow \CL_4$ of line bundle gerbes.
	Let $(\phi, \id_{U {\times_{Y_{34}}} V}): (F,\beta,\zeta^U)
        \rightarrow (H,\rho,\zeta^V)$, $(\psi,\id_{Z {\times_{Y_{12}}}
          X}): (E,\alpha,\zeta^Z) \rightarrow (G,\eta,\zeta^X)$, and $(\psi',\id_{Z' {\times_{Y_{12}}} X'}): (E',\alpha',\zeta^{Z'}) \rightarrow (G',\eta',\zeta^{X'})$ be 2-morphisms in $\CatLBGrb^\nabla(M)$. We then need to show that
	\begin{equation}
	 \sfd^{\rm l}_{H,G,G'}\circ \big(\phi\otimes(\psi\oplus
         \psi'\, )\big)=\big((\phi\otimes \psi)\oplus (\phi\otimes
         \psi'\, )\big)\circ \sfd^{\rm l}_{F,E,E'}~.
	\end{equation}
	Again the $\sfd^{\rm l}_H$ part of $\sfd^{\rm l}_{H,G,G'}$ can
        be readily pulled through $\phi\otimes(\psi\oplus \psi'\, )$, leaving us with
	\begin{equation}
	 \big( ( \dd_{(H,\beta)} \otimes \dd_{(G,\alpha)}) \oplus
         \big( \dd_{(H,\beta)} \otimes \dd_{(G',\alpha'\, )}) \big)\circ 
	 \big((\phi\otimes \psi)\oplus (\phi\otimes \psi'\, )\big)\circ \sfd^{\rm l}_F~.
	\end{equation}
	Applying Lemma~\ref{st:dd_and_other_2-mors}, we simplify this
        to $\big((\phi\otimes \psi)\oplus (\phi\otimes \psi'\,
        )\big)\circ \sfd^{\rm l}_{F,E,E'}$. This completes the proof
        of Theorem~\ref{st:Direct_sum_as_fctr}.

\setcounter{equation}{0}
\setcounter{theorem}{0}

\subsection{Transgression of forms to mapping spaces}
\label{sect:mapping_space_geometry}

Here we summarise some statements concerning the transgression of
differential forms from a manifold $M$ to mapping spaces over $M$.
Let $\Sigma$ be a smooth manifold of dimension $d$ and denote by $C^\infty(\Sigma,M) = (\Sigma{\rightarrow}M)$ its mapping space into $M$. These mapping spaces carry the structure of a Fr\'echet manifold, but we can also work in the category of diffeological spaces.

We start by investigating the tangent bundle of
$(\Sigma{\rightarrow}M)$, cf.\ e.g.\ \cite{kobayashi1989}. A tangent
vector to a map $f: \Sigma \rightarrow M$ is an infinitesimal
displacement of the point $f$ in $(\Sigma{\rightarrow}M)$; that is, it is an infinitesimal deformation of the map and therefore corresponds to an assignment of a tangent vector $X_{| f(\sigma)} \in T_{f(\sigma)}M$ to every $\sigma \in \Sigma$. Thus\footnote{If one considers based maps, the deformation of $f$ at a basepoint $\sigma_0$ has to vanish in order to preserve the based property.}
\begin{equation}
T_f (\Sigma{\rightarrow}M) = \Gamma(\Sigma,f^*TM)~.
\end{equation}
A section $X$ of the pullback bundle can be equivalently described by a
lift $\tilde{X}$ of $f$ to $TM$; that is, $\pi \circ \tilde{X} = f$ where $\pi: TM \rightarrow M$ is the projection to the base.
This is summarised in the commutative diagram
\begin{equation}
\xymatrixcolsep{1cm}
\xymatrixrowsep{1cm}
\myxymatrix{%
	f^*TM \ar@{->}[r]^-{\hat{f}} \ar@<0.5ex>[d]^-{f^*\pi} & TM \ar@{->}[d]^-{\pi}\\
	\Sigma \ar@<0.5ex>[u]^-{X} \ar@{->}[r]_-{f} \ar@{->}[ur]_-{\tilde{X}} & M
}
\end{equation}
The total tangent space is then $T(\Sigma{\rightarrow }M)=\bigcup_f\,
T_f(\Sigma{\rightarrow }M)$. An element $\tilde X\in
(\Sigma{\rightarrow }TM)$ can be regarded as a tangent vector at
$f=\pi\circ \tilde X$, which yields an isomorphism 
\begin{equation}\label{eq:isomorphism_tangents}
 T(\Sigma{\rightarrow} M)\cong (\Sigma {\rightarrow} TM)~.
\end{equation}

There is a double fibration
\begin{equation}
\xymatrixcolsep{1cm}
\xymatrixrowsep{1cm}
\myxymatrix{%
	 & \Sigma\times (\Sigma{\rightarrow} M)\ar@<0.5ex>[dl]_{\ev} \ar@<0.5ex>[dr]^-{\pr} & \\
	M & & (\Sigma{\rightarrow} M)
}
\end{equation}
where $\ev(\sigma,f):=f(\sigma)$ and $\pr(\sigma,f):=f$. The transgression of a
differential form $\omega\in \Omega^p(M)$ is now obtained by first
pulling back $\omega$ from $M$ to the correspondence space $\Sigma
\times (\Sigma{\rightarrow}M)$, and then using integration over
$\Sigma$ as a pushforward of $\omega$ to the mapping space $(\Sigma{\rightarrow}M)$.

\begin{definition}
	The \uline{transgression of a $p$-form} $\omega$ on $M$ to $(\Sigma{\rightarrow}M)$ is defined by its evaluation on tangent vectors $X_1,\ldots,X_{p-d}$ to the mapping space at a point $f$ as
	\begin{equation}\label{eq:def_transgression}
	\CT(\omega)_{|f}(X_1,\ldots,X_{p-d}) :=  \int_\Sigma\, f^*\big( \iota_{X_{p-d}} \cdots \iota_{X_1} \omega \big)~.
	\end{equation}
\end{definition}
In \eqref{eq:def_transgression} we made use of the isomorphism \eqref{eq:isomorphism_tangents}.

\begin{example}
Let $\Sigma = [0,1] = I$, $\tau$ a coordinate on the interval, and $\gamma: I \rightarrow M$ a based path in $M$. Then 
\begin{equation}
\begin{aligned}
\CT(\omega)_{|\gamma}(X_1,\ldots,X_{p-1})
&= \int_I\, \gamma^*\big( \iota_{X_{p-1}} \cdots \iota_{X_1} \omega \big)\\[4pt]
&= \int_I\, \dd\tau \ \omega_{|\gamma(\tau)} \big( X_1(\tau), \ldots, X_{p-1}(\tau), \dot{\gamma}(\tau) \big)~.
\end{aligned}
\end{equation}
This is indeed multilinear on vector fields, since a function on the mapping space is constant on every individual map.
\end{example}

Evidently, we have the following statement.
\begin{proposition}
	Transgression defines a morphism of vector spaces
	\begin{equation}
	\CT: \Omega^p(M) \longrightarrow
        \Omega^{p-d}(\Sigma{\rightarrow}M) \ .
	\end{equation}
\end{proposition}
A more involved property is compatibility with the exterior and Lie derivatives, which we will explore in the following. As done e.g.\ in \cite[Lemma 3]{Wurzbacher:1995gb}, we can use the global definition of the exterior derivative
\begin{equation}
\begin{aligned}
&\dd\, \CT(\omega)\big( X_{0}, \ldots, X_{p-d} \big)\\ & \hspace{1cm}
\qquad := \sum_{i=0}^{p-d}\, (-1)^i\, \pounds_{X_i}\, \CT (\omega) \big( X_0, \ldots, \widehat{X_i}, \ldots, X_{p-d} \big) \\
&\hspace{1.5cm} \qquad +\, \sum_{0 \leq i < j \leq p-d}\, (-1)^{i+j}\,
\CT(\omega) \big( [X_i,X_j], X_0, \ldots, \widehat{X_i}, \ldots, \widehat{X_j}, \ldots, X_{p-d} \big) \ ,
\end{aligned}
\end{equation}
where the hats indicate omitted slots, but for this we have to extend the tangent vectors to local vector fields. Our extension will be different from that of \cite[Lemma 3]{Wurzbacher:1995gb}.

There exists a natural pullback of vector fields on $M$ to vector fields on $(\Sigma{\rightarrow}M)$ given by evaluating a vector field on the image of each $f \in (\Sigma{\rightarrow}M)$.
Since sections of vector bundles form soft sheaves, this pullback is surjective onto $T_f(\Sigma{\rightarrow}M)$ whenever $f$ is an embedding with closed image. This enables us to investigate the geometry of the mapping space in terms of the geometry of $M$.

However, points $f$ in our mapping spaces are not necessarily embeddings, and the pullback construction does not even yield all tangent vectors to the mapping space at a given point. Generic vector fields on $(\Sigma{\rightarrow}M)$ do not even locally correspond to vector fields on $M$.
For example, let $t \mapsto f_t$ be a curve in
$(\Sigma{\rightarrow}M)$ through $f$ such that at some $\sigma \in \Sigma$
the family $f_t(\sigma)$ is constant, that is, $f_t(\sigma) = f(\sigma)$ for all $t$;
then any vector field on $(\Sigma{\rightarrow}M)$ whose value at $f_t(\sigma)$ is not constant does not even induce a local vector field on $M$ around $f(\sigma)$.

A way around these injectivity problems is to replace maps $f \in (\Sigma{\rightarrow}M)$ by their graphs
\begin{equation}
\bbG f: \Sigma \longrightarrow \Sigma{\times}M, \quad \bbG f(\sigma) = (\sigma,f(\sigma))~,
\end{equation}
which are always embeddings as well as sections of $\pr_\Sigma: \Sigma{\times}M \rightarrow \Sigma$; that is, $\bbG$ is a map
\begin{equation}
\bbG: (\Sigma{\rightarrow}M) \longrightarrow \Gamma(\Sigma,\Sigma{\times}M)~.
\end{equation}
Moreover, the vertical tangent bundle $T^{\rm ver}(\Sigma{\times}M)$ coincides with the pullback of
$f^*TM$ along $\pr_\Sigma$.

Consider a tangent vector $X \in T_{f}
(\Sigma{\rightarrow}M)=\Gamma(\Sigma, f^*TM)$
at a point $f$, that is, a smooth assignment $X(\sigma) \in T_{f(\sigma)}M$ to every $\sigma \in \Sigma$. 
This in turn is equivalent to a vertical vector field $\bbG X \in
\Gamma \big( \bbG f(\Sigma), T^{\rm ver}(\Sigma{\times}M) \big)$ living on the image of $\bbG f$ in $\Sigma{\times}M$. Thus generic sections of the vertical tangent bundle yield vector fields on $(\Sigma{\rightarrow} M)$,
\begin{equation}
\bbG^*: \Gamma \big( \Sigma{\times}M, T^{\rm ver}(\Sigma{\times}M) \big) \longrightarrow \Gamma \big( (\Sigma{\rightarrow}M), T(\Sigma{\rightarrow}M) \big), \quad X \longmapsto \bbG^*X~,
\end{equation}
where over each point $f\in(\Sigma{\rightarrow} M)$ the map $\bbG^*$
is the restriction to the image of $\bbG f$. Smoothness of the resulting vector field follows from smoothness of the vector field on $\Sigma{\times} M$.

If $\Sigma$ is compact, the image $\bbG f(\Sigma)$ of $\Sigma$ under
every $\bbG f$ is closed in $\Sigma{\times}M$. Given a tangent vector $X\in T_{f}(\Sigma{\rightarrow} M)$ encoded as a vertical vector field over $\bbG f(\Sigma)$, we can extend it to a global vertical vector field over $\Sigma{\times} M$ by softness of the sheaf of sections of the vertical tangent bundle and we find the desired extension of a tangent vector to a vector field.

The pullback vector fields $\bbG^*X$ have further properties related to
those of $X$ which we will find useful in our computations. We briefly
summarise them in the following three lemmas.

\begin{lemma}
	Let $X,Y \in \Gamma \big( \Sigma{\times}M, T^{\rm ver}(\Sigma{\times}M) \big)$.
	We denote the flow of the vector field $X$ by $t \mapsto
        \Phi^X_t \in {\rm Diff}(\Sigma{\times}M)$.
	\begin{enumerate}
		\item The flows of $X$ and $\bbG^*X$ are related via
		\begin{equation}
			\bbG \Phi^{\bbG^*X}_t (f) = \Phi^X_t (\bbG f) \ ,
		\end{equation}
		that is, $\big( \bbG \Phi^{\bbG^*X}_t (f) \big) (\sigma) =
                \Phi^X_t \big( \bbG f (\sigma) \big) $ for all $f \in (\Sigma{\rightarrow}M)$ and $\sigma \in \Sigma$.
		\item The flow of a pullback vector field acts on another pullback vector field as
		\begin{equation}
			\Phi^{\bbG^*X}_{t*|f}\, (\bbG^*Y_{|f}) = \bbG^* \big( \Phi^X_{t*}\, Y \big)_{|\Phi^{\bbG^*X}_t(f)}~.
		\end{equation}
		\item The map $\bbG^*$ is a homomorphism of Lie algebras
		\begin{equation}
			\big[\bbG^*X, \bbG^*Y
                        \big]_{T_f(\Sigma{\rightarrow}M)} = \big(
                        \bbG^* [X,Y]_{T^{\rm ver}(\Sigma{\times}M)} \big)_{|f}
			= \bbG^*_{|f}[X,Y]_{T^{\rm ver}(\Sigma{\times}M)}~.
		\end{equation}
	\end{enumerate}
\end{lemma}

\begin{proof}
	 (1):
	This is the definition of the flow of a vector field on
        $(\Sigma{\rightarrow}M)$: Every point $f(\sigma)$ flows along the tangent vector $X(\sigma)$ at $f(\sigma)$.
	This only requires extensions of $X(\sigma)$ for fixed $\sigma \in \Sigma$ to a local vector field on $M$.
	
	\noindent
        (2):
	Let $s \mapsto \Gamma_s$ be a curve in
        $(\Sigma{\rightarrow}M)$ through $f$ with tangent vector at $s
        = 0$ given by $\frac{\dd}{\dd s}_{|0} \Gamma_s =
        \bbG^*Y_{|f}$, or equivalently $\frac{\dd}{\dd s}_{|0}
        \bbG\Gamma_s (\sigma) = Y_{|f(\sigma)}$.
	Then we use item (1) to compute
	\begin{equation}
	\begin{aligned}
	\big( \Phi^{\bbG^*X}_{t*}\, (\bbG^*Y)_{|f} \big) (\sigma)
	&= \frac{\dd}{\dd s}_{|0} \Phi^{\bbG^*X}_{t} \circ \Gamma_s (\sigma)\\[4pt]
	&= \frac{\dd}{\dd s}_{|0} \pr_M\circ \Phi^X_{t} \circ \bbG \Gamma_s (\sigma)\\[4pt]
	&= \pr_{M*} \Phi^X_{t*} (Y_{|f(\sigma)})\\[4pt]
	&= \bbG^* \big( \Phi^X_{t*}\, Y \big)_{|\Phi^{\bbG^*X}_{t} (f)}\, (\sigma)~.
	\end{aligned}
	\end{equation}
		
	\noindent
        (3):
	This follows from differentiating the result of item (2) with respect to $t$ and then evaluating at $t=0$.
\end{proof}

\begin{lemma}
The transgression of a form $\omega\in\Omega^p(M)$ contracts with tangent vectors of the above type according to
\begin{equation}
\CT(\omega)_{|f} \big( \bbG^*X_1, \ldots, \bbG^*X_{p-d} \big)
= \int_\Sigma\, (\bbG f)^* \big( \iota_{X_{p-d}} \cdots \iota_{X_1} \,
\pr_M^* \omega \big)~,
\end{equation}
where $X_1,\dots,X_{p-d}\in \Gamma(\Sigma{\times}M,T^{\rm
  ver}(\Sigma{\times}M))$ and $\pr_M:\Sigma{\times}M\to M$ is the projection.
\label{lem:transGomega}\end{lemma}
\begin{proof}
This follows by direct computation.
\end{proof}

\begin{lemma}
	\label{st:commuting_d_and_iotas}
	Let $X_0,\ldots, X_n$ be vector fields on a manifold $M$.
	For any $\omega \in \Omega^p(M)$ with $p \geq n+1$,
	\begin{equation}
	  \begin{aligned}
	  &\iota_{X_n} \cdots \iota_{X_0}\, \dd\omega + (-1)^n\, \dd\,
          \iota_{X_n} \cdots \iota_{X_0} \omega \\
	  & \qquad\qquad = \sum_{i=0}^n\, (-1)^i\, \pounds_{X_i}\,
          \iota_{X_n} \cdots \widehat{\iota_{X_i}} \cdots \iota_{X_0}
          \omega \\
	  &\hspace{4cm}\quad + \sum_{0 \leq i < j \leq n}\, (-1)^{i+j}\, \iota_{X_n} \cdots \widehat{\iota_{X_j}} \cdots \widehat{\iota_{X_i}} \cdots \iota_{X_0}\, \iota_{[X_i,X_j]} \omega~.
	  \end{aligned}
	\end{equation}
\end{lemma}

\begin{proof}
	The proof is by induction.
	For $n = 0$ the identity collapses to the Cartan formula for the Lie derivative $\pounds_{X} = \dd\, \iota_X + \iota_X\, \dd$. Using the induction hypothesis and $[\pounds_X,\iota_Y] = \iota_{[X,Y]}$ for $X,Y \in \Gamma(M,TM)$, we readily compute
	\begin{align}
	\dd\, \iota_{X_{n+1}} \cdots \iota_{X_0} \omega
	&= \pounds_{X_{n+1}}\, \iota_{X_n} \cdots \iota_{X_0} \omega - \iota_{X_{n+1}}\, \dd\, \iota_{X_n} \cdots \iota_{X_0} \omega \nt[4pt]
	&= \pounds_{X_{n+1}}\, \iota_{X_n} \cdots \iota_{X_0} \omega \nt*
	&\quad - \sum_{i=0}^n\, (-1)^{n+i}\, \iota_{X_{n+1}}\, \pounds_{X_i}\, \iota_{X_n} \cdots \widehat{\iota_{X_i}} \cdots \iota_{X_0} \omega \nt*
	&\quad - \sum_{0 \leq i < j \leq n}\, (-1)^{n+i+j}\, \iota_{X_{n+1}} \cdots \widehat{\iota_{X_j}} \cdots \widehat{\iota_{X_i}} \cdots \iota_{X_0}\, \iota_{[X_i,X_j]} \omega \nt*
	&\quad + (-1)^{n+2}\, \iota_{X_{n+1}} \cdots \iota_{X_0}\, \dd\omega \nt[4pt]
	&= \sum_{i=0}^{n+1}\, (-1)^{n+1+i}\, \pounds_{X_i}\, \iota_{X_{n+1}} \cdots \widehat{\iota_{X_i}} \cdots \iota_{X_0} \omega \nt*
	&\quad + \sum_{0 \leq i < j \leq n+1}\, (-1)^{n+1+i+j}\, \iota_{X_{n+1}} \cdots \widehat{\iota_{X_j}} \cdots \widehat{\iota_{X_i}} \cdots \iota_{X_0}\, \iota_{[X_i,X_j]} \omega \nt*
	&\quad + (-1)^{n+2}\, \iota_{X_{n+1}} \cdots \iota_{X_0}\, \dd\omega~,
	\end{align}
	as required.
\end{proof}
In particular, we recover the coordinate-free expression for the exterior derivative $\dd\omega$ of $\omega \in \Omega^{n+1}(M)$.

We are now in a position to compute the derivative of a transgressed
form on the mapping space.

\begin{proposition}
	\label{st:d_and_transgression}
	Let $\omega \in \Omega^p(M)$ be a $p$-form on $M$, and $\Sigma$ a compact $d$-dimensional manifold.
	The derivative of the transgression of $\omega$ on the mapping space $(\Sigma{\rightarrow}M)$ reads as
	\begin{equation}
	\begin{aligned}
	\label{eq:d_and_transgression}
	&\dd \, \CT (\omega)_{|f} \big( \bbG^*X_{0}, \ldots, \bbG^*X_{p-d} \big)\\
	&= (-1)^{p-d}\, \int_{\partial \Sigma}\, (\bbG\, \partial f)^* \big( \iota_{X_{p-d}} \cdots \iota_{X_0}\, \pr_M^* \omega \big)
	+ \CT(\dd \omega )_{|f} \big( \bbG^*X_{0}, \ldots, \bbG^*X_{p-d} \big)~,
	\end{aligned}
	\end{equation}
	where $X_0,\dots,X_{p-d}\in \Gamma(\Sigma{\times}M,T^{\rm
  ver}(\Sigma{\times}M))$ and we use the notation $\partial f = f_{|\partial \Sigma}$.
\end{proposition}

\begin{proof}
	Using the coordinate-free definition of the exterior derivative and Lemma~\ref{lem:transGomega}, the proof amounts to a straightforward computation
	\begin{align}
	&\dd \, \CT\, (\omega )_{|f} \big( \bbG^*X_{0}, \ldots, \bbG^*X_{p-d} \big) \nt*[0.2cm]
	&= \sum_{i=0}^{p-d}\, (-1)^i\, \pounds_{\bbG^*X_i}\, \CT (\omega)_{|f}\big( \bbG^*X_0, \ldots, \widehat{\bbG^*X_i}, \ldots, \bbG^*X_{p-d} \big) \nt*
	&\quad + \sum_{0 \leq i < j \leq p-d}\, (-1)^{i+j}\, \CT(\omega)_{|f} \big( [\bbG^*X_i, \bbG^*X_j], \bbG^*X_0, \ldots, \widehat{\bbG^*X_i}, \ldots, \widehat{\bbG^*X_j}, \ldots, \bbG^*X_{p-d} \big) \nt[0.2cm]
	&= \sum_{i=0}^{p-d}\, (-1)^i\, \pounds_{\bbG^*X_i}\ \int_\Sigma\, (\bbG f)^* \big( \iota_{X_{p-d}} \cdots \widehat{\iota_{X_i}} \cdots \iota_{X_0}\, \pr_M^* \omega \big) \nt*
	&\quad + \sum_{0 \leq i < j \leq p-d}\, (-1)^{i+j}\ \int_\Sigma\, (\bbG f)^* \big( \iota_{X_{p-d}} \cdots \widehat{\iota_{X_j}} \cdots \widehat{\iota_{X_i}} \cdots \iota_{X_0}\, \iota_{[X_i, X_j]}\, \pr_M^* \omega \big) \nt[0.2cm]
	&= \sum_{i=0}^{p-d}\, (-1)^i\, \frac{\dd}{\dd t}_{|0}\ \int_\Sigma\,  \Phi^{\bbG^*X_i}_t (f)^* \big( \iota_{X_{p-d}} \cdots \widehat{\iota_{X_i}} \cdots \iota_{X_0}\, \pr_M^*\omega \big) \nt*
	&\quad + \sum_{0 \leq i < j \leq p-d}\, (-1)^{i+j}\ \int_\Sigma\, (\bbG f)^* \big( \iota_{X_{p-d}} \cdots \widehat{\iota_{X_j}} \cdots \widehat{\iota_{X_i}} \cdots \iota_{X_0}\, \iota_{[X_i, X_j]}\, \pr_M^*\omega \big) \nt[0.2cm]
	&= \sum_{i=0}^{p-d}\, (-1)^i\ \int_\Sigma\, (\bbG f)^* \big( \pounds_{X_i}\, \iota_{X_{p-d}} \cdots \widehat{\iota_{X_i}} \cdots \iota_{X_0}\, \pr_M^* \omega \big) \nt*
	&\quad + \sum_{0 \leq i < j \leq p-d}\, (-1)^{i+j}\ \int_\Sigma\, (\bbG f)^* \big( \iota_{X_{p-d}} \cdots \widehat{\iota_{X_j}} \cdots \widehat{\iota_{X_i}} \cdots \iota_{X_0}\, \iota_{[X_i, X_j]}\, \pr_M^* \omega \big) \nt[0.2cm]
	&= (-1)^{p-d}\, \int_\Sigma\, (\bbG f)^* \big( \dd\, \iota_{X_{p-d}} \cdots \iota_{X_0}\, \pr_M^* \omega \big)
	+ \int_\Sigma\, (\bbG f)^* \big( \iota_{X_{p-d}} \cdots \iota_{X_0}\, \pr_M^*\, \dd \omega \big) \nt[0.2cm]
	&= (-1)^{p-d}\, \int_{\partial \Sigma}\, (\bbG\,\partial f)^* \big( \iota_{X_{p-d}} \cdots \iota_{X_0}\, \pr_M^* \omega \big)
	+ \int_\Sigma\, (\bbG f)^* \big( \iota_{X_{p-d}} \cdots \iota_{X_0}\, \pr_M^*\, \dd \omega \big)~.
	\end{align}
	Comparing this with Lemma~\ref{lem:transGomega} yields the required identity.
\end{proof}

\begin{remark}
The first term on the right-hand side of
\eqref{eq:d_and_transgression} can be regarded as a transgression to
the mapping space $(\dpar \Sigma{\rightarrow} M)$; in the following we
denote it by $\CT_{\partial}$. With this notation \eqref{eq:d_and_transgression} reads
\beq
\dd\circ\CT=\CT\circ\dd+(-1)^{{\rm deg}-d}\, \partial^*\circ\CT_\partial \ ,
\eeq
where $\partial:(\Sigma{\rightarrow} M)\to (\dpar \Sigma{\rightarrow} M)$ is the map $f\mapsto f_{|\partial\Sigma}$.
\end{remark}

\begin{corollary}
	If $\Sigma$ is closed, then transgression is a chain map inducing a pushforward map of degree $-d$ on the de Rham complexes and therefore on de Rham cohomology,
	\begin{equation}
	\CT: H^p_{\rm dR}(M) \longrightarrow H^{p-d}_{\rm dR}(\Sigma{\rightarrow}M)~,\quad [\omega] \longmapsto [\CT(\omega)]~.
	\end{equation}
\end{corollary}

We can now compute the Lie derivative of a transgressed form from the Cartan formula $\pounds_X = \dd\, \iota_X + \iota_X\, \dd$.
For a vertical vector field $\pr_M^*X \in \Gamma(\Sigma{\times}M, T^{\rm ver}(\Sigma{\times}M))$ which is constant along the $\Sigma$-fibres, that is, which is the pullback along $\pr_M$ of a vector field $X$ on $M$, we have
\begin{equation}
\begin{aligned}
\iota_{\bbG^*(\pr_M^*X)}\, \CT(\omega)_{|f}& \big( \bbG^*X_1, \ldots, \bbG^*X_{p-d-1} \big)\\
&\qquad\qquad\qquad\qquad = \CT(\omega)_{|f} \big( \bbG^*(\pr_M^*X), \bbG^*X_1, \ldots, \bbG^*X_{p-d-1} \big)\\[4pt]
&\qquad\qquad\qquad\qquad = \int_\Sigma\, (\bbG f)^* \big( \iota_{X_{p-d-1}} \cdots \iota_{X_1}\, \iota_{\pr_M^*X}\, \pr_M^*\, \omega \Big)\\[4pt]
&\qquad\qquad\qquad\qquad = \CT(\iota_X \omega)_{|f} \big( \bbG^*X_1, \ldots, \bbG^*X_{p-d-1} \big)~.
\end{aligned}
\end{equation}
Therefore by using Proposition~\ref{st:d_and_transgression} we find
\begin{equation}
\dd\, \iota_{\bbG^*(\pr_M^*X)}\, \CT(\omega ) = (-1)^{p-1-d}\, \CT_{\partial}\, (\iota_X \omega) + \CT (\dd\, \iota_X\omega)~.
\end{equation}
On the other hand, one has
\begin{equation}
\begin{aligned}
\iota_{\bbG^*(\pr_M^*X)}\, \dd \, \CT(\omega)
&= \iota_{\bbG^*(\pr_M^*X)}\, \big( (-1)^{p-d}\, \CT_{\partial}(\omega) + \CT (\dd \omega) \big)\\[4pt]
&= (-1)^{p-d}\, \CT_{\partial}\, (\iota_X \omega) + \CT (\iota_X\, \dd \omega)~.
\end{aligned}
\end{equation}
Adding both contributions, we arrive at
\begin{equation}
\label{eq:Lie-derivative_on_mapping_space}
\pounds_{\bbG^*(\pr_M^*X)} \, \CT(\omega) = \CT(\pounds_X \omega)~.
\end{equation}

It is tempting to ask for an extension of this identity to more general vector fields on the mapping space $(\Sigma{\rightarrow}M)$.
However, the term on the right-hand side of~\eqref{eq:Lie-derivative_on_mapping_space} is sensible if and only if $X$ is a globally well-defined vector field on $M$.
From this we infer that such a generalisation is not possible.
This is a consequence of the problem mentioned before that vector fields on the mapping space do not generally correspond even locally to vector fields on the original manifold $M$. Nevertheless, the identity~\eqref{eq:Lie-derivative_on_mapping_space} is sufficient for all our purposes in the present paper.

\bibliographystyle{latexeu}


\end{document}